\documentclass{article}
\usepackage{a4wide}
\usepackage{amsmath}
\usepackage{amssymb}
\usepackage{amsthm}
\usepackage{marginnote}

\newtheorem{thm}{Theorem}
\newtheorem{prop}{Proposition}

\newtheorem{lemma}{Lemma}

\newtheorem*{prop*}{Proposition}

\newcommand{\bbr}{{\mathbb R}}
\newcommand{\bbs}{{\mathbb S}}
\newcommand{\bbn}{{\mathbb N}}
\newcommand{\p}{\hat{p}}
\newcommand{\q}{\hat{q}}

\def\be{\begin{equation}}
\def\ee{\end{equation}}

\newcommand{\ts}{\tilde{s}}
\newcommand{\tg}{\tilde{g}}

\newcommand{\tnabla}{\tilde{\nabla}}%
\newcommand{\tp}{\tilde{p}}
\newcommand{\tq}{\tilde{q}}
\newcommand{\h}{\tilde{h}}%
\newcommand{\tM}{\widetilde{M}}

\newcommand{\tC}{\widetilde{C}}
\newcommand{\tG}{\widetilde{G}}
\newcommand{\tomega}{\tilde{\omega}}
\newcommand{\tT}{\widetilde{T}}

\newcommand{\bea}{\begin{eqnarray}}

\newcommand{\eea}{\end{eqnarray}}
\newcommand{\bean}{\begin{eqnarray*}}

\newcommand{\eean}{\end{eqnarray*}}

\newcounter{mnotecount}[section]

\renewcommand{\themnotecount}{\thesection.\arabic{mnotecount}}

\newcommand{\mnote}[1]
{\protect{\stepcounter{mnotecount}}$^{\mbox{\footnotesize
$
\bullet$\themnotecount}}$ \marginpar{
\raggedright\tiny\em
$\!\!\!\!\!\!\,\bullet$\themnotecount: #1} }

\usepackage{authblk}

\begin{document}
\title{Well-posedness of anisotropic and homogeneous solutions to the Einstein-Boltzmann system with a conformal gauge singularity}

\author[1]{Ho Lee\footnote{holee@khu.ac.kr}}
\author[2]{Ernesto Nungesser\footnote{em.nungesser@upm.es}}
\author[3]{John Stalker\footnote{stalker@maths.tcd.ie}}
\author[4]{Paul Tod\footnote{tod@maths.ox.ac.uk}}

\affil[1]{Department of Mathematics and Research Institute for Basic Science, College of Sciences, Kyung Hee University, Seoul 02447, Republic of Korea}
\affil[2]{M2ASAI, Universidad Polit{\'e}cnica de Madrid, ETSI Navales, Avda. de la Memoria
4, 28040 Madrid, Spain}
\affil[3]{School of Mathematics, Trinity College, Dublin 2, Ireland and Hamilton Mathematics Institute, Dublin 2, Ireland}
\affil[4]{Mathematical Institute, University of Oxford, Oxford OX2 6GG}

\maketitle

\begin{abstract}
We consider the Einstein-Boltzmann system for massless particles in the Bianchi I space-time with scattering cross-sections in a certain range of soft potentials. We assume that the space-time has an initial conformal gauge singularity and show that the initial value problem is well posed with data given at the singularity. This is understood by considering conformally rescaled equations. The Einstein equations become a system of singular ordinary differential equations, for which we establish an existence theorem which requires several differentiability and eigenvalue conditions on the coefficient functions together with the Fuchsian conditions. The Boltzmann equation is regularized by a suitable choice of time coordinate, but still has singularities in momentum variables. This is resolved by considering singular weights, and the existence is obtained by exploiting singular moment estimates. 
\end{abstract}


%
\section{Introduction}

Since the celebrated theorems of Hawking and Penrose, the importance of singularities and their study in general relativity and in particular in cosmology is well-known. These theorems give little information about the structure of space-time singularities, which is expected to be very complicated in general.  However, Penrose \cite{WCH1} has presented arguments that the singularities relevant to the Big Bang in cosmology are much simpler than they are allowed mathematically to be. His `Weyl Curvature Hypothesis' is that at any initial singularity the Weyl curvature is finite or zero, even while the Ricci curvature is singular. Singularities with this character have variously been called `isotropic' \cite{AT99a,GW} or `conformally compactifiable' but following \cite{LT} it seems more natural now to call them `conformal gauge singularities'. By this we mean that while the space-time with the physical metric is singular, there is another choice of metric in the conformal class of the physical metric -- equivalently, a conformal rescaling of the physical metric by a conformal factor -- for which the curvature is nonsingular. From this point of view, the space-time singularity, which is certainly present in the physical universe, is attributable wholly to the vanishing or singularity of the conformal factor. From the standpoint of conformal geometry there is no singularity, while the physical metric as a particular (but preferred) representative of the conformal class none-the-less is singular. In this sense the singularity, while certainly present in the physical metric, may be thought of as due to the choice of conformal gauge, whence the name.  For recent progress in the study of the formation and properties of general cosmological singularities we refer to \cite{ABIO, FRS, Ring1,Ring2}. 

There has been other work on extending a space-time through a singularity -- see e.g.\ \cite{GL, Sb} -- but the emphasis there is on the possibility of $C^0$ extensions of the physical space-time rather than on smooth extensions of the conformal space-time. There has also been other work on giving data at the big-bang -- see e.g.\ \cite{Ring2} and other works by this author -- but again the emphasis is different, with no assumption of smoothness of the conformal metric. One might also mention the pre-big-bang cosmology of \cite{Ven}, which is very much in the context of string theory rather than conventional general relativity.

In this article we study the well-posedness of solutions to the Einstein-Boltzmann system with data given at a conformal gauge singularity. There have been several interesting recent results concerning massless solutions to the Einstein-Vlasov system \cite{BFJST,GT,HL,JTV}. The massless Vlasov case is not only interesting as a limiting case, but of interest by itself mathematically and physically when modelling massless particles like photons or neutrinos, but also in the context of Burnett’s conjecture \cite{Burnett} that under appropriate assumptions the limit of highly oscillatory solutions to the Einstein vacuum equations is a solution of the Einstein-massless Vlasov system, see \cite{GT,HL} for recent progress. For an introduction and background to the field of general relativistic kinetic theory we refer to \cite{AGS,Andreasson, Ehlers, Ring, SZ, Stewart}. 

Concerning the massless Einstein-Boltzmann system, however, there are almost no results whatsoever. There are heuristic indications \cite{T03} that there should exist a well-posed initial value problem with data at the singularity for this system, at least for some choices of cross-section, as there are for the Einstein-Vlasov system \cite{AT99,A3} and the Einstein-Euler system \cite{AT99a}. The matter is taken to be massless because any positive mass would be negligible compared to kinetic energy near the initial singularity. In \cite{T03,T07} power series solutions in the time-variable were considered and no difference between massive and massless Vlasov near the singularity was found.  We allow for a non-negative cosmological constant to make contact with Penrose's Conformal Cyclic Cosmology \cite{CCC}, albeit the constant does not play an important role in our mathematical analysis.

In a previous paper \cite{LNT1} three of the present authors considered the homogeneous and isotropic case, so that the metric takes the Friedman-Lema\^itre-Roberston-Walker (hereafter FLRW) form. In that case the FLRW scale factor is fixed by the assumption of massless particles and choice of cosmological constant, and it remains only to solve the massless Boltzmann equation in the fixed FLRW background. This is still nontrivial as there are two sources of singularities in the collision integral for the Boltzmann equation, one from the space-time volume form going to zero at the initial singularity, and one intrinsic to the massless case where there are poles at small momenta. These can both be resolved by a choice of suitably restricted, but still physically plausible cross-sections, and by a redefinition of the time-coordinate.

The next problem, treated in this paper, is to consider the Einstein-Boltzmann system with the space-time restricted to a spatially homogeneous but anisotropic cosmology. Now we need to solve the Einstein equations, which becomes a nonlinear system of ODEs, coupled with the Boltzmann equation, which is an integro-differential equation and even in the homogeneous case has partial derivatives with respect to both time and momenta, and with data for both given at an initial `big bang' singularity.

Our main result is the well posedness of solutions to the coupled Einstein-Boltzmann system with Bianchi I symmetry having an initial conformal gauge singularity with data given at the singularity. We consider again the family of cross-sections which allowed the proof of existence in the isotropic case \cite{LNT1}. The aim of this article is twofold. First, we generalize the paper \cite{AT99} from the Vlasov to the Boltzmann case. For the Vlasov equation classical solutions can easily be obtained by considering characteristic equations. The singularities caused by the massless particles can also be removed by simply assuming that the distribution function is compactly supported in $ \bbr^3 \setminus \{ 0 \} $. However, for the Boltzmann equation classical solutions are only available for rather restricted cases, and furthermore they are not suitable for the massless case. Instead, we consider $ L^1 $ solutions. There are many well established $ L^1 $ theories for the Boltzmann equation, and we can make use of them up to a certain point, since they are all available only for the massive case. We need to extend the theory to the massless case in order that the singularities at $ p = 0 $ can be controlled. The idea for this is to use singular weights. Then, we exploit `singular' moment estimates, which will lead to the existence result. It is interesting to note that the singularities at $ p = 0 $ are controlled by singular weights. Second, we generalize the paper \cite{LNT1} from the FLRW to the Bianchi I case. It was first shown in \cite{LNT1} that the singular moment estimates can be used to obtain the existence in the massless case, but the isotropy assumption on $ f $ was crucial, and this only applies to the FLRW cosmology. In order to extend to the Bianchi I case we introduce another type of singular weight in $ L^\infty $. This is basically the inverse of the J{\" u}ttner distribution, but multiplied by a singular factor, which should be introduced to work together properly with the former type of singular weight in $ L^1 $. For more details we refer to Section \ref{sNotations}.

Once one is able to deal with the Boltzmann equation the main argument is similar to \cite{AT99} making suitable time coordinate changes in order to have at most first order singularities. In \cite{AT99} a theorem established by Rendall and Schmidt \cite{RS} is used. Here we use a similar argument, which requires an improved version of the Rendall-Schmidt theorem. This is our Theorem \ref{ODEthm} which we prove in an appendix. This theorem is stronger in three ways.

Firstly, the theorem of Rendall and Schmidt applies to equations which are singular, but linear. Ours applies to certain non-linear equations, including those which we derive below from the Einstein equations. The equations Anguige and Tod derive in \cite{AT99}, and to which they need to apply Rendall and Schmidt's result, are also non-linear. In their case they are able to avoid this difficulty by a trick which unfortunately does not work in our case. The generalisation to non-linear equations is, as one would expect, non-trivial. Like most existence and uniqueness results for ordinary differential equations it involves converting the system of differential equations to an equivalent integral equation but for our theorem this integral equation involves multiple integrals.

Secondly, the dependence of the size and domain of definition of the solution on the various functions appearing in the equation is made fully explicit in our theorem. We need this in order to make our iterative argument work. No new ideas are needed in order to make the dependence explicit, just careful bookkeeping.

Thirdly, we prove considerably more regularity than Rendall and Schmidt. While their solutions are merely bounded as the time coordinate tends to zero ours are continuously differentiable. When applied to the physical problem under consideration this allows us to find two additional terms in the asymptotic expansion of the metric beyond the leading term. In order to get this improvement we need stronger assumptions on the eigenvalues of the coefficient matrix, but fortunately those assumptions are satisfied in our case. With the original assumptions of Rendall and Schmidt one can only prove continuity as time tends to zero and can only get one term in the asymptotics beyond the leading one.

We would like to point out that to the best of our knowledge there are no previous results whatsoever concerning massless Einstein-Boltzmann solutions which go beyond the special relativistic case except \cite{LNT1}. Neither in the seminal papers concerning the local existence of solutions of Daniel Bancel and Yvonne Choquet-Bruhat  \cite{Bancel,BC}, nor in the global existence results of Norbert Noutchegueme, David Dongo and Etienne Takou \cite{NDT,NT} and others, is the massless case considered. In some cases the extension might not be too difficult, but if one considers reasonable scattering cross-sections the extension is not trivial. For instance some recent results concerning global existence and asymptotic behaviour by two of the present authors \cite{LN,LN1} where massive particles were considered do not cover the whole family of cross-sections considered in this article. 

The paper is structured as follows. In Section \ref{EB}, we derive the main equations and state the main results. We first consider the Einstein-Boltzmann system for massless particles. The family of scattering cross-sections \eqref{cross} which will be considered in this paper will be introduced. In Section \ref{sCR}, we assume that the space-time has an initial conformal gauge singularity, and investigate the transformation properties of the Einstein-Boltzmann system under conformal rescaling. In Section \ref{sSHC}, we assume the spatial homogeneity to derive the equations with the standard conformal time coordinate $\tau$, and in Section \ref{sMeqns} we 
change the time coordinate to $s$ to derive the main equations, which we will call the rescaled Einstein-Boltzmann system in this paper. In Section \ref{mainresults}, we state the main results. We first collect in Section \ref{sNotations} the notations that will be used throughout the paper. The main results of the paper will be given in Section \ref{maintheorems}. The main theorem of the paper is Theorem \ref{phystheorem}. The theorem shows that the massless Einstein-Boltzmann system with an initial conformal gauge singularity has a unique local solution with data at the singularity. We will consider the rescaled Einstein-Boltzmann system to prove Theorem \ref{phystheorem}. We notice that the initial data given in Theorem \ref{phystheorem} is a single distribution function $ f_0 $. It will be shown in Theorem \ref{fuchsian} that the initial data $ a_0 $, $ b_0 $, $ K_0 $ and $ \hat{ Z }_0 $ for the rescaled Einstein equations are uniquely determined by $ f_0 $. With these data we prove in Theorem \ref{conformalprop} that the rescaled Einstein-Boltzmann system has a unique local solution. The proof of Theorem \ref{conformalprop} will be given separately in Sections \ref{sEinstein}, \ref{sBoltzmann} and \ref{sMain1}. In Section \ref{sEinstein}, we consider the Einstein equations for a given distribution function, and the existence is obtained in Proposition \ref{propeinstein}. The proof of Proposition \ref{propeinstein} will be given in Section \ref{sproofE}, where Theorem \ref{ODEthm} from the appendix will be applied extensively. In Section \ref{sBoltzmann}, we consider the Boltzmann equation for a given metric, and the existence is obtained in Proposition \ref{prop B}. The existence will be given in $ L^1_{ - 1 } $, but it will not be enough to be applied to the coupled Einstein-Boltzmann case. In Lemma \ref{lem dfdp}, we improve the existence in $ L^1_{ - 2 - \delta / 2 } $ for some small $ \delta > 0 $, and this will lead to the continuous dependence on the metric, in Section \ref{ContiBoltzmann}, which is necessary for an iteration scheme to work. In Section \ref{sMain1}, we combine Propositions \ref{propeinstein}, \ref{prop B} and \ref{prop C} to complete the proof of Theorem \ref{conformalprop}. Finally, we prove Theorem \ref{phystheorem} in Section \ref{sMain2}.

\section{The Einstein-Boltzmann system with massless particles} \label{EB}
Let us consider the Einstein-Boltzmann system in a Lorentz manifold $(\tM,\tg)$ for massless particles where we use signature $(-+++)$. In massless kinetic theory, the matter content of space-time is described by a distribution function $ f $ on the cotangent bundle, supported on the null-cone bundle. Call this bundle $N$, with fibre $N_x$ the null-cone at $x\in\tM$ and local coordinates $ ( x^\alpha , p_i ) $. Here and throughout the paper, Greek indices run from $ 0 $ to $ 3 $, while Latin indices run from $ 1 $ to $ 3 $. Introduce $p=(p_1,p_2,p_3)$ for the spatial part of momentum $\tp_\alpha=(\tp_0,p)$, where $\tp^{0}$ is determined by the null condition
\begin{align}
\tg^{\alpha\beta}\tp_{\alpha}\tp_{\beta}=0,\label{masslesst}
\end{align}
and the requirement that~$\tp^{\alpha}$~be future directed. The (null) geodesic spray~$\mathcal{L}_{\tg}$~is the vector field on the cotangent bundle tangent to $N$ and defined, in local coordinates, by
\begin{align}
\mathcal{L}_{\tg}=\tg^{\alpha\beta}\tp_{\alpha}\frac{\partial}{\partial
x^{\beta}}-\frac{1}{2}\tp_{\alpha}\tp_{\beta}\frac{\partial \tg^{\alpha\beta}}{\partial
x^{\gamma}}\frac{\partial}{\partial \tp_{\gamma}}.\label{Liouville}
\end{align}
The cotangent space to~$\tM$ at a point $x$ is a flat Lorentz manifold, and on
the fibres of the submanifold $N$ there exists an invariant volume measure $\tomega_{\tp}$ given by
\begin{align}
\tomega_{\tp}=\frac{1}{\tp^{0}\sqrt{- \det \tg}}\, d^{3}p\label{f1},
\end{align}
where $ \det \tg $ is the determinant of the $ 4 \times 4 $ matrix $ \tg_{ \alpha \beta } $ and $d^3p=dp_1dp_2dp_3$.

The distribution of particles and momenta is described by a non-negative scalar function $ f = f ( x^\alpha , p_i ) $ on $ N $, which satisfies the Boltzmann equation
\begin{equation}
\mathcal{L}_{\tg}(f)=\tC(f,f),\label{V1}
\end{equation}
where the collision term $\tC(f,f)$ will be given below. The stress-energy-momentum tensor due to these particles is given by
\begin{equation}
\tT_{\alpha\beta}(x)=\int_{N_x}f\tp_{\alpha}\tp_{\beta}\,\tomega_{\tp},
\label{T1}
\end{equation}
and if the Boltzmann equation is satisfied then
\begin{equation}
\tilde{\nabla}^{\alpha}\tT_{\alpha\beta}=0,\label{T2}
\end{equation}
where $\tilde{\nabla}$ is the metric covariant derivative for $\tg$. The coupled Einstein-Boltzmann equations, for the metric $\tg_{\alpha\beta}$ and the particle distribution function $f$, are therefore
\begin{align}\label{bolt1}
\tG_{\alpha\beta}+\Lambda \tg_{\alpha\beta} &=8\pi\int_{N_x}f\tp_{\alpha}\tp_{\beta}\,\tomega_{\tp},\\
\mathcal{L}_{\tg}(f)&=\tC(f,f),\label{bolt2}
\end{align}
taking $G=c=1$.

Turning to $\tC(f,f)$, we consider only binary collisions, so a pair of incoming particles with null momenta $\tp_\alpha$ and $\tq_\alpha$ collide to produce a pair of outgoing particles with null momenta $\tp'_\alpha$ and $\tq'_\alpha$, where the momentum is conserved so that
\begin{align}\label{conservationtilde}
\tp_\alpha+\tq_\alpha=\tp'_\alpha+\tq'_\alpha.
\end{align}
Here, the momenta $\tp_\alpha$ and $\tq_\alpha$ are called the pre-collision momenta, and $\tp'_\alpha$ and $\tq'_\alpha$ are called the post-collision momenta. There are various possible parametrisations of $\tp'_\alpha$ and $\tq'_\alpha$ in terms of $\tp_\alpha$, $\tq_\alpha$ and a 3-vector $\omega_i$, but we shall retain the parametrisation used in \cite{LNT1}. Explicit representations will be given in \eqref{p'^0}--\eqref{q'}. Let $ d \omega $ denote the area-form on the unit 2-sphere of $ \omega_i $. Then, the collision term in the Boltzmann equation is
\begin{align}\label{collisionterm}
\tC(f,f)=\int_{\mathbb{R}^3}\int_{\mathbb{S}^2}\h\sqrt{\ts}\tilde{\sigma}(\h,\Theta)(f(\tp')f(\tq')-f(p)f(q)) \, d\omega\, \tomega_{\tq},
\end{align}
where $\h$ and $\ts$ are the relative momentum and the square of the energy in the center of momentum system, respectively, defined by 
\begin{align}
\h=\sqrt{\tg^{\alpha\beta}(\tp_\alpha-\tq_\alpha)(\tp_\beta-\tq_\beta)},\qquad \ts=-\tg^{\alpha\beta}(\tp_\alpha+\tq_\alpha)(\tp_\beta+\tq_\beta),\label{momentumscalt}
\end{align}
and $\tilde{\sigma}$ is the scattering cross-section, and the dependence of $ f $ on $ x^\alpha $ is to be understood. In this paper we will consider massless particles, so that the assumption \eqref{masslesst} implies
\begin{align}\label{htilde0mass}
\h^2=\ts=-2\tg^{\alpha\beta}\tp_\alpha \tq_\beta. 
\end{align}
Moreover, we will consider the scattering cross-sections in a certain range of soft potentials. The following family of cross-sections from \cite{LNT1} (setting $C_1=1$) will be considered:
\begin{align}\label{cross}
\tilde{\sigma}(\h,\Theta)= \h^{- \gamma } \qquad ( 1 < \gamma < 2 ) , 
\end{align}
which does not depend on the scattering angle $ \Theta $. Now, the collision term for massless particles with the cross-sections in \eqref{cross} reduces to
\begin{align}\label{collisiontermt}
\tC(f,f)=\int_{\mathbb{R}^3}\int_{\mathbb{S}^2}\h^{ 2 - \gamma } (f(\tp')f(\tq')-f(p)f(q)) \, d\omega\, \tomega_{\tq}.
\end{align}
For more details about the relativistic collision operator we refer to \cite{CK, LNT1, SY}. We also refer to \cite{CIP} for basic information about the Boltzmann equation.

In this paper, we are interested in the Einstein-Boltzmann system \eqref{bolt1}--\eqref{bolt2}, \eqref{collisiontermt} for massless particles with the cross-sections in \eqref{cross}. Now, we will assume that the space-time has an initial conformal gauge singularity, and will show that the Einstein-Boltzmann system has a solution with initial data given at the singularity. This will be studied by considering conformally rescaled equations and applying the assumption of spatial homogeneity. The conformal transformation will be investigated in Section \ref{sCR}, and the assumption of spatial homogeneity will be introduced in Section \ref{sSHC}. In Section \ref{sMeqns}, we derive our main equations. The main results will be given in Section \ref{mainresults}.

\subsection{Conformally rescaling the Einstein-Boltzmann system}\label{sCR}
In this part, we investigate the transformation properties of the Einstein-Boltzmann system under conformal rescaling. Suppose that the physical space-time is $(\tM,\tg)$, while the rescaled, unphysical space-time is $(M,g)$ with
\begin{align}
\tg_{\alpha\beta}=\Omega^2g_{\alpha\beta}.\label{g1}
\end{align}
We will assume that the space-time has an initial conformal gauge singularity \cite{LT}, so that the function $\Omega$ vanishes on a smooth, space-like hypersurface $\Sigma$ in the unphysical space-time $M$. Our aim is to formulate the equations in $ M $ with data at $ \Sigma $.

Let us first consider the Boltzmann equation \eqref{bolt2} with \eqref{collisiontermt}. As part of the definition of a conformal gauge singularity, we shall assume that the distribution function $ f $ extends to a smooth function on $T^*\Sigma$. We also assume that $ f $ does not change under rescaling. Now, we choose
\begin{align}\label{pscaling}
\tilde{p}_\alpha =p_\alpha,
\end{align}
so that the canonical one-form is unchanged
\begin{align}
\theta = \tilde{p}_\alpha dx^\alpha = p_\alpha dx^\alpha.
\end{align}
Then, the null condition \eqref{masslesst} is preserved
\begin{align}
g^{\alpha\beta}p_\alpha p_\beta=\Omega^2\tg^{\alpha\beta}\tp_\alpha \tp_\beta=0,\label{massless}
\end{align}
while the geodesic spray transforms as
\begin{align}
\mathcal{L}_{\tg}f=\Omega^{-2}\mathcal{L}_gf.
\end{align}
For the volume-form on $N$ we obtain
\begin{align}
\omega_p = \frac{1}{p^{0}\sqrt{- \det g}}\, d^{3}p = \Omega^2\tomega_{\tp},
\end{align}
where $p^0=g^{00}p_0$ rather than $\tg^{00}p_0$, and $ \det g $ denotes the determinant of the $ 4 \times 4 $ matrix $ g_{ \alpha \beta } $. This implies that we can define a rescaled energy-momentum tensor as
\begin{align}
T_{\alpha\beta} = \Omega^2 \tT_{\alpha\beta} = \int_{N_x}fp_{\alpha}p_{\beta}\,\omega_p,
\end{align}
which will still be divergence-free, since in fact
\begin{align}
\tg^{\alpha\beta}\tnabla_\alpha\tT_{\beta\gamma}=\Omega^{-4}g^{\alpha\beta}\nabla_\alpha T_{\beta\gamma}=0,
\end{align}
where $\nabla$ is the metric covariant derivative for $g$. For the relative momentum $ \h $, we have
\begin{align}
h = \sqrt{ g^{ \alpha \beta } ( p_\alpha - q_\alpha ) ( p_\beta - q_\beta ) } = \Omega \h,
\end{align}
so that the Boltzmann equation \eqref{bolt2} with \eqref{collisiontermt} transforms as
\begin{align}\label{V2}
\mathcal{L}_{g}f=C(f,f):=\Omega^{ \gamma -2}\int_{\mathbb{R}^3}\int_{\mathbb{S}^2} h^{2- \gamma }(f(p')f(q')-f(p)f(q)) \, d\omega\, \omega_{q},
\end{align}
taking this to be the definition of $C(f,f)$. Note that the rescaled collision term $ C ( f , f ) $ has the singular factor $ \Omega^{\gamma - 2 } $, so that we need to redefine the time-coordinate as in \cite{LNT1}.

\subsubsection{Spatial homogeneity}\label{sSHC}
In this paper we will assume that the space-time and the distribution function are both spatially homogeneous. One introduces a basis of left-invariant one-forms $\sigma^i$ satisfying
\begin{align}
d\sigma^i=\frac{1}{2}C^i_{\;jk}\sigma^j\wedge\sigma^k, \label{m2}
\end{align}
where $ { C^i }_{ j k } $ are the structure constants of the space-time (see \cite{hom} for more details). The metric, in the physical space-time $ (\tM,\tg) $, is given in this basis by
\begin{equation}
\tilde{g}_{\alpha\beta}dx^\alpha dx^\beta = - dt^2 + \tilde{a}_{ij}(t)\sigma^{i}\sigma^{j}, \label{m1}
\end{equation}
where $t$ is proper-time on the congruence orthogonal to the surfaces of homogeneity (which is necessarily geodesic) and $\tilde{a}_{ij}$ is the physical metric on the surfaces of homogeneity. 
We define the rate of change, which is a constant multiple of the second fundamental form
\begin{align}\label{physicalk}
\tilde{k}_{ij} =  \frac{d\tilde{a}_{ij}}{dt} .
\end{align}
Now, we introduce a new time-coordinate $\tau$ by 
\begin{align}\label{ttau}
t=\frac12 \tau^2,
\end{align}
and choose the conformal factor
\begin{align}\label{Omegachoice}
\Omega=\tau.
\end{align}
Then, the rescaled metric $g_{\alpha\beta}$ is given by
\begin{align}
g_{\alpha\beta}dx^\alpha dx^\beta&=\Omega^{-2}\tg_{\alpha\beta}dx^\alpha dx^\beta= -d\tau^2+ a_{ij}(\tau)\sigma^i\sigma^j,
\label{m3}
\end{align}
where $ a_{ i j } $ is defined by
\begin{align}\label{aa}
a_{ij}=\tau^{-2}\tilde{a}_{ij}.
\end{align} 
Here, we assume that $ a_{ i j } $ is regular at the initial surface $\Sigma$ where $\tau=0$, hence the metric extends to $ M $. We shall systematically use $b^{ij}$ for the matrix inverse to $a_{ij}$, so that indices $i,j,\ldots$ are raised with $b^{ij}$ and lowered with $a_{ij}$.

Let us write the momentum in the basis of invariant one-forms as
\begin{align}
p_0d\tau+p_i\sigma^i.
\end{align}
The null condition \eqref{massless} implies
\begin{align}
p^0 = \sqrt{ b^{ij} p_i p_j } , \label{p3}
\end{align}
where we used \eqref{m3}. For the Boltzmann equation, we first obtain the geodesic equation, which is
\begin{align}
\frac{dp_i}{d\lambda}=C^k_{\;\;ij}b^{jm}p_kp_m,
\end{align}
where $ \lambda $ is proper time for $ g $. We abuse notation slightly by writing $ f ( p ) $ both for $ f ( t , p ) $ and $ f ( \tau , p ) $ to obtain
\begin{align}
\frac{df}{d\lambda}=p^0\frac{\partial f}{\partial\tau}+\frac{\partial f}{\partial p_i}C^k_{\;\;ij}p_kp_mb^{jm}.
\end{align}
Hence, the Boltzmann equation \eqref{V2} reduces to
\be
\frac{\partial f}{\partial\tau}+\frac{1}{p^0}C^k_{\;\;ij}p_kp_mb^{jm}\frac{\partial f}{\partial p_i}=\frac{1}{p^0}C(f,f),
\label{V3}
\ee
with the right hand side given by 
\be\label{V4} 
\frac{1}{p^0}C(f,f)= \tau^{\gamma -2} \int_{\bbr^3}\int_{\bbs^2}\frac{h^{2- \gamma }}{p^0}(f(p')f(q')-f(p)f(q)) \, d\omega \,\omega_{q},
\ee
where we used \eqref{Omegachoice}. Note that we still have the singular factor $ \tau^{ \gamma - 2 } $. As mentioned earlier, we need to redefine the time-coordinate, and this will be done in Section \ref{sMeqns}.

The Einstein equations are the same as for the Einstein-Vlasov case. The following first-order form of the Einstein evolution part of the Einstein-Vlasov system based on that in \cite{AT99, T07} but with the convention \eqref{ttau} (rather than $t=\tau^2$) may be obtained: with $k_{ij}$ proportional to the second-fundamental form, we have
\begin{align}
\frac{d}{d\tau}a_{ij}&= k_{ij},\label{E1}\\
\frac{d}{d\tau}b^{ij}&= -b^{im}b^{jn}k_{nm},\label{E2}\\
\frac{d}{d\tau}k_{ij}&= -2R_{ij}+\frac{1}{\tau}(2Z_{ij}-2k_{ij}-b^{mn}k_{mn}a_{ij})+k_{im}k_{jn}b^{mn}\nonumber\\
&\quad -\frac{1}{2}(b^{mn}k_{mn})k_{ij}+2\Lambda\tau^2a_{ij},\label{E3}
\end{align}
where $R_{ij}$ is the spatial Ricci tensor, so that
\begin{align}\label{ricci}
2R_{ij}&=-C^k_{\;\;ck}(C^r_{\;\;tj}a_{ir}+C^r_{\;\;ti}a_{jr})b^{ct}-C^c_{\;\;ki}(C^k_{\;\;cj}+C^m_{\;\;tj}a_{cm}b^{kt})\nonumber\\
&\quad -\frac12 C^m_{\;\;ks}C^r_{\;\;ct}a_{jm}a_{ir}b^{kt}b^{sc},
\end{align}
and $ Z_{ i j } $ is defined by
\begin{align}
Z_{ij}=\frac{1}{\tau}\left(\frac{8\pi}{\sqrt{ \det a}}\int fp_ip_j\frac{d^3p}{p^0} -a_{ij}\right).\label{z1}
\end{align}
The new variable $Z_{ij}$ was introduced in \cite{AT99} so that the pole in $\tau$ in (\ref{E3}) is no worse than first-order, but we therefore need an evolution equation for $Z_{ij}$. We can retain (\ref{E1})--(\ref{E3}) for the Einstein evolution part of the Einstein-Boltzmann system, but the evolution equation for $Z_{ij}$ must be changed from what it is in \cite{AT99}: the collision term will be involved through the equation \eqref{V3}. On the other hand, we need to redefine the time-coordinate in order to regularise the singularity at $ \tau = 0 $ in \eqref{V4}. Therefore, we need to recast \eqref{E1}--\eqref{E3} in terms of the new time-coordinate rather than $\tau$, and in doing this it will be convenient to define new $k_{ij}$ and $Z_{ij}$. The evolution equation for $ Z_{ i j } $ will be given in terms of the new variables. This will be done in Section \ref{sMeqns} below.

\subsubsection{The rescaled Einstein-Boltzmann system}\label{sMeqns}
Now, we introduce a new time coordinate to derive the main equations. Let us first consider the Boltzmann equation \eqref{V3}--\eqref{V4}. We regularise the singularity at $\tau=0$ in (\ref{V4}) by changing the time coordinate. Following \cite{LNT1} we introduce $s$ for $\tau$ via $ds=\tau^{ \gamma -2}d\tau$ so that
\begin{align}\label{stau}
( \gamma -1)s=\tau^{ \gamma -1}.
\end{align}
Then, the Boltzmann equation \eqref{V3}--\eqref{V4} transforms as
\begin{align}
\frac{\partial f}{\partial s} + ( \gamma -1)^{(2- \gamma )/( \gamma -1)} s^{(2- \gamma )/( \gamma -1)}\frac{1}{p^0}C^k_{\;\;ij}p_kp_mb^{jm}\frac{\partial f}{\partial p_i}\nonumber\\
=\int_{\bbr^3}\int_{\bbs^2}\frac{h^{2- \gamma }}{p^0}(f(p')f(q')-f(p)f(q))\, d\omega\, \omega_{q}.
\end{align}
For simplicity of notation we write
\begin{align}\label{cgamma}
c_\gamma = \frac{1}{ \gamma -1},
\end{align}
so that $(2- \gamma )/( \gamma -1) = c_\gamma -1$ and $( \gamma -1)^{(2- \gamma )/( \gamma -1)} = c_\gamma^{ 1- c_\gamma }$. We observe that $1< c_\gamma <\infty$, since $ 1 < \gamma < 2 $. This leads to the Boltzmann equation written as
\begin{align}
\frac{ \partial f }{ \partial s } + B ( f , b^{ i j } ; p_i , s ) = Q ( f , f ) , \label{V5}
\end{align}
where $ B $ and $ Q $ are given by
\begin{align}\label{V6}
B(f,b^{ij};p_i,s)&= c_\gamma^{1- c_\gamma } s^{ c_\gamma -1}\frac{1}{p^0}C^k_{\;\;ij}p_kp_mb^{jm}\frac{\partial f}{\partial p_i},\\
Q(f,f)&= \frac{1}{\sqrt{ \det a}}\int_{\bbr^3}\int_{\bbs^2}\frac{h^{2- \gamma }}{p^0q^0}(f(p')f(q')-f(p)f(q))\, d\omega\, d^3q .\label{C}
\end{align}
We note that the Boltzmann equation is still singular at $ p = 0 $, since $ p^0 $ is given by \eqref{p3}. This will be resolved by considering singular weights. See Section \ref{sNotations} for more details.

Next, we consider the Einstein equations. We need to derive an additional evolution equation for $ Z_{ i j } $ and rewrite it in terms of the time coordinate $ s $, as well as the equations \eqref{E1}--\eqref{E3}. It will be convenient to introduce new $ k_{ i j } $ and $ Z_{ i j } $ defined by
\begin{align}
K_{ij}=\tau^{2- \gamma }k_{ij},\label{K}\\
\hat{Z}_{ij}=\tau^{2- \gamma }Z_{ij}.\label{F9}
\end{align}
Then, we obtain the following evolution equations for $ a_{ i j } $, $ b^{ i j } $, $ K_{ i j } $ and $ \hat{ Z }_{ i j } $:
\begin{align}
\frac{d}{ds}a_{ij} & =K_{ij},\label{F1}\\
\frac{d}{ds}b^{ij} & =-b^{im}b^{jn}K_{nm},\label{F2}\\
\frac{d}{ds}K_{ij}&= \frac{1}{( \gamma -1)s}\left(- \gamma K_{ij}-b^{mn}K_{mn} a_{ij}+2\hat{Z}_{ij}\right)\nonumber\\
&\quad +K_{im}b^{mn}K_{nj}-\frac12 b^{mn}K_{mn} K_{ij}-2\tau^{4-2 \gamma }R_{ij}+2\Lambda\tau^{6-2 \gamma }a_{ij},\label{F3}\\
\frac{d}{ds}\hat{Z}_{ij}&=-\frac{1}{s}\hat{Z}_{ij}+\frac{1}{( \gamma -1)s}\left(-K_{ij}+\frac{4\pi}{\sqrt{ \det a}}\int fp_ip_j(K^{mn}p_mp_n-(p^0)^2b^{mn}K_{mn})\frac{d^3p}{(p^0)^3}\right)\nonumber\\
&\quad -\frac{1}{( \gamma -1)s}\left(\frac{8\pi c_\gamma^{1- c_\gamma } s^{ c_\gamma -1} }{\sqrt{ \det a}}\int p_ip_jC^k_{\;\;mn}b^{nr}p_kp_r\frac{\partial f}{\partial p_m}\frac{d^3p}{(p^0)^2}\right)\nonumber\\
&\quad +\frac{1}{( \gamma -1)s}\left(\frac{8\pi }{ \det a}\iiint \frac{h^{2- \gamma }}{p^0 q^0}f(p) f(q)\left(\frac{p'_ip'_j}{p'^0} - \frac{p_ip_j}{p^0}\right)\, d\omega\, d^3q\, d^3p \right),\label{F4}
\end{align}
where $ \hat{ Z }_{ i j } $ is given by
\begin{align}
\hat{Z}_{ij}=\frac{1}{( \gamma -1)s}\left(\frac{8\pi}{\sqrt{ \det a}}\int fp_ip_j\frac{d^3p}{p^0} -a_{ij}\right). \label{F9'}
\end{align}
We note that the above evolution equations are singular at $s=0$ but the singularity is only first-order so the system has a chance to be Fuchsian, as the corresponding system for the Einstein-Vlasov equations \cite{AT99} was. Hence, we need to identify the \emph{Fuchsian conditions}, which will be given in Section \ref{ssfc}. Moreover, there should be the Einstein constraints together with the evolution equations, and this will be given in Section \ref{constraints}.

To summarize, we have obtained the coupled system of rescaled equations \eqref{V5}, \eqref{F1}--\eqref{F4}. The equations are derived by assuming the spatial homogeneity \eqref{m1}, rescaling the metric \eqref{m3} with the choice \eqref{ttau}--\eqref{Omegachoice} of the conformal factor, and changing the time coordinate with \eqref{stau}. The unknowns are $ a_{ i j } $, which is the unphysical metric, $ b^{ i j } $ the inverse of $ a_{ i j } $, $ K_{ i j } $ defined by \eqref{E1} and \eqref{K}, $ { \hat Z }_{ i j } $ by \eqref{z1} and \eqref{F9}, and the distribution function $ f $. We note that the physical metric $ { \tilde a }_{ i j } $ and the second fundamental form $ { \tilde k }_{ i j } $ can be recovered by \eqref{physicalk} and \eqref{aa} using the time coordinates \eqref{ttau} and \eqref{stau}. In the rest of the paper, the Einstein or the Boltzmann equations will refer to the rescaled Einstein or the rescaled Boltzmann equations unless otherwise specified.

\subsection{Main results} \label{mainresults}
We first collect the notations that will be used in the rest of the paper.

\subsubsection{Notations} \label{sNotations}
The unknowns for the Einstein-Boltzmann system are $ 3 \times 3 $ real symmetric matrices and a distribution function. We denote by $S_3(\bbr)$ the space of $ 3 \times 3 $ symmetric matrices equipped with the following norm:
\begin{align}\label{matrixnorm}
\| A \| = \max_{i,j = 1,2,3} | A_{ij}|. 
\end{align}
The above norm will also be used for matrices with other types of indices, for instance
\begin{align}
\| b \| = \max_{i,j = 1 , 2 , 3 } | b^{ij} | , \qquad \| e \| = \max_{ i , j = 1 , 2 , 3 } | { e_j }^i | , \qquad \| e^{ - 1 } \| = \max_{ i , j = 1 , 2 , 3 } | { e^i }_j | , 
\end{align}
which can be found in \eqref{ortho_e}--\eqref{einverse}. This will also be used for tensors, or any other quantities, with an arbitrary number of indices, such as $ \Psi $ in \eqref{Psi}, $ F $ in \eqref{x}, etc.

For the distribution function we use weighted $ L^p $-spaces. Let $ L^1_r ( \bbr^3 ) $ and $ L^\infty_\eta ( \bbr^3 ) $ denote the spaces of functions equipped with the following norms:
\begin{align}
& \| f \|_{L^1_r} = \int_{\bbr^3} | f(p) | ( p^0)^r \, d^3p,\qquad p^0 = \sqrt{b^{ij} p_i p_j}, \\
& \| f \|_{ L^\infty_\eta} = \sup_{ p \in \bbr^3 } | w_\eta f ( p ) | , \qquad w_\eta = p^0 \exp ( s^{ - 1 }_\eta p^0 ) , \qquad s_\eta = ( s + \eta^2 )^\eta , \qquad \eta > 0 .
\end{align}
We shall consider negative values for $ r $ in order to control the singularities at $ p = 0 $, which appear in the case of massless particles. The weight $ w_\eta $ is basically the inverse of the J{\" u}ttner distribution, but it should be multiplied by an additional $ p^0 $ in order to control the singularities. This was first introduced in \cite{Lee21}, where the massless Boltzmann equation was studied in the FLRW case, and the factor $ s_\eta^{ - 1 } $ was not taken into account. In the Bianchi case, we need the factor $ s_\eta^{ - 1 } $ in order that the weight works properly. This will be crucially used in the estimate \eqref{dw ds}.

We notice that the above norms depend on the metric $b^{ij}$, which complicates the iteration procedure for the coupled Einstein-Boltzmann equations. To avoid this complication we introduce the following norm: 
\begin{align}
\| f \|_{ \langle r \rangle } & = \int_{ \bbr^3 } | f ( p ) | \langle p \rangle^r \, d^3 p , \qquad \langle p \rangle = \sqrt{ \delta^{ i j } p_i p_j } ,
\end{align}
which does not depend on $ b^{ i j } $. The space of functions with the norm $ \| \cdot \|_{ \langle r \rangle } $ will be denoted by $ L^1_{ \langle r \rangle } ( \bbr^3 ) $.

\subsubsection{Main theorems}\label{maintheorems}
Now, we state the main results of this paper. The main theorem is Theorem \ref{phystheorem}.

\begin{thm}\label{phystheorem}
Let $ f_0 \geq 0 $ be a smooth function with compact support in $ \bbr^3 \setminus \{ 0 \} $. Suppose that $ f_0 $ is not identically zero and satisfies the constraint \eqref{C3}. Then, there exists a unique Bianchi I solution $ \tilde{a}_{ i j } , \tilde{k}_{ i j } \in C^1 ( ( 0 , T ] ; S_3 ( \bbr ) ) $ and $ 0 \leq f \in C^1 ( ( 0 , T ] ; L^1 ( \bbr^3 ) ) $ to the massless (unrescaled) Einstein-Boltzmann system with an initial conformal gauge singularity for the scattering cross-sections in \eqref{cross} such that $ f $ converges to $ f_0 $ in $ L^1 $ as $ t \to 0^+ $. Furthermore, the solutions have the following asymptotics as $t \rightarrow 0^+$:
\begin{align}
\label{1} \tilde a _ { i j } & =  \mathcal{A} _ {  i j } t + \mathcal{B}_{ij} t  ^ { \frac { \gamma + 1 } 2 } + \mathcal{C}_{ij} t  ^ \gamma + o ( t ^ \gamma ),\\
\label{2}  \tilde{k}_{ij} & =  \mathcal{A} _ {  i j } +  \frac { \gamma + 1 } 2 \mathcal{B}_{ij} t  ^ { \frac { \gamma - 1 } 2 } + { \gamma} \mathcal{C}_{ij} t  ^ {\gamma-1} + o ( t ^ {\gamma-1} ),
\end{align}
where $\mathcal{A} _ {  i j } $, $\mathcal{B}_{ij} $ and $ \mathcal{C}_{ij}$ are constants which only depend on $f_0$ and $\gamma$ given in \eqref{not}.
\end{thm}

Theorem \ref{phystheorem} shows that the Cauchy problem for the massless (unrescaled) Einstein-Boltzmann system with an initial conformal gauge singularity is well-posed in the Bianchi I case. The theorem will be proved by considering the rescaled Einstein-Boltzmann system \eqref{V5}, \eqref{F1}--\eqref{F4}, which are valid for arbitrary Bianchi types, but we will restrict to the Bianchi I case to obtain the desired result. In the rest of the paper, we will only consider the Bianchi I case, i.e.~$ { C^k }_{ i j } = 0 $, so that the quantity $ B $ given by \eqref{V6} will not be considered for the Boltzmann equation, and the evolution equation \eqref{F4} for $ \hat{ Z }_{ i j } $ will also be simplified. Theorem \ref{phystheorem} will be proved by applying Theorems \ref{fuchsian} and \ref{conformalprop}, which will be combined in Section \ref{sMain2} to complete the proof of Theorem \ref{phystheorem}.

\medskip

The initial data given in Theorem \ref{phystheorem} is a single distribution function $ f_0 $. To study the rescaled equations \eqref{F1}--\eqref{F4}, we also need the initial values of $ a_{ i j } $, $ b^{ i j } $, $ K_{ i j } $ and $ \hat{Z}_{ i j } $. Theorem \ref{fuchsian} shows that they are uniquely determined by $ f_0 $.

\begin{thm}\label{fuchsian}
Let $ f_0 \geq 0 $ be a smooth function with compact support in $ \bbr^3 \setminus \{ 0 \} $. Suppose that $ f_0 $ is not identically zero. Then, there exist unique $ 3 \times 3 $ symmetric matrices $ a_0 $, $ b_0 $, $ K_0 $ and $ { \hat Z }_0 $ satisfying the Fuchsian conditions \eqref{FF1}--\eqref{FF3}.
\end{thm}

Theorem \ref{fuchsian} shows that the initial values $ a_0 $, $ b_0 $, $ K_0 $ and $ { \hat Z }_0 $ for the rescaled equations \eqref{F1}--\eqref{F4} are uniquely determined by $ f_0 $ through the Fuchsian conditions \eqref{FF1}--\eqref{FF3}. The Fuchsian conditions are necessary conditions for singular ODEs to have solutions differentiable at $ s = 0 $, which we will assume in Theorem \ref{conformalprop} to obtain the existence of solutions $ a_{ i j } $, $ b^{ i j } $, $ K_{ i j } $, $ \hat{ Z }_{ i j } $ and $ f $ to the rescaled Einstein-Boltzmann system \eqref{V5}, \eqref{F1}--\eqref{F4}. The proof of Theorem \ref{fuchsian} will be given separately in Sections \ref{ssfc} and \ref{eigen}.

\medskip

Now, we are given a unique set of initial data $ a_0 $, $ b_0 $, $ K_0 $, $ \hat{ Z}_0 $ and $ f_0 $ for the rescaled Einstein-Boltzmann system \eqref{V5}, \eqref{F1}--\eqref{F4} satisfying the Fuchsian conditions \eqref{FF1}--\eqref{FF3}. Theorem \ref{conformalprop} shows that the rescaled Einstein-Boltzmann system \eqref{V5}, \eqref{F1}--\eqref{F4} has a local solution corresponding to the initial data.

\begin{thm}\label{conformalprop}
Let $ a_0 , b_0 , K_0 , { \hat Z }_0 \in S_3 ( \bbr ) $ and $ 0 \leq f_0 \in L^1 ( \bbr^3 ) $ be initial data of the rescaled Einstein-Boltzmann system \eqref{V5}, \eqref{F1}--\eqref{F4} with Bianchi I symmetry, satisfying the Fuchsian conditions \eqref{FF1}--\eqref{FF3} and the constraint \eqref{C3}. Suppose that
\begin{align}
f_0 \in L^1_1( \bbr^3 ) \cap L^1_{ - 2 - \delta / 2 } ( \bbr^3 ) \cap L^\infty_\eta ( \bbr^3) , \qquad \frac{ \partial f_0 }{ \partial p } \in L^1_1 ( \bbr^3 ) \cap L^1_{ - 1 - \delta / 2 } ( \bbr^3 ) ,
\end{align}
where $ \delta $ and $ \eta $ are positive real numbers satisfying
\begin{align}
\gamma + \delta < 2 , \qquad \eta< \frac{ 1 }{ 2 \max ( \| a_0 \| , \| b_0 \| , \| K_0 \| , \| \hat{Z}_0 \| )^2 } .
\end{align}
Then, there exists a time interval $ [ 0 , T ] $ on which the rescaled Einstein-Boltzmann system has a unique solution $ a_{ i j } , b^{ i j } , K_{ i j } , { \hat Z }_{ i j } \in C^1 ( [ 0 , T ] ; S_3 ( \bbr ) ) $ and $ 0 \leq f \in C^1 ( [ 0 , T ] ; L^1 ( \bbr^3 ) ) $. Moreover, the distribution function $ f $ satisfies
\begin{align}
\sup_{ 0 \leq s \leq T } \| f ( s ) \|_{ L^1_{ - 1 } } + \sup_{ 0 \leq s \leq T } \| f ( s ) \|_{ L^1_1 } + \sup_{ 0 \leq s \leq T } \| f ( s ) \|_{ L^\infty_\eta } \leq C .
\end{align}
\end{thm}

We note that Theorem \ref{conformalprop} is proved for $ f_0 \in L^1 ( \bbr^3 ) $ with several different weights. We assume in Theorem \ref{phystheorem} that $ f_0 $ is compactly supported in $ \bbr^3 \setminus \{ 0 \} $, but it is understood in Theorem \ref{conformalprop} as $ f_0 \in L^1_r ( \bbr^3 ) $ with $ r $ negative. For the proof of Theorem \ref{conformalprop}, we study the Einstein and the Boltzmann equations separately in Sections \ref{sEinstein} and \ref{sBoltzmann}. In Section \ref{sEinstein}, we study the Einstein equations \eqref{F1}--\eqref{F4} for a given $ f $ to obtain the existence in Proposition \ref{propeinstein}. In Section \ref{sBoltzmann}, we study the Boltzmann equation. The existence is obtained in Proposition \ref{prop B}, and the continuous dependence on the metric is obtained in Proposition \ref{prop C}. Finally in Section \ref{sMain1}, we combine Propositions \ref{propeinstein}, \ref{prop B} and \ref{prop C} to complete the proof of Theorem \ref{conformalprop}.

\subsection{Rescaled collision operator}
Let us consider the collision operator \eqref{C} in more detail. Recall that the equation \eqref{V5} with \eqref{V6} and \eqref{C} is the rescaled Boltzmann equation in the unphysical space-time $ ( M , g ) $, obtained from the physical space-time $ ( \tM , \tg ) $ through the rescaling \eqref{g1}. By the choice \eqref{pscaling} of the coordinates on $ N $, the momentum conservation \eqref{conservationtilde} reduces to
\begin{align}
p_\alpha + q_\alpha = p'_\alpha + q'_\alpha. \label{conservation}
\end{align}
Here, $ p_\alpha $, $ q_\alpha $, $ p'_\alpha $ and $ q'_\alpha $ are null momenta in $ ( M , g ) $, where $ g $ is given by \eqref{m3} so that we have
\begin{align}
p^0 = \sqrt{ b^{ i j } p_i p_j } , \qquad q^0 = \sqrt{ b^{ i j } q_i q_j } ,
\end{align}
and similar formulae for $ p'^0 $ and $ q'^0 $, where $ b^{ i j } $ is the matrix inverse of $ a_{ i j } $. In order to derive the parametrization of post-collision momenta, we consider an orthonormal frame $ e_i $ with $ e^i $ satisfying
\begin{align}
{ e_i }^k { e_j }^l a_{ k l } = \delta_{ i j }, \qquad { e^i }_k { e^j }_l b^{ k l } = \delta^{ i j }, \label{aeb}
\end{align}
which can be explicitly given, as in \cite{LLN}, by
\begin{align}
( { e_j }^i ) & = 
\begin{pmatrix}
e_1 & e_2 & e_3
\end{pmatrix} =
\begin{pmatrix}
\frac{b^{11}}{\sqrt{b^{11}}}  & 0 & 0 \\ 
\frac{b^{12}}{\sqrt{b^{11}}} & \frac{b^{11} b^{22} - (b^{12})^2}{\sqrt{b^{11}(b^{11} b^{22} - (b^{12})^2)}} & 0 \\
\frac{b^{13}}{\sqrt{b^{11}}} & \frac{b^{11} b^{23} - b^{12}b^{13}}{\sqrt{b^{11}(b^{11} b^{22} - (b^{12})^2)}} & \frac{(\det b)}{\sqrt{(b^{11} b^{22} - (b^{12})^2)(\det b)}}
\end{pmatrix}, \label{ortho_e} \\
( { e^i }_j ) & = 
\begin{pmatrix}
e^1 \\
e^2 \\
e^3
\end{pmatrix} =
\begin{pmatrix}
\frac{1}{\sqrt{b^{11}}}  & 0 & 0 \\ 
\frac{-b^{12}}{\sqrt{b^{11}(b^{11} b^{22} - (b^{12})^2)}} & \frac{b^{11}}{\sqrt{b^{11}(b^{11} b^{22} - (b^{12})^2)}} & 0 \\
\frac{b^{12}b^{23} - b^{13} b^{22}}{\sqrt{(b^{11} b^{22} - (b^{12})^2)(\det b)}} & \frac{-b^{11} b^{23} + b^{12}b^{13}}{\sqrt{(b^{11} b^{22} - (b^{12})^2)(\det b)}} & \frac{b^{11} b^{22} - (b^{12})^2}{\sqrt{(b^{11} b^{22} - (b^{12})^2)(\det b)}}
\end{pmatrix}. \label{ortho_theta}
\end{align}
For simplicity of notations we will denote the $ 3 \times 3 $ matrices \eqref{ortho_e} and \eqref{ortho_theta} by
\begin{align}
e = ( { e_j }^i )_{ i , j = 1 , 2 , 3 } , \qquad e^{ - 1 } = ( { e^i }_j )_{ i , j = 1 , 2 , 3 } , \label{einverse}
\end{align}
respectively, as $ { e_j }^i $ and $ { e^i }_j $ are matrix inverses to each other.

Now, we write the collision operator \eqref{C} with respect to the orthonormal frame to obtain the same representation as in the Minkowski case. To be precise, let us write
\begin{align}
p = p_i \sigma^i = \hat{ p }_j e^j , \qquad p_i = \hat{ p }_j { e^j }_i , \label{phat}
\end{align}
to obtain
\begin{align}
h = \sqrt{2p^0q^0 - 2\delta^{ij}\hat{p}_i \hat{q}_j },\qquad p^0 = \sqrt{\delta^{ij } \hat{p}_i \hat{p}_j },\qquad q^0 = \sqrt{\delta^{ij } \hat{q}_i \hat{q}_j }. \label{hp^0q^0hat}
\end{align}
In particular, the representations of post-collision momenta are given by
\begin{align}
p'^0 & = \frac{1}{2}( n^0 + \hat{n}_i \omega_j \delta^{ i j } ), \label{phat'^0} \\
q'^0 & = \frac{1}{2}( n^0 - \hat{n}_i \omega_j \delta^{ i j } ), \label{qhat'^0}
\end{align}
and
\begin{align}
\hat{p}'_j & = \frac{1}{2}\left( \hat{n}_j + h\omega_j + \frac{ \hat{n}_k \omega_l \delta^{ k l } \hat{ n }_j }{ n^0 + h } \right), \label{phat'} \\
\hat{q}'_j & = \frac{1}{2}\left( \hat{n}_j - h\omega_j - \frac{ \hat{n}_k \omega_l \delta^{ k l } \hat{ n }_j }{ n^0 + h } \right), \label{qhat'}
\end{align}
where we write $ n^0 = p^0 + q^0 $ and $ \hat{ n }_i = \hat{ p }_i + \hat{ q }_i $ for simplicity, and $ \omega_j $ is a unit vector in $ \bbs^2 $. These are the same as (19) of \cite{LNT1}, and they can also be derived from the representations in \cite{SY} by assuming massless particles. In order to write the post-collision momenta in terms of the variables $ p_i $, $ q_i $ and $ \omega_i $, we use \eqref{aeb} and \eqref{phat} again to obtain
\begin{align}\label{hp^0q^0}
h = \sqrt{2p^0q^0 - 2b^{ij} p_i q_j},\qquad p^0 = \sqrt{b^{ij} p_i p_j},\qquad q^0 = \sqrt{b^{ij} q_i q_j}.
\end{align}
Moreover, since $ { e^j }_l b^{ k l } = \delta^{ i j } { e_i }^k $, we have
\begin{align}
\hat{ n }_i \omega_j \delta^{ i j } = \hat{ n }_i \omega_j { e^i }_k { e^j }_l b^{ k l } = n_k \omega_j { e^j }_l b^{ k l } = n_k \omega_j \delta^{ i j } { e_i }^k ,
\end{align}
so that the representations \eqref{phat'^0}--\eqref{qhat'} are now written as
\begin{align}
p'^0 & = \frac{1}{2} (n^0 + n_k \omega_j \delta^{ i j } { e_i }^k),\label{p'^0}\\
q'^0 & = \frac{1}{2} (n^0 - n_k \omega_j \delta^{ i j } { e_i }^k),\label{q'^0}
\end{align}
and
\begin{align}
p'_i & = \frac{1}{2}\left(n_i + h \omega_j { e^j }_i + \frac{ n_l \omega_k \delta^{ j k } { e_j }^l n_i }{ n^0 + h}\right),\label{p'}\\
q'_i & = \frac{1}{2}\left(n_i - h \omega_j { e^j }_i - \frac{ n_l \omega_k \delta^{ j k } { e_j }^l n_i }{ n^0 + h}\right). \label{q'}
\end{align}
Here, $n^0$ and $n_i$ denote 
\begin{align}
n^0 = p^0 + q^0,\qquad n_i = p_i + q_i,
\end{align}
and $\omega_j \in \bbs^2 $ is unit in the sense
\begin{align}
\sum_{j=1}^3(\omega_j)^2 = 1. 
\end{align}
We conclude from \eqref{p'^0}--\eqref{q'} that the post-collision momenta are now parametrised by $p_i , q_i \in \bbr^3 $, $ \omega_i \in \bbs^2 $ and $b^{ij}$. Finally, the collision operator \eqref{C} is determined by \eqref{hp^0q^0}, \eqref{p'} and \eqref{q'}.

\subsection{The Fuchsian conditions}\label{ssfc}
If the system consisting of \eqref{F1}--\eqref{F4} with \eqref{F9'} admits solutions differentiable at $s=0$, then we may multiply \eqref{F3}, \eqref{F4} and \eqref{F9'} by $s$ and set $s=0$ to obtain relations between the initial values $a_{0\,ij} $, $ b_0^{\;\;ij} $, $ K_{0\,ij} $, $ \hat{Z}_{0\,ij} $ and $ f_0 $. These are called the Fuchsian conditions and given by
\begin{align}
0 & = \frac{ 8 \pi }{ \sqrt{ \det a_0 } } \int f_0 p_i p_j \frac{ d^3 p }{ p^0 } - a_{0\,ij}, \label{FF1} \\
0 & = - \gamma K_{0\,ij} - (b_0^{\;mn}K_{0\,mn})a_{0\,ij}+2\hat{Z}_{0\,ij}, \label{FF2}\\
0&= - ( \gamma -1)\hat{Z}_{0\,ij}-K_{0\,ij} +\frac{4\pi}{\sqrt{ \det a_0}}\int f_0p_ip_j(K_0^{\;mn}p_mp_n - ( p^0 )^2 b_0^{\;mn}K_{0\,mn} )\frac{d^3p}{ ( p^0 )^3 } \nonumber\\
&\quad + \frac{8\pi }{ \det a_0 } \iiint \frac{h^{2- \gamma }}{p^0 q^0} f_0 (p) f_0 (q)\left(\frac{p'_ip'_j}{p'^0} - \frac{p_ip_j}{p^0}\right)\, d\omega\, d^3q\, d^3p . \label{FF3}
\end{align}
Here, we remark that the quantities $ p^0 $, $ q^0 $, $ h $, $ p'_i $, $ p'_j $ and $ p'^0 $ are evaluated at $ s = 0 $. To be precise, 
\begin{align}
p^0 |_{ s = 0 } = \sqrt{ { b_0 }^{ i j } p_i p_j } , \qquad q^0|_{ s = 0 } = \sqrt{ { b_0 }^{ i j } q_i q_j } ,
\end{align}
so that
\begin{align}
h |_{ s = 0 } = \sqrt{ 2 p^0 |_{ s = 0 } q^0 |_{ s = 0 } - 2 { b_0 }^{ i j } p_i q_j } ,
\end{align}
and $ p'_i $, $ p'_j $ and $ p'^0 $ are evaluated at $ s = 0 $ by using \eqref{p'^0} and \eqref{p'} with $ { b_0 }^{ i j } $.

The Fuchsian conditions \eqref{FF1}--\eqref{FF3} imply that the initial values $ { a_0 }_{ i j } $, $ { b_0 }^{ i j } $, $ { K_0 }_{ i j } $ and $ \hat{Z}_{ 0 \, i j } $ are uniquely determined by $ f_0 $ as we will now show.
Equation \eqref{FF1} determines $ { a_0 }_{ i j } $ uniquely given $ f_0 $ by the argument of Theorem 4.1 of \cite{AT99}. Define $ { b_0 }^{ij}$ to be the inverse of $ a _ { 0 \, i j } $. Then $b^{ij}$ will be the inverse of $a_{ij}$ for all time as a direct computation shows. Define
\begin{align}
\zeta^{k}_i = a_{ij} b^{jk} - \delta_i^k.
\end{align}
Then from \eqref{F1} and \eqref{F2} we have
\begin{align}
\frac{d}{ds} \zeta^{k}_i & = b^{jk} K_{ij} -a_{ij} b^{jm}b^{kn}K_{nm} \nonumber \\
& = b^{nk}K_{in} - a_{ij} b^{jm}b^{nk}K_{mn} \nonumber \\
& = \delta^m_i b^{nk} K_{mn}  - a_{ij} b^{jm}b^{nk}K_{mn} \nonumber \\
& = - \zeta^m_i b^{nk} K_{mn}.
\end{align}
As a consequence if $\zeta$ is zero initially it will remain zero for all time, which means that $b^{ i j } $ is the inverse for $a_{ i j } $ for all time if it is initially. Now given $f_0$, $a _ { 0 \, i j } $ and ${ b_0}^{ij}$ we will see in the end of Section \ref{eigen} that $ K _ { 0 \, i j } $ and $ \hat Z _ { 0 \, i j } $ are
then also uniquely determined.

\subsection{The Einstein constraints}\label{constraints}
Together with the Einstein evolution equations \eqref{E1}--\eqref{E3} there will be the Einstein constraints. For the Hamiltonian constraint in terms of conformal time $\tau$ we have from \cite{T07}:
\be
C:=\frac{4\pi}{\sqrt{ \det a}}\int fp^0 d^3p-\frac{3}{2}-\tau^2 \left( \frac{R}{4}-\frac{1}{16}k^{ij}k_{ij}+\frac{1}{16}k^2 \right) -\frac{1}{2}k\tau+\frac{\Lambda}{2}\tau^4=0,
\label{C1}
\ee
while for the momentum constraint we have:
\be
C_i:=\frac{8\pi}{\sqrt{ \det a}}\int fp_i \, d^3p-\tau^2(b^{mn}k_{nj}C^j_{\;\;mi}+b^{mn}k_{ni}C^j_{\;\;mj})=0.
\label{C2}
\ee
Following \cite{T07} we have for the evolution of the constraints:
\begin{align}
\partial_\tau(\tau^2 \det a \, C)&= C_ib^{ij}C^m_{\;\;jm},\\
\partial_\tau(\sqrt{ \det a} \, C_i)&= 0,
\end{align}
so that the constraints are satisfied at all times if satisfied initially. (This argument can be rephrased in terms of $s$, when it still holds.) To satisfy the constraints initially we need for (\ref{C2}) a condition on the initial distribution function $f_0$:
\be
\int f_0p_i\,d^3p=0.
\label{C3}
\ee
Note that we have from \eqref{C1}
\be
\int f_0(b_0^{\;ij}p_ip_j)^{1/2}\, d^3p=\frac{3\sqrt{ \det a_0}}{8\pi},
\label{C4}
\ee
which is essentially just a normalisation condition and follows from the Fuchsian condition \eqref{FF1}. Hence, the Einstein constraints reduce to the single integral constraint \eqref{C3}.

\section{Estimates for the Einstein equations}\label{sEinstein}
In this part we obtain the existence for the Einstein equations. Assuming that the distribution function $ f $ is given, we will show that the Einstein equations \eqref{F1}--\eqref{F4} admit solutions in a suitable sense. The main result of this section is Proposition \ref{propeinstein}, and its proof is a careful application of Theorem \ref{ODEthm} from the appendix. In order to apply Theorem \ref{ODEthm} we first transform the equations \eqref{F1}--\eqref{F4} into a suitable form, which will be done in Section \ref{sEinsteinODE}. Proposition \ref{propeinstein} will be presented in Section \ref{ExistEinstein}, and its proof will be given in Section \ref{sproofE}.

\subsection{The Einstein equations in a suitable form}\label{sEinsteinODE}
We need to put the equations \eqref{F1}--\eqref{F4} in a suitable form in order to apply Theorem~\ref{ODEthm} from the appendix. The distribution function $f$ will appear in the system, but the Boltzmann equation $\eqref{V5}$ with $\eqref{V6}$ and \eqref{C} will not be part of the system. It will be convenient to introduce the following tensor with an arbitrary number $n \geq 2 $ of indices:
\begin{align}\label{Psi}
\Psi_{i_1 i_2 \cdots i_n}= \frac { 4 \pi } { \sqrt { \det a } } \int \left ( b ^ { k l } p_k p_l \right)^{ -\frac{n-1}{2} } p_{i_1} p_{i_2}\cdots p_{i_n} f  \, d^3 p.
\end{align}
Note that it is symmetric under permutation of any of its indices and
\begin{align}
b^{ij} \Psi_{ijk_1\cdots k_m}  = \Psi_{k_1\cdots k_m}
\end{align}
for any $m\geq 2$. In the same way we define for an arbitrary number $n \geq 2$ of indices
\begin{align}\label{Phi}
\Phi_{i_1 i_2 \cdots i_n} = \frac { 4 \pi } { \sqrt { \det a } } \int \left ( b ^ { k l } p _ k p _ l \right ) ^ { - \frac{n-1}{2}  } p_{i_1} p_{i_2} \cdots p_{i_n}  \frac{ \partial f}{\partial s}  \, d^3p .
\end{align}
We also define the following tensors with $n \geq 4$ lower indices  and one upper index
\begin{align}
\Xi_{i_1 i_2 \cdots i_n}^j & = - \frac{ 8 \pi c_\gamma^{2- c_\gamma }} { \sqrt{ \det a } }  \int  \left( b ^ { k l } p _ k p _ l \right ) ^ { - \frac{n-2}{2} }  p_{i_1} p_{i_2} \cdots p_{i_n}  \frac{ \partial f}{\partial p_j} \, d^3 p, \label{Xi} \\
\Upsilon_{i_1 i_2 \cdots i_n}^j & = - \frac { 8 \pi c_\gamma^{2- c_\gamma }} { \sqrt{ \det a } }  \int  \left( b ^ { k l } p _ k p _ l \right ) ^ { - \frac{n-2}{2} }  p_{i_1} p_{i_2}\cdots p_{i_n}  \frac{ \partial^2 f}{\partial p_j \partial s}  \, d ^ 3 p. \label{xi}
\end{align}
Note that the above tensors \eqref{Psi}--\eqref{xi} do not depend on $ K_{ i j } $ nor $ \hat{ Z }_{ i j } $.

\subsubsection{The evolution equations}
We first write the evolution equations \eqref{F1}--\eqref{F4} as
\begin{align}
&\label{I}\frac{d}{ds}a_{ij}=K_{ij},\\
&\label{II}\frac{d}{ds}b^{ij}=-b^{im}b^{jn}K_{nm},\\
&\label{III}\frac{d}{ds}K_{ij}= \frac{ c_\gamma }{s}(- \gamma K_{ij} - 3 \pi^{mn}_{ij}  K_{mn} +2\hat{Z}_{ij}) + G_{ij},\\
&\label{IV}\frac{d}{ds}\hat{Z}_{ij}=-\frac{1}{s}\hat{Z}_{ij}+\frac{ c_\gamma }{s}(-K_{ij}+\chi^{qr}_{ij} K_{qr}-\Pi^{mn}_{ij} K_{mn} )+\frac1s H_{ij},
\end{align}
where
\begin{align}
&\pi ^ { m n } _ { i j } = \frac 1 3 a _ { i j } b ^ { m n } ,
\qquad \chi^{qr}_{ij}= \Psi_{ijmn}b^{mq}b^{nr}, 
\qquad \Pi_{ij}^{mn}= \Psi_{ij} b^{mn}, \label{pi}
\end{align}
and
\begin{align}
G_{ij} &=K_{im}b^{mn}K_{nj}-\frac12 b^{mn}K_{mn} K_{ij} - 2 c_\gamma^{2-2 c_\gamma } s^{2 c_\gamma -2}R_{ij}+2\Lambda c_\gamma^{2-4 c_\gamma } s^{4 c_\gamma -2} a_{ij},\\
H_{ij} & = \frac{8\pi c_\gamma }{ \det a}\iiint \frac{h^{2- \gamma }}{p^0 q^0}f(p)f(q)\left(\frac{p'_i p'_j}{p'^0} - \frac{p_i p_j}{p^0}\right) d\omega \, d^3q\, d^3p. \label{H}
\end{align}
Here, we assumed the Bianchi I symmetry. Otherwise, the following term should be added in \eqref{H}:
\begin{align}
s^{ c_\gamma -1} C^k_{mn} b^{nr} \Xi_{ijkr}^m .
\end{align}
We can transform the equations \eqref{I}--\eqref{IV} into matrix form as follows:
\begin{align}
\frac{d}{ds}
& \begin{pmatrix}
a_{ij} \\ 
b^{ij}
\end{pmatrix}
=
\begin{pmatrix}
K_{ij}\\
-b^{im}b^{jn}K_{nm}
\end{pmatrix}, \label{ODEx} \\
s \frac{d}{ds} 
& \begin{pmatrix}
K_{ij} \\ 
\hat{Z}_{ij} 
\end{pmatrix}
+ c_\gamma 
\begin{pmatrix}
 \gamma \delta^m_i \delta^n_j + 3 \pi_{ij}^{mn}  & -2\delta^m_i \delta^n_j \\ 
\delta^m_i \delta^n_j - \chi_{ij}^{mn} + \Pi_{ij}^{mn}& c_\gamma^{-1}\delta^m_i \delta^n_j 
\end{pmatrix}
\begin{pmatrix}
K_{mn} \\ 
\hat{Z}_{mn} 
\end{pmatrix}
= s \begin{pmatrix}
G_{ij} \\ 0 
\end{pmatrix}
+ \begin{pmatrix}
0 \\ 
H_{ij} 
\end{pmatrix}. \label{ODE}
\end{align}
Let $ x $, $ y $, $ F $, $ G $ and $ H $ denote the following quantities:
\begin{align} \label{x}
x= \begin{pmatrix}
a_{ij}\\
b^{ij}
\end{pmatrix},
\quad
y=\begin{pmatrix}
K_{ij}\\
\hat{Z}_{ij}
\end{pmatrix},
\quad
F= \begin{pmatrix}
K_{ij}\\
-b^{im}b^{jn} K_{nm}
\end{pmatrix},
\quad
G= \begin{pmatrix}
G_{ij}\\
0
\end{pmatrix},
\quad
H= \begin{pmatrix}
0\\
H_{ij}
\end{pmatrix},
\end{align}
and $ N $ be the coefficient function
\begin{align}\label{N}
N= c_\gamma 
\begin{pmatrix}
\gamma \delta^m_i \delta^n_j + 3 \pi_{ij}^{mn}  & -2\delta^m_i \delta^n_j \\ 
\delta^m_i \delta^n_j - \chi_{ij}^{mn} + \Pi_{ij}^{mn} & c_\gamma^{-1}\delta^m_i \delta^n_j 
\end{pmatrix}.
\end{align}
Now, we observe that $ F $ and $ G $ depend on both $ x $ and $ y $, while $ N $ and $ H $ depend only on $ x $, so that the equations \eqref{ODEx}--\eqref{ODE} are in the form \eqref{IVP} from the appendix.

\subsubsection{Fuchsian and compatibility condition}
Together with the Einstein equations \eqref{F1}--\eqref{F4}, which have been rewritten as \eqref{ODEx}--\eqref{ODE}, the Fuchsian conditions \eqref{FF1}--\eqref{FF3} also need to be rewritten in a similar way. In our new notation the first condition \eqref{FF1} is written as
\begin{align}\label{FF1bis}
  \Psi _ { i j } = \frac 1 2 a _ { i j } ,
\end{align}
which is for $s=0$. In this case we have that
\begin{align}
\Pi_{ij}^{mn} = \frac{1}{2}a_{ij} b^{mn} = \frac{3}{2} \pi_{ij}^{mn},
\end{align}
so that the second and third Fuchsian conditions \eqref{FF2} and \eqref{FF3} can be expressed as
\begin{align}
\gamma K_{ i j } + 3 \pi_{ij}^{ m n } K _ { m n } - 2 \hat Z _ { i j } & = 0,\label{FF2bis}\\
K_{ i j } - \chi^{ m n }_{ i j } K_{ m n } + \frac{3}{2} \pi_{ij}^{mn} K_{mn} + c_\gamma^{-1} \hat{Z}_{ij} & = c_\gamma^{-1} H_{ij},\label{FF3bis}
\end{align}
which hold for $s=0$. We remark that the conditions \eqref{FF2bis} and \eqref{FF3bis} agree with the compatibility condition \eqref{compatibility} from the appendix. With the notations in \eqref{x} and \eqref{N}, the conditions \eqref{FF2bis} and \eqref{FF3bis} are written as
\begin{align}\label{Compatibility}
Qy= c_\gamma^{-1}H,
\end{align}
where $ Q $ is given by
\begin{align}\label{Q}
Q = 
\begin{pmatrix}
\gamma \delta^m_i \delta^n_j + 3 \pi_{ij}^{mn} & - 2 \delta^m_i \delta^n_j \\ 
\delta^m_i \delta^n_j - \chi_{ij}^{mn} + 3\pi_{ij}^{mn} / 2 & c_\gamma^{-1}\delta^m_i \delta^n_j 
\end{pmatrix},
\end{align}
evaluated at $ s = 0 $. Note that $ Q = c_\gamma^{ - 1 } N |_{ s = 0 } $.

\subsection{Existence theorem for the Einstein equations}\label{ExistEinstein}
We now present the existence result for the Einstein equations.

\begin{prop}\label{propeinstein}
Suppose that there exist a time interval $[0,T]$ and positive constants $B_1$ and $B_2$ such that $f$ is defined on $[0,T]$ and satisfy
\begin{align}
\sup_{0 \leq s \leq T} \| f(s) \|_{\langle -1 \rangle} + \sup_{0 \leq s \leq T} \| f(s) \|_{\langle 1 \rangle} & \leq B_1, \\
\sup_{0 \leq s \leq T} \left\| \frac{\partial f}{\partial s} (s) \right\|_{\langle -1 \rangle} + \sup_{0 \leq s \leq T} \left\| \frac{\partial f}{\partial s} (s) \right\|_{\langle 1 \rangle} & \leq B_2.
\end{align}
Then, for any initial data $a_{0}, {b_0}, K_{0}, \hat{Z}_{0} \in S_3(\bbr)$ satisfying the Fuchsian conditions \eqref{FF1bis} and \eqref{FF3bis}, there exists $0 < T_B \leq T$ such that the Einstein equations \eqref{ODEx}--\eqref{ODE} have a unique solution $a_{ij}, b^{ij}, K_{ij}, \hat{Z}_{ij} \in C^1([0, T_B]; S_3(\bbr) )$ satisfying
\begin{align}
\sup_{ [0, T_B] } \max \left( \| a \| , \| b \| , \| K \| , \| \hat{Z} \| \right) \leq 2 \max \left( \| a_0 \| , \| b_0 \| , \| K_0 \| , \| \hat{Z}_0 \| \right), 
\end{align}
where $T_B$ depends on $B_1$ and $B_2$. 
\end{prop}

Proposition \ref{propeinstein} will be proved in Section \ref{sproofE} by using Theorem \ref{ODEthm}. We need to show that the differentiability conditions for $ F $, $ G $, $ N $ and $ H $ are satisfied, which will be given in Lemmas \ref{lem F}, \ref{lem N} and \ref{lem H}, together with the eigenvalue conditions for $ N $, which will be given in Section \ref{eigen}. These lemmas will be combined in Section \ref{sproofprop2} to complete the proof of Proposition \ref{propeinstein}.

\subsection{Proof of the theorem}\label{sproofE}
We will need to verify that certain derivatives of $F$, $G$, $H$ and $N$ involved in \eqref{IVP} are bounded. To compute the derivatives the following expression for the derivatives of symmetric matrices will be frequently used:
\begin{align*}
\frac{\partial a_{ij}}{\partial a_{st}} = \frac{1}{2^{\delta_{st}}} \left(\delta^s_i \delta^t_j + \delta^t_i \delta^s_j \right),
\end{align*}
where there is no summation over the indices on the right hand side. Note also that
\begin{align}\label{derivdet}
\frac{\partial \det a}{\partial a_{ij}} = \frac{2 \det a}{2^{\delta_{ij}}} b^{ij},
\end{align}
where there is also no summation over the indices on the right hand side.

\subsubsection{Differentiability conditions for $ F $ and $ G $}
We first observe that the functions $F$ and $G$ are just polynomials of $x$, $y$ and $s$.  Thus, $\partial F/\partial x$, $\partial F/ \partial y$, $\partial G / \partial x$ and $\partial G / \partial y$ are bounded as long as $x$, $y$ and $s$ are. Recall that 
\begin{align*}
F= \begin{pmatrix}
K_{ij}\\
-b^{im}b^{jn} K_{nm}
\end{pmatrix},
\qquad
x= \begin{pmatrix}
a_{ij}\\
b^{ij}
\end{pmatrix},
\qquad
y=\begin{pmatrix}
K_{ij}\\
\hat{Z}_{ij}
\end{pmatrix},
\end{align*}
and the derivatives $\partial F / \partial x$ and $\partial F / \partial y$ are explicitly given by
\begin{align*}
& \frac{\partial F}{\partial a_{st}} 
= \begin{pmatrix}
 0 \\
 0 
\end{pmatrix}, \\
& \frac{\partial F}{\partial b^{st}} 
= - \frac{1}{2^{\delta_{st}}}
\begin{pmatrix}
0 \\
b^{jn} (\delta^i_s K_{nt}+\delta^i_t K_{ns} )+ b^{im}(\delta^j_s K_{tm}+\delta^j_t K_{ms})
\end{pmatrix}, \allowdisplaybreaks \\
& \frac{\partial F}{\partial K_{st}} 
= \frac{1}{2^{\delta_{st}}}
\begin{pmatrix}
\delta_i^s \delta_j^t+\delta_i^t \delta^s_j  \\
-b^{it} b^{js}-b^{is}b^{jt} 
\end{pmatrix}, \\
& \frac{\partial F}{\partial \hat{Z}_{st}}
= \begin{pmatrix}
 0 \\
 0 
\end{pmatrix}.
\end{align*}
Similarly, we write $G = G_{ij}$ by abuse of notation, where
\begin{align*}
G_{ij} =K_{im}b^{mn}K_{nj}-\frac12 b^{mn}K_{mn} K_{ij} - 2 c_\gamma^{2-2 c_\gamma } s^{2 c_\gamma -2}R_{ij}+2\Lambda c_\gamma^{2-4 c_\gamma } s^{4 c_\gamma -2} a_{ij},
\end{align*}
and the derivatives $\partial G / \partial x$ and $\partial G / \partial y$ are given by
\begin{align*}
& \frac{\partial G}{\partial a_{st}} = \frac{1}{2^{\delta_{st}}}
\left\{ - \frac14 c_\gamma^{2-2 c_\gamma } s^{2 c_\gamma -2} R_A+ 2 \Lambda c_\gamma^{2-4 c_\gamma } s^{4 c_\gamma -2} \left(\delta^s_i\delta^t_j+\delta^s_j \delta^t_i\right) \right\}, \\
& \frac{\partial G}{\partial b^{st}} = \frac{1}{2^{\delta_{st}}}
\left\{  K _ { i s} K _ { j t }+K_{it}K_{js} - K _ { i j } K _ { st}  - \frac14 c_\gamma^{2-2 c_\gamma } s^{2 c_\gamma -2}R_{B} \right\}, \\
& \frac{\partial G}{\partial K_{st}} = \frac{1}{2^{\delta_{st}}}
\left\{ ( \delta^s_i b^{tn}+\delta^t_i b^{sn}) K_{jn} + (\delta_j^s b^{mt}+\delta^t_j b^{ms})K_{im} -\frac12b^{mn} K_{mn} (\delta^s_i\delta^t_j+\delta^s_j \delta^t_i)-K_{ij}b^{st} \right\}, \\
& \frac{\partial G}{\partial \hat{Z}_{st}} = 0.
\end{align*}
Here, $R_A$ and $R_B$ are given by
\begin{align*}
R_A & = {} -4C^k_{\;\;ck}(C^t_{\;\;lj}\delta^s_{i}+C^s_{\;\;lj}\delta^t_i+C^t_{\;\;li}\delta^s_{j}+C^s_{\;\;li}\delta^t_j)b^{cl} -4(C^s_{\;\;ki} C^t_{\;\;lj}+C^t_{\;\;ki} C^s_{\;\;lj})b^{kl} \nonumber \\
& \quad -2 \left\{(C^t_{\;\;kp}\delta^s_{j}+C^s_{\;\;kp}\delta^t_j)C^r_{\;\;cl} a_{ir}+(C^t_{\;\;cl}\delta_i^s+C^s_{\;\;cl}\delta_i^t) C^m_{\;\;kp} a_{jm} \right\}b^{kl}b^{pc},\\
R_B & = {} -4 \left\{C^k_{\;\;sk}(C^r_{\;\;tj}a_{ir}+C^r_{\;\;ti}a_{jr}) + C^k_{\;\;tk} (C^r_{\;\;sj}a_{ir}+C^r_{\;\;si}a_{jr})\right\} -4(C^c_{\;\;si} C^m_{\;\;tj}+C^c_{\;\;ti} C^m_{\;\;sj} )a_{cm} \nonumber \\
& \quad - 2\left\{ (C^m_{\;\;sp}C^r_{\;\;ct}+ C^m_{\;\;tp}C^r_{\;\;cs})b^{pc} + ( C^m_{\;\;ks}C^r_{\;\;tl}+C^m_{\;\;kt}C^r_{\;\;sl}) b^{kl} \right\}a_{jm}a_{ir},
\end{align*}
which are obtained using \eqref{ricci} such that $\partial (8R_{ij})/\partial a_{st} = R_A /2^{\delta_{st}}$ and $\partial (8R_{ij})/\partial b^{st} = R_B /2^{\delta_{st}}$. We obtain the following lemma.

\begin{lemma}\label{lem F}
The functions $F$ and $G$ in \eqref{x} are differentiable with respect to $x$ and $y$ and satisfy for any $T>0$, 
\begin{align}
\| F \| + \left\| \frac{\partial F}{\partial x} \right\| + \left\| \frac{\partial F}{\partial y} \right\| & \leq C( 1 + \| b \| + \| K \|)^3, \\
\sup_{0\leq s \leq T} \| G (s) \| + \sup_{0 \leq s \leq T}\left\| \frac{\partial G}{\partial x} (s) \right\| + \left\| \frac{\partial G}{\partial y} \right\| & \leq C (1 + T)^{4 c_\gamma -2} (1 + \| a \| + \| b \| + \| K \|)^4,
\end{align}
where the constants $C$ are independent of $T$.
\end{lemma}

\begin{proof}
The lemma follows from the above computations together with the formula \eqref{ricci}. 
\end{proof}

\subsubsection{Differentiability condition for $ N $}\label{sDN}
We recall that the matrix $N$ is given by
\begin{align*}
N= c_\gamma 
\begin{pmatrix}
\gamma \delta^m_i \delta^n_j + 3 \pi_{ij}^{mn} & -2\delta^m_i \delta^n_j \\ 
\delta^m_i \delta^n_j - \chi_{ij}^{mn} + \Pi_{ij}^{mn} & c_\gamma^{-1}\delta^m_i \delta^n_j 
\end{pmatrix}.
\end{align*}
We need to check $\partial N / \partial x$, $\partial N / \partial s$, $\partial^2 N / \partial x \partial s$ and $\partial^2 N / \partial x^2$, but may only consider the non-trivial components, which are $\pi_{ij}^{mn}$, $\chi_{ij}^{mn}$ and $\Pi_{ij}^{mn}$. First, we have for $ 3 \pi_{ij}^{mn} = a_{ij} b^{mn}$ the following derivatives:
\begin{align*}
& \frac{\partial (a_{ij} b^{mn})}{\partial a_{st}}
= \frac{1}{2^{\delta_{st}}} (\delta^s_i \delta^t_j + \delta^s_j \delta^t_i) b^{mn}, \\
& \frac{\partial (a_{ij} b^{mn})}{\partial b^{st}}
= \frac{1}{2^{\delta_{st}}} 
a_{ij} (\delta^m_s \delta^n_t+\delta_s^n \delta^m_t ),\\
& \frac{\partial^2 (a_{ij} b^{mn})}{\partial a_{uv} \partial a_{st}}
= 0, \allowdisplaybreaks \\
& \frac{\partial^2 (a_{ij} b^{mn})}{\partial b^{uv} \partial a_{st}}
= \frac{1}{2^{\delta_{st}} 2^{\delta_{uv}} } 
(\delta^s_i \delta^t_j + \delta^s_j \delta^t_i)  (\delta^m_u \delta^n_v+ \delta^m_v \delta^n_u), \\
& \frac{\partial^2 (a_{ij} b^{mn})}{\partial a_{uv} \partial b^{st}}
= \frac{1}{2^{\delta_{st}} 2^{\delta_{uv}} } 
(\delta^m_s \delta^n_t+\delta_s^n \delta^m_t )  (\delta_i^u \delta_j^v + \delta_i^v \delta_j^u),\\
& \frac{\partial^2 (a_{ij} b^{mn})}{\partial b^{uv} \partial b^{st}}
= 0.
\end{align*}
Note that we don't need to consider the $s$-derivatives of $\pi_{ij}^{mn}$. Next, for $\chi_{ij}^{mn}$ we note that
\begin{align*}
\frac{\partial}{\partial a_{st}}\left(\frac{1}{\sqrt{ \det a}}\right) & = -\frac{b^{st}}{2^{\delta_{st}}\sqrt{ \det a}},\\
\frac{\partial}{\partial b^{st}}\left(\frac{1}{(p^0)^r }\right) & = -\frac{ r }{2^{\delta_{st}}}\frac{p_s p_t}{(p^0)^{ r +2}},
\end{align*}
so that we have
\begin{align*}
\frac{\partial (\Psi_{i_1\cdots i_n})}{\partial a_{st}} & = -\frac{b^{st}}{2^{\delta_{st}}} \Psi_{i_1\cdots i_n},\\
\frac{\partial (\Psi_{i_1\cdots i_n})}{\partial b^{st}} & = -\frac{(n-1)}{2^{\delta_{st}}}\Psi_{i_1\cdots i_n s t}.
\end{align*}
Hence, we have for $\chi_{ij}^{mn} = \Psi_{ijkl} b^{km} b^{ln}$ the following derivatives:
\begin{align*}
&\frac{\partial (\Psi_{ijkl} b^{km} b^{ln})}{\partial a_{st}} = -\frac{b^{st}}{2^{\delta_{st}}} \Psi_{ijkl} b^{km} b^{ln},\\
&\frac{\partial (\Psi_{ijkl} b^{km} b^{ln})}{\partial b^{st}} = -\frac{3}{2^{\delta_{st}}} \Psi_{ijklst}b^{km}b^{ln}+\frac{1}{2^{\delta_{st}}} ( \Psi_{ijsl} \delta^m_t b^{ln} + \Psi_{ijtl} \delta^m_s b^{ln} + \Psi_{ijks} \delta^n_t b^{km} + \Psi_{ijkt} \delta^n_s b^{km} ), \allowdisplaybreaks\\
&\frac{\partial (\Psi_{ijkl} b^{km} b^{ln})}{\partial s} = \Phi_{ijkl}  b^{km} b^{ln},\\
&\frac{\partial^2 (\Psi_{ijkl} b^{km} b^{ln})}{\partial a_{st}\partial s} = -\frac{b^{st}}{2^{\delta_{st}}} \Phi_{ijkl} b^{km} b^{ln},\\
&\frac{\partial^2 (\Psi_{ijkl} b^{km} b^{ln})}{\partial b^{st}\partial s} = -\frac{3}{2^{\delta_{st}}} \Phi_{ijklst}b^{km}b^{ln}+\frac{1}{2^{\delta_{st}}} ( \Phi_{ijsl} \delta^m_t b^{ln} + \Phi_{ijtl} \delta^m_s b^{ln} + \Phi_{ijks} \delta^n_t b^{km} + \Phi_{ijkt} \delta^n_s b^{km} ).
\end{align*}
and
\begin{align*}
\frac{\partial^2 \chi_{ij}^{mn}}{\partial x^2} 
= \begin{pmatrix}
\partial^2\chi_{ij}^{mn} / \partial a_{uv} \partial a_{st} & \partial^2\chi_{ij}^{mn} / \partial b^{uv} \partial a_{st} \\
\partial^2\chi_{ij}^{mn} / \partial a_{uv} \partial b^{st} & \partial^2\chi_{ij}^{mn} / \partial b^{uv} \partial b^{st}
\end{pmatrix}
=  \frac{1}{2^{\delta_{st}} 2^{\delta_{uv}}} 
\begin{pmatrix} 
\chi_{xx11} & \chi_{xx12}\\
\chi_{xx21}& \chi_{xx22} 
\end{pmatrix},
\end{align*}
where
\begin{align*}
\chi_{xx11} & =  b^{st} b^{uv}\Psi_{ijkl} b^{km} b^{ln},\\
\chi_{xx12} & = -\Psi_{ijkl}\left\{ (\delta^s_u \delta^t_v+\delta^s_v \delta^t_u) b^{km} b^{ln}+ b^{st}(\delta^k_u \delta^m_v+\delta^k_v \delta^m_u) b^{ln}+ b^{st} b^{km}(\delta^l_u \delta^n_v+\delta^l_v \delta^n_u)\right\} \nonumber \\
&\quad +3 \Psi_{ijkluv}b^{st}b^{km} b^{ln}, \\
\chi_{xx21} & = 3 b^{uv} \Psi_{ijklst}b^{km}b^{ln} - b^{uv} \left\{ \Psi_{ijsl} \delta^m_t b^{ln} + \Psi_{ijtl} \delta^m_s b^{ln} + \Psi_{ijks} \delta^n_t b^{km} + \Psi_{ijkt} \delta^n_s b^{km} \right\},\\
\chi_{xx22} & = 15 \Psi_{ijklstuv} b^{km}b^{ln}- 3 \Psi_{ijklst}\left\{ (\delta^k_u \delta^m_v+\delta^k_v \delta^m_u) b^{ln}+ b^{km} (\delta^l_u \delta^n_v+\delta^l_v \delta^n_u )\right\} \nonumber \\
&\quad - 3 \Psi_{ijkluv} \left\{ (\delta^k_s \delta^m_t+\delta^k_t \delta^m_s) b^{ln}+ b^{km} (\delta^l_s \delta^n_t+\delta^l_t \delta^n_s)\right\} \nonumber \\
&\quad +\Psi_{ijkl}  \left\{ (\delta^k_s \delta^m_t+\delta^k_t \delta^m_s ) (\delta^l_u \delta^n_v+\delta^l_v \delta^n_u ) + (\delta^l_s \delta^n_t+\delta^l_t \delta^n_s )(\delta^k_u \delta^m_v+\delta^k_v \delta^m_u ) \right\}.
\end{align*}
Finally, for $\Pi_{ij}^{mn} = \Psi_{ij} b^{mn}$ we obtain
\begin{align*}
& \frac{\partial (\Psi_{ij} b^{mn})}{\partial a_{st}} = -\frac{1}{2^{\delta_{st}}} b^{st} \Psi_{ij} b^{mn}, \\
& \frac{\partial (\Psi_{ij} b^{mn})}{\partial b^{st}} = \frac{1}{2^{\delta_{st}}} \left\{ -\Psi_{ijst} b^{mn} + \Psi_{ij} (\delta^m_s \delta^n_t + \delta^n_s \delta^m_t ) \right\}, \\
& \frac{\partial (\Psi_{ij} b^{mn})}{\partial s} = \Phi_{ij} b^{mn},\\
& \frac{\partial^2 (\Psi_{ij} b^{mn})}{\partial a_{st} \partial s} = -\frac{1}{2^{\delta_{st}}} b^{st} \Phi_{ij} b^{mn}, \\
& \frac{\partial^2 (\Psi_{ij} b^{mn})}{\partial b^{st} \partial s} = \frac{1}{2^{\delta_{st}}} \left\{ -\Phi_{ijst} b^{mn} + \Phi_{ij} (\delta^m_s \delta^n_t + \delta^n_s \delta^m_t ) \right\},
\end{align*}
and
\begin{align*}
\frac{\partial^2 (\Psi_{ij} b^{mn})}{\partial x^2} 
= \begin{pmatrix}
\partial^2(\Psi_{ij} b^{mn}) / \partial a_{uv} \partial a_{st} & \partial^2(\Psi_{ij} b^{mn}) / \partial b^{uv} \partial a_{st} \\
\partial^2(\Psi_{ij} b^{mn}) / \partial a_{uv} \partial b^{st} & \partial^2(\Psi_{ij} b^{mn}) / \partial b^{uv} \partial b^{st}
\end{pmatrix}
=  \frac{1}{2^{\delta_{st}} 2^{\delta_{uv}}} 
\begin{pmatrix} 
\Pi_{xx11} & \Pi_{xx12}\\
\Pi_{xx21}& \Pi_{xx22} 
\end{pmatrix},
\end{align*}
where
\begin{align*}
\Pi_{xx11}&=b^{st} b^{uv} \Psi_{ij} b^{mn},\\
\Pi_{xx12}&= - (\delta^s_u \delta^t_v + \delta^s_v \delta^t_u) \Psi_{ij} b^{mn}  + b^{st} \Psi_{ijuv} b^{mn} -b^{st} \Psi_{ij} (\delta^m_u \delta^n_v + \delta^n_u \delta^m_v ),\\
\Pi_{xx21}&=  b^{uv} \Psi_{ijst} b^{mn} - b^{uv} \Psi_{ij}  (\delta^m_s \delta^n_t + \delta^n_s \delta^m_t ),\\
\Pi_{xx22}&= 3 \Psi_{ijstuv} b^{mn} - \Psi_{ijst} (\delta^m_u \delta^n_v +\delta^m_v \delta^n_u) - \Psi_{ijuv} (\delta^m_s \delta^n_t + \delta^n_s \delta^m_t ).
\end{align*}
To estimate the above quantities we need the following lemma.

\begin{lemma}\label{lem e}
Let $ { e_j}^i $ and $ { e^i }_j $ be given by \eqref{ortho_e} and \eqref{ortho_theta}. Then, we have
\begin{align}
&|{ e_j}^i|\leq\frac{C\| b \|^{\frac72}}{\det b},\quad |\partial_{b^{st}}{e_j}^i | \leq \frac{C \| b \|^{\frac{17}{2}}}{(\det b)^3},\quad |\partial^2_{b^{uv}b^{st}} {e_j}^i | \leq \frac{C\|b\|^{\frac{27}{2}}}{(\det b)^5},\\
&| { e^i}_j|\leq \frac{C\| b\|^{\frac52}}{\det b},\quad |\partial_{b^{st}}{ e^i}_j | \leq \frac{C \| b \|^{\frac{15}{2}}}{(\det b)^3},\quad |\partial^2_{b^{uv}b^{st}} { e^i}_j | \leq \frac{C\|b\|^{\frac{25}{2}}}{(\det b)^5},
\end{align}
for any $i,j=1,2,3$. 
\end{lemma}
\begin{proof}
Since the matrix $ b^{ij} $ is symmetric and positive definite, we can apply Lemmas 1 and 2 of \cite{LN} to obtain
\begin{align}
\frac{1}{b^{11}}\leq\frac{b^{22} b^{33}}{\det b}, \qquad \frac{1}{b^{11}b^{22} - (b^{12})^2}\leq \frac{b^{33}}{\det b}.
\end{align}
Applying the above to the formulas \eqref{ortho_e} and \eqref{ortho_theta} we obtain the following:
\begin{align}
|{e_j}^i |\leq\frac{C\| b \|^{\frac72}}{\det b}, \qquad |{e^i}_j|\leq \frac{C\| b\|^{\frac52}}{\det b} ,
\end{align}
for any $i,j=1,2,3$. For the derivatives we note that
\begin{align}
{e_j}^i = \frac{P}{\sqrt{Q}},
\end{align}
where $P$ and $Q$ are polynomials of $b^{kl}$ of degree $j$ and $2j-1$, respectively. Here, we have
\begin{align}
\frac{1}{Q} \leq\frac{C\| b\|^{7-2j}}{(\det b)^2},
\end{align}
so that we obtain
\begin{align}
|\partial_{b^{st}} {e_j}^i | \leq \left|\frac{(\partial_{b^{st}}P) Q - P (\partial_{b^{st}}Q)/2}{Q^{3/2}}\right| \leq \frac{C \|b\|^{3j-2}}{Q^{3/2}} \leq \frac{C \| b \|^{\frac{17}{2}}}{(\det b)^3}.
\end{align}
For the second derivatives we note that
\begin{align}
\partial_{b^{st}} {e_j}^i = \frac{P}{Q^{3/2}},
\end{align}
where $P$ and $Q$ are polynomials of $b^{kl}$ of degree $3j-2$ and $2j-1$, respectively. Then, we obtain
\begin{align}
|\partial^2_{b^{uv}b^{st}} {e_j}^i | \leq \left| \frac{(\partial_{b^{uv}}P)Q - 3P(\partial_{b^{uv}}Q)/2}{Q^{5/2}}\right| \leq \frac{C \| b\|^{5j-4}}{Q^{5/2}} \leq \frac{C\|b\|^{\frac{27}{2}}}{(\det b)^5}.
\end{align}
In a similar way we can obtain the estimates for ${ e^i}_j$, and we skip the proof. 
\end{proof}

In order to estimate the tensors $ \Psi $ and $ \Phi $, which are defined by certain integrations of $ f $, we use the weighted $ L^1 $-norms $ \| \cdot \|_{ L^1_r } $ and $ \| \cdot \|_{ \langle r \rangle } $ for $ f $, where the weights are given by $ ( p^0 )^r $ and $ \langle p \rangle^r $, respectively. We first note that the weights can be estimated as follows:
\begin{align}
p^0 & = \sqrt{b^{ij} p_i p_j} \leq C \| b \|^{\frac{1}{2}} \langle p \rangle, \label{p^0 1} \\
\langle p \rangle & = \sqrt{ \delta^{ i j } p_i p_j } = \sqrt{\delta^{ij} \hat{p}_k { e^k }_i \hat{p}_l { e^l }_j } \leq \frac{C \| b \|^{\frac{5}{2}}}{\det b} p^0, \label{p^0 2}
\end{align}
where we used \eqref{phat} and Lemma \ref{lem e}. We now obtain the following lemma.

\begin{lemma}\label{lem N}
Suppose that there exist a time interval $[0,T]$ and positive constants $B_1$ and $B_2$ such that $f$ is defined on $[0,T]$ and satisfy
\begin{align}
&\sup_{0 \leq s \leq T} \| f(s) \|_{\langle 1 \rangle} \leq B_1, \\
&\sup_{0 \leq s \leq T} \left\| \frac{\partial f}{\partial s} (s) \right\|_{\langle 1 \rangle} \leq B_2.
\end{align}
Then, the coefficient function $N$ satisfies
\begin{align}
\sup_{0 \leq s \leq T} \| N (s) \| + \sup_{0 \leq s \leq T} \left\| \frac{\partial N}{\partial x}(s) \right\| + \sup_{0 \leq s \leq T} \left\| \frac{\partial^2 N}{\partial x^2} (s) \right\| & \leq C(a_{ij}, b^{ij}) (1 + B_1), \\
\sup_{0 \leq s \leq T} \left\| \frac{\partial N}{\partial s} (s) \right\| + \sup_{0 \leq s \leq T} \left\| \frac{\partial^2 N}{\partial x \partial s} (s) \right\| & \leq C (a_{ij}, b^{ij} ) ( 1 + B_2),
\end{align}
where $C(a_{ij}, b^{ij})$ are positive constants depending only on $a_{ij}$ and $b^{ij}$.
\end{lemma}
\begin{proof}
The quantities $\Psi_{i_1\cdots i_n}$ can be easily estimated as follows:
\begin{align*}
|\Psi_{i_1\cdots i_n}| \leq \frac{C }{\sqrt{ \det a}} \int_{\bbr^3} \frac{\langle p \rangle^n}{(p^0)^{n-1}} |f(s,p)|\, d^3 p \leq \frac{ C \| b \|^{\frac{5(n-1)}{2}} B_1 }{(\det a)^{1/2} (\det b)^{n-1}}.
\end{align*}
Then, we obtain the estimates for $N$, $\partial N / \partial x$ and $\partial^2 N / \partial x^2$ by collecting the above computations. For instance, we estimate the quantity $\Psi_{ijklstuv} b^{km}b^{ln}$, which appears in $\chi_{xx22}$, as follows:
\begin{align*}
|\Psi_{ijklstuv} b^{km}b^{ln}| & \leq \frac{ C \| b \|^{\frac{39}{2}} B_1 }{(\det a)^{1/2} (\det b)^{7}}.
\end{align*}
In a similar way we obtain the estimates of the quantities $\Phi_{i_1\cdots i_n}$, which are used for the estimates of $\partial N / \partial s$ and $\partial^2 N / \partial x \partial s$, and this completes the proof. 
\end{proof}

\subsubsection{Differentiability condition for $ H $}\label{sDH}
We need to estimate $\partial H/\partial s$, $\partial H/\partial x$, $\partial^2 H/\partial x\partial s$ and $\partial^2 H/\partial x^2$. Let us write $H = H_{ij}$ by abuse of notation again. We observe from \eqref{hp^0q^0}--\eqref{q'} with \eqref{ortho_e} and \eqref{ortho_theta} that $h$, $p^0$, $q^0$, $p'_i$, $q'_i$, $p'^0$ and $q'^0$ are functions of $b^{ij}$ so that we may write \eqref{H} as
\begin{align}
H_ { i j } = \frac { 8 \pi c_\gamma } { \det a }  \iiint W_{ij}(p_k,q_k,\omega_k,b^{kl}) f(p) f(q) \, d \omega \, d ^ 3 q \, d ^ 3 p,
\end{align}
where
\begin{equation}
W_{ij}(p_k,q_k,\omega_k,b^{kl}) = \frac{h ^ { 2 - \gamma }}{p^0 q^0} \left( \frac{p' _ i p' _ j}{p'^0} - \frac{p_i p_j}{p^0}\right).
\end{equation}
The derivatives with respect to $s$ and $a_{st}$ are now easily obtained:
\begin{align}
\frac{\partial H_{ij}}{\partial s} & = \frac { 8 \pi c_\gamma } { \det a }  \iiint W_{ij} \left(\frac{\partial f}{\partial s}(p) f(q) + f(p)\frac{\partial f}{\partial s}(q)\right) \, d \omega \, d ^ 3 q \, d ^ 3 p,\label{H1}\\
\frac{\partial H_ { i j } }{\partial a_{st}} & = \left(-\frac{2b^{st}}{2^{\delta_{st}}}\right) \frac{8 \pi c_\gamma }{ \det a}  \iiint W_{ij} f(p) f(q) \, d \omega \, d ^ 3 q \, d ^ 3 p,\label{H2}\\
\frac{\partial^2 H_{ij}}{\partial s \partial a_{st}} & = \left(-\frac{2b^{uv}}{2^{\delta_{uv}}}\right) \frac{8 \pi c_\gamma }{ \det a} \iiint W_{ij} \left(\frac{\partial f}{\partial s}(p) f(q) + f(p)\frac{\partial f}{\partial s}(q)\right) \, d \omega \, d ^ 3 q \, d ^ 3 p,\label{H3}\\
\frac{\partial^2 H_{ij}}{\partial a_{uv} \partial a_{st}} & = \left(-\frac{2b^{st}}{2^{\delta_{st}}}\right)\left(-\frac{2b^{uv}}{2^{\delta_{uv}}}\right) \frac{8 \pi c_\gamma }{ \det a} \iiint W_{ij} f(p) f(q)\, d \omega \, d ^ 3 q \, d ^ 3 p,\label{H4}
\end{align}
where we used \eqref{derivdet} for the derivatives with respect to $a_{st}$. For the derivatives with respect to $b^{st}$ we only need to compute
\begin{align*}
\frac{\partial W_{ij}}{\partial b^{st}},\qquad \frac{\partial^2 W_{ij}}{\partial b^{uv} \partial b^{st}}.
\end{align*}
We first consider the derivatives of $1/p^0$ and $1/q^0$. Since
\begin{equation}
\frac{\partial p^0}{\partial b^{st}} = \frac{1}{2^{\delta_{st}}}\frac{p_s p_t}{p^0},
\end{equation}
we have
\begin{align}
\frac{\partial}{\partial b^{st}}\left(\frac{1}{p^0}\right) & = -\frac{1}{2^{\delta_{st}}}\frac{p_s p_t}{(p^0)^3}, \label{W1}\\
\frac{\partial^2}{\partial b^{uv} \partial b^{st}} \left(\frac{1}{p^0}\right) & = \frac{3}{2^{\delta_{st}}2^{\delta_{uv}}}\frac{p_s p_t p_u p_v}{(p^0)^5}, \label{W2} \\
\frac{\partial}{\partial b^{st}}\left(\frac{1}{q^0}\right) & = -\frac{1}{2^{\delta_{st}}}\frac{q_s q_t}{(q^0)^3}, \label{W3}\\
\frac{\partial^2}{\partial b^{uv} \partial b^{st}} \left(\frac{1}{q^0}\right) & = \frac{3}{2^{\delta_{st}}2^{\delta_{uv}}}\frac{q_s q_t q_u q_v}{(q^0)^5}. \label{W4}
\end{align}
Next, we recall that
\begin{align*}
h & = \sqrt{2p^0q^0 - 2b^{ij} p_i q_j},
\end{align*}
and
\begin{align*}
\frac{\partial b^{ij}}{\partial b^{st}} & = \frac{1}{2^{\delta_{st}}}(\delta^i_s \delta^j_t + \delta^i_t \delta^j_s).
\end{align*}
Then, we have
\begin{align}
\frac{\partial h}{\partial b^{st}}
& = \frac{1}{2h}\left(2\frac{\partial p^0}{\partial b^{st}} q^0 + 2p^0 \frac{\partial q^0}{\partial b^{st}} - 2\frac{\partial b^{ij}}{\partial b^{st}} p_i q_j\right) \nonumber\\
& = \frac{1}{h}\left(\frac{1}{2^{\delta_{st}}}\right)\left(\frac{p_s p_t}{p^0}q^0 + p^0\frac{q_s q_t}{q^0} - p_sq_t - p_t q_s\right) \nonumber\\
& = \frac{p^0 q^0}{2^{\delta_{st}}h}\left(\frac{p_s}{p^0} - \frac{q_s}{q^0}\right)\left(\frac{p_t}{p^0} - \frac{q_t}{q^0}\right). \label{W5}
\end{align}
Similarly, we have
\begin{align}
\frac{\partial^2 h}{\partial b^{uv} \partial b^{st}} & = \frac{1}{2^{\delta_{st}}h}\left(\frac{1}{2^{\delta_{uv}}}\right)\left(\frac{p_u p_v}{p^0}q^0 + p^0\frac{q_u q_v}{q^0}\right)\left(\frac{p_s}{p^0} - \frac{q_s}{q^0}\right)\left(\frac{p_t}{p^0} - \frac{q_t}{q^0}\right) \nonumber\\
& \quad + \frac{p^0 q^0}{2^{\delta_{st}}}\left(-\frac{1}{h^{2}}\right)\frac{p^0 q^0}{2^{\delta_{uv}}h}\left(\frac{p_u}{p^0} - \frac{q_u}{q^0}\right)\left(\frac{p_v}{p^0} - \frac{q_v}{q^0}\right)\left(\frac{p_s}{p^0} - \frac{q_s}{q^0}\right)\left(\frac{p_t}{p^0} - \frac{q_t}{q^0}\right) \nonumber \\
& \quad + \frac{p^0 q^0}{2^{\delta_{st}}h}\left(-\frac{1}{2^{\delta_{uv}}}\right)\left(\frac{p_s p_u p_v}{(p^0)^3} - \frac{q_s q_u q_v}{(q^0)^3}\right)\left(\frac{p_t}{p^0} - \frac{q_t}{q^0}\right) \nonumber \\
& \quad + \frac{p^0 q^0}{2^{\delta_{st}}h}\left(-\frac{1}{2^{\delta_{uv}}}\right)\left(\frac{p_s}{p^0} - \frac{q_s}{q^0}\right)\left(\frac{p_tp_u p_v}{(p^0)^3} - \frac{q_tq_u q_v}{(q^0)^3}\right). \label{W6}
\end{align}
For the quantity $h^{2- \gamma }$ we obtain
\begin{align}
\frac{\partial h^{2- \gamma }}{\partial b^{st}} & = \frac{(2- \gamma )p^0 q^0}{2^{\delta_{st}}h^b}\left(\frac{p_s}{p^0} - \frac{q_s}{q^0}\right)\left(\frac{p_t}{p^0} - \frac{q_t}{q^0}\right), \label{W7} \\
\frac{\partial^2 h^{2- \gamma }}{\partial b^{uv} \partial b^{st}} & = \frac{(2- \gamma )}{2^{\delta_{st}}h^b}\left(\frac{1}{2^{\delta_{uv}}}\right)\left(\frac{p_u p_v}{p^0}q^0 + p^0\frac{q_u q_v}{q^0}\right)\left(\frac{p_s}{p^0} - \frac{q_s}{q^0}\right)\left(\frac{p_t}{p^0} - \frac{q_t}{q^0}\right) \nonumber \\
& \quad + \frac{(2- \gamma )p^0 q^0}{2^{\delta_{st}}}\left(\frac{-b}{h^{1+b}}\right)\frac{p^0 q^0}{2^{\delta_{uv}}h}\left(\frac{p_u}{p^0} - \frac{q_u}{q^0}\right)\left(\frac{p_v}{p^0} - \frac{q_v}{q^0}\right)\left(\frac{p_s}{p^0} - \frac{q_s}{q^0}\right)\left(\frac{p_t}{p^0} - \frac{q_t}{q^0}\right) \nonumber \\
& \quad + \frac{(2- \gamma )p^0 q^0}{2^{\delta_{st}}h^b}\left(-\frac{1}{2^{\delta_{uv}}}\right)\left(\frac{p_s p_u p_v}{(p^0)^3} - \frac{q_s q_u q_v}{(q^0)^3}\right)\left(\frac{p_t}{p^0} - \frac{q_t}{q^0}\right) \nonumber \\
& \quad + \frac{(2- \gamma )p^0 q^0}{2^{\delta_{st}}h^b}\left(-\frac{1}{2^{\delta_{uv}}}\right)\left(\frac{p_s}{p^0} - \frac{q_s}{q^0}\right)\left(\frac{p_tp_u p_v}{(p^0)^3} - \frac{q_tq_u q_v}{(q^0)^3}\right). \label{W8}
\end{align}
Now, we consider the derivatives of post-collision momenta. Recall that 
\begin{equation}
p'^0 = \frac{1}{2}(n^0 + n_k \omega_j \delta^{ i j } { e_i }^k ),
\end{equation}
where $n^0 = p^0 + q^0$ and $n_i = p_i + q_i$, so that we have
\begin{align}
\frac{\partial p'^0}{\partial b^{st}} & = \frac{1}{2}\left(\frac{1}{2^{\delta_{st}}}\frac{p_s p_t}{p^0} + \frac{1}{2^{\delta_{st}}}\frac{q_s q_t}{q^0} + n_k \omega_j \delta^{ i j } (\partial_{b^{st}}{ e_i }^k) \right), \label{W9} \\
\frac{\partial^2 p'^0}{\partial b^{uv} \partial b^{st}} & = \frac{1}{2}\left(-\frac{1}{2^{\delta_{st}}2^{\delta_{uv}}}\frac{p_s p_t p_u p_v}{(p^0)^3} - \frac{1}{2^{\delta_{st}}2^{\delta_{uv}}}\frac{q_s q_t q_u q_v}{(q^0)^3} + n_k \omega_j \delta^{ i j } (\partial^2_{ b^{ u v } b^{st}}{ e_i }^k) \right). \label{W10}
\end{align}
For the spatial components we recall that
\begin{equation}
p'_i = \frac{1}{2}\left(n_i + h \omega_j { e^j }_i + \frac{ n_l \omega_k \delta^{ j k } { e_j }^l n_i }{ n^0 + h}\right) .
\end{equation}
Then, we compute
\begin{align}
\frac{\partial p'_i}{\partial b^{st}} & = \frac{1}{2}\bigg\{ \frac{\partial h}{\partial b^{st}} \omega_j { e^j }_i + h\omega_j (\partial_{b^{st}} { e^j }_i ) + \frac{n_l \omega_k \delta^{ j k } (\partial_{b^{st}} { e_j }^l ) n_i }{ n^0 + h} - \frac{n_l \omega_k \delta^{ j k } { e_j }^l n_i }{ (n^0 + h)^2}\left(\frac{\partial n^0}{\partial b^{st}} + \frac{\partial h}{\partial b^{st}}\right)\bigg\}, \label{W11}
\end{align}
and
\begin{align}
& \frac{\partial^2 p'_i}{\partial b^{uv} \partial b^{st}} = \frac{1}{2}\bigg\{\frac{\partial^2 h}{\partial b^{uv} \partial b^{st}} \omega_j { e^j }_i + \frac{\partial h}{\partial b^{st}}\omega_j (\partial_{b^{uv}} { e^j }_i) + \frac{\partial h}{\partial b^{uv}}\omega_j (\partial_{b^{st}} { e^j }_i) + h\omega_j (\partial^2_{b^{uv} b^{st}} { e^j }_i) \nonumber \\
& + \frac{n_l \omega_k \delta^{ j k } (\partial^2_{ b^{ u v } b^{st}} { e_j }^l ) n_i }{ n^0 + h} - \frac{n_l \omega_k \delta^{ j k } (\partial_{b^{st}} { e_j }^l ) n_i }{ (n^0 + h)^2}\left(\frac{\partial n^0}{\partial b^{uv}} + \frac{\partial h}{\partial b^{uv}}\right) - \frac{n_l \omega_k \delta^{ j k } (\partial_{b^{u v }} { e_j }^l ) n_i }{ (n^0 + h)^2}\left(\frac{\partial n^0}{\partial b^{st}} + \frac{\partial h}{\partial b^{st}}\right) \nonumber \\
& + \frac{2n_l \omega_k \delta^{ j k } { e_j }^l n_i }{ (n^0 + h)^3}\left(\frac{\partial n^0}{\partial b^{uv}} + \frac{\partial h}{\partial b^{uv}}\right)\left(\frac{\partial n^0}{\partial b^{st}} + \frac{\partial h}{\partial b^{st}}\right) - \frac{n_l \omega_k \delta^{ j k } { e_j }^l n_i }{ (n^0 + h)^2}\left(\frac{\partial^2 n^0}{\partial b^{uv}\partial b^{st}} + \frac{\partial^2 h}{\partial b^{uv} \partial b^{st}}\right)\bigg\}. \label{W12}
\end{align}
The derivatives of $ { e_j }^i $ and $ { e^i }_j $ are estimated by Lemma \ref{lem e}. We are now ready to prove the following lemma.

\begin{lemma}\label{lem W}
Let $W_{ij}$ be the quantity defined by
\begin{equation}
W_{ij}(p_k,q_k,\omega_k,b^{kl}) = \frac{h ^ { 2 - \gamma }}{p^0 q^0} \left( \frac{p' _ i p' _ j}{p'^0} - \frac{p_i p_j}{p^0}\right).
\end{equation}
Then, the derivatives of $W_{ij}$ with respect to $b^{st}$ are estimated as follows:
\begin{align}
| W_{ij}| + \left| \frac{\partial  W_{ij} }{\partial b^{st}} \right| & \leq C( a_{ij}, b^{ij} ) \left( \frac{ (p^0)^{ 1 - \frac{ \gamma }{2} } }{ (q^0)^{\frac{ \gamma }{2}}} + \frac{ (q^0)^{1-\frac{ \gamma }{2}}}{ (p^0)^{\frac{ \gamma }{2}}} \right), \label{lem W1}\\
\int_{\bbs^2} \left| \frac{\partial^2 W_{ij} }{ \partial b^{uv} \partial b^{st}} \right| d\omega & \leq C( a_{ij}, b^{ij} ) \left( \frac{ (p^0)^{ 2- \gamma } }{ q^0 } + \frac{ (p^0)^{ 1 - \frac{ \gamma }{2} } }{ (q^0)^{\frac{ \gamma }{2}}} + \frac{ (q^0)^{1-\frac{ \gamma }{2}}}{ (p^0)^{\frac{ \gamma }{2}}} + \frac{ (q^0)^{ 2- \gamma }}{ p^0 } \right), \label{lem W2}
\end{align}
where $C(a_{ij}, b^{ij})$ are positive constants depending only on $a_{ij}$ and $b^{ij}$. 
\end{lemma}
\begin{proof}
Let us first consider the quantity $h^{2- \gamma }/(p^0 q^0)$. We use the relation $p_i = \hat{p}_j { e^j }_i$ in \eqref{phat} to estimate the derivatives of $1/p^0$ and $1/q^0$ in \eqref{W1}--\eqref{W4}:
\begin{align}
\left| \frac{\partial}{\partial b^{st}}\left(\frac{1}{p^0}\right) \right|  & \leq C \| e^{ - 1 } \|^2 \frac{ 1 }{p^0},\qquad \left| \frac{\partial^2}{\partial b^{uv} \partial b^{st}} \left(\frac{1}{p^0}\right) \right| \leq C \| e^{ - 1 } \|^4 \frac{ 1 }{p^0}, \label{Wconti1} \\
\left| \frac{\partial}{\partial b^{st}}\left(\frac{1}{q^0}\right) \right|  & \leq C \| e^{ - 1 } \|^2 \frac{ 1 }{q^0},\qquad \left| \frac{\partial^2}{\partial b^{uv} \partial b^{st}} \left(\frac{1}{q^0}\right) \right| \leq C \| e^{ - 1 } \|^4 \frac{ 1 }{q^0} , \label{Wconti2}
\end{align}
where $ e^{ - 1 } $ denotes the $ 3 \times 3 $ matrix $ { e^i }_j $. For the derivatives of $h^{2- \gamma }$ we note that in an orthonormal frame the angle $\phi$ between $\p$ and $\q$ is defined by 
\begin{equation}
\p \cdot \q = |\p| |\q| \cos \phi = p^0 q^0 \cos\phi,
\end{equation}
where $\p \cdot \q$ is the usual inner product in three dimensions. Hence, we have
\begin{align}
\left| \frac{\p_i}{p^0} - \frac{\q_i}{q^0} \right| \leq 2 \sin\frac{\phi}{2} = \frac{ h }{\sqrt{p^0 q^0}}, \label{Wsin}
\end{align}
where we used \eqref{l1} of Lemma \ref{lem h} below. We obtain from \eqref{W7} and \eqref{W8} the following estimates:
\begin{align}
\left| \frac{\partial h^{2- \gamma }}{\partial b^{st}} \right| \leq C \| e^{ - 1 } \|^2 h^{2- \gamma }, \qquad \left| \frac{\partial^2 h^{2- \gamma }}{\partial b^{uv} \partial b^{st}} \right|  \leq C \| e^{ - 1 } \|^4 h^{2- \gamma }. \label{Wconti3}
\end{align}
For later use, we consider \eqref{W5} to obtain the following in a similar way: 
\begin{align}
\left| \frac{ \partial h }{ \partial b^{ s t } } \right| \leq C \| e^{ - 1 } \|^2 h . \label{Wconti4}
\end{align}
We now apply Lemma \ref{lem e} and conclude that the derivatives of the quantity $h^{2- \gamma } / (p^0 q^0)$ can be estimated as follows:
\begin{align}
\left| \frac{h^{2- \gamma }}{p^0 q^0} \right| + \left| \frac{\partial}{\partial b^{st}} \left(\frac{h^{2- \gamma }}{p^0 q^0} \right) \right| + \left| \frac{\partial^2}{\partial  b^{uv} \partial b^{st}} \left(\frac{h^{2- \gamma }}{p^0 q^0} \right) \right| \leq  \frac{ C (a_{ij}, b^{ij}) }{(p^0)^{\frac{ \gamma }{2}} (q^0)^{\frac{ \gamma }{2}}}, \label{W13}
\end{align}
where $C(a_{ij}, b^{ij})$ is a positive constant depending only on $a_{ij}$ and $b^{ij}$.

Next, we consider the quantity $p'_i p'_j / p'^0$. Applying Lemma \ref{lem e}, we obtain from \eqref{W9} and \eqref{W10} the following estimates of the derivatives of $p'^0$:
\begin{align*}
| p'^0| + \left| \frac{\partial p'^0}{\partial b^{st}} \right| + \left| \frac{\partial^2 p'^0}{\partial b^{uv} \partial b^{st}} \right| \leq C(a_{ij}, b^{ij}) n^0,
\end{align*}
and from \eqref{W11} and \eqref{W12} with \eqref{W5} and \eqref{W6} the following estimates of $p'_i$:
\begin{align*}
| p_i'| + \left| \frac{\partial p_i'}{\partial b^{st}} \right| + \left| \frac{\partial^2 p_i'}{\partial b^{uv} \partial b^{st}} \right| \leq C(a_{ij}, b^{ij}) n^0.
\end{align*}
Then, the following is easily obtained:
\begin{align}
\left| \frac{p' _ i p' _ j}{p'^0} \right| \leq C \| e^{ - 1 } \|^2 p'^0 \leq C(a_{ij}, b^{ij}) n^0, \label{W14}
\end{align}
since $ p'_i = \p'_j { e^j }_i$ and $p'^0 \leq n^0$ by \eqref{conservation}. For the first derivatives we compute 
\begin{align*}
\frac{\partial}{\partial b^{st}} \left( \frac{p' _ i p' _ j}{p'^0} \right) = \frac{(\partial_{b^{st}} p'_i) p'_j}{ p'^0 } + \frac{ p'_i (\partial_{b^{st}} p'_j)}{ p'^0 } - \frac{p'_i p'_j (\partial_{b^{st}} p'^0)}{ (p'^0)^2 },
\end{align*}
and notice that the quantities of the form $p'_i / p'^0$ are bounded by $C(a_{ij}, b^{ij})$. Hence, we obtain
\begin{align}
\left| \frac{\partial}{\partial b^{st}} \left( \frac{p' _ i p' _ j}{p'^0} \right) \right| \leq C(a_{ij}, b^{ij}) n^0. \label{W15}
\end{align}
For the second derivatives we have
\begin{align*}
&\frac{\partial^2 }{\partial b^{uv} \partial b^{st}} \left( \frac{p' _ i p' _ j}{p'^0} \right) = \frac{(\partial_{b^{uv}}\partial_{b^{st}} p'_i) p'_j}{ p'^0 } + \frac{(\partial_{b^{st}} p'_i) ( \partial_{b^{uv}} p'_j) }{ p'^0 } - \frac{ (\partial_{b^{st}} p'_i) p'_j (\partial_{b^{uv}} p'^0)}{ (p'^0)^2 } \nonumber \\
& \quad + \frac{(\partial_{b^{uv}} p'_i) ( \partial_{b^{st}} p'_j) }{ p'^0 } + \frac{ p'_i (\partial_{b^{uv}}\partial_{b^{st}} p'_j)}{ p'^0 } - \frac{ p'_i (\partial_{b^{st}} p'_j) (\partial_{b^{uv}} p'^0)}{ (p'^0)^2 } \nonumber \\
& \quad - \frac{ (\partial_{b^{uv}} p'_i) p'_j (\partial_{b^{st}} p'^0)}{ (p'^0)^2 } - \frac{p'_i (\partial_{b^{uv}} p'_j) (\partial_{b^{st}} p'^0)}{ (p'^0)^2 } - \frac{p'_i p'_j ( \partial_{b^{uv}} \partial_{b^{st}} p'^0)}{ (p'^0)^2 } + \frac{ 2 p'_i p'_j (\partial_{b^{st}} p'^0) (\partial_{b^{uv}} p'^0) }{ (p'^0)^3 },
\end{align*}
and obtain the following estimate:
\begin{align}
\left| \frac{\partial^2 }{\partial b^{uv} \partial b^{st}} \left( \frac{p' _ i p' _ j}{p'^0} \right)\right| \leq C(a_{ij}, b^{ij}) \left( n^0 + \frac{(n^0)^2}{p'^0} \right) \leq C(a_{ij}, b^{ij} ) \frac{ (n^0)^2 }{p'^0}, \label{W16}
\end{align}
where we used $p'^0 \leq n^0$ again.

The quantity $ p_i p_j / p^0$ is easily estimated, since $p_i$ and $p_j$ do not depend on $b^{st}$. We have
\begin{align*}
\left| \frac{ p_i p_j }{ p^0} \right| + \left| \frac{\partial }{\partial b^{st}} \left( \frac{ p_i p_j }{ p^0} \right) \right| + \left| \frac{\partial^2}{\partial b^{uv} \partial b^{st}} \left( \frac{ p_i p_j }{ p^0 } \right) \right| \leq C(a_{ij}, b^{ij}) p^0, 
\end{align*}
which can be absorbed into \eqref{W14} or \eqref{W15}.

Now, by direct calculations we obtain the desired result. By \eqref{W13} and \eqref{W14} we first obtain
\begin{align}
| W_{ij} | \leq C( a_{ij}, b^{ij} ) \left( \frac{ (p^0)^{ 1 - \frac{ \gamma }{2} } }{ (q^0)^{\frac{ \gamma }{2}}} + \frac{ (q^0)^{1-\frac{ \gamma }{2}}}{ (p^0)^{\frac{ \gamma }{2}}} \right). \label{W17}
\end{align}
For the first derivatives of $W_{ij}$ we use \eqref{W13}, \eqref{W14} and \eqref{W15} to obtain the same estimate:
\begin{align}
\left| \frac{\partial  W_{ij} }{\partial b^{st}} \right| \leq C( a_{ij}, b^{ij} ) \left( \frac{ (p^0)^{ 1 - \frac{ \gamma }{2} } }{ (q^0)^{\frac{ \gamma }{2}}} + \frac{ (q^0)^{1-\frac{ \gamma }{2}}}{ (p^0)^{\frac{ \gamma }{2}}} \right). \label{W18}
\end{align}
For the second derivatives we write them as follows:
\begin{align*}
\frac{\partial^2 W_{ij} }{\partial b^{uv} \partial b^{st}} = W_{lower} + \frac{ h^{ 2 - \gamma } }{ p^0 q^0 } \frac{\partial^2 }{\partial b^{uv} \partial b^{st}} \left( \frac{p' _ i p' _ j}{p'^0} \right),
\end{align*}
where $W_{lower}$ is estimated by using \eqref{W13}, \eqref{W14} and \eqref{W15}:
\begin{align}
|W_{lower}| \leq C( a_{ij}, b^{ij} ) \left( \frac{ (p^0)^{ 1 - \frac{ \gamma }{2} } }{ (q^0)^{\frac{ \gamma }{2}}} + \frac{ (q^0)^{1-\frac{ \gamma }{2}}}{ (p^0)^{\frac{ \gamma }{2}}} \right). \label{W19}
\end{align}
The second term is estimated by using \eqref{W16}:
\begin{align*}
\left| \frac{ h^{ 2 - \gamma } }{ p^0 q^0 } \frac{\partial^2 }{\partial b^{uv} \partial b^{st}} \left( \frac{p' _ i p' _ j}{p'^0} \right) \right| \leq C(a_{ij}, b^{ij} ) \frac{ h^{2- \gamma } (n^0)^2 }{p^0 q^0 p'^0}.
\end{align*}
Now, we need to use \eqref{l2} of Lemma \ref{lem p'} below to take the integration on $\bbs^2$. We choose $\delta > 0$ sufficiently small such that
\begin{equation}
\gamma + \delta < 2
\end{equation}
to obtain
\begin{align}
\int_{\bbs^2} \left| \frac{h ^ { 2 - \gamma }}{p^0 q^0} \frac{\partial^2 }{\partial b^{uv} \partial b^{st}} \left( \frac{p' _ i p' _ j}{p'^0} \right) \right| \,  d\omega 
& \leq C(a_{ij}, b^{ij} ) \frac{ h^{2- \gamma - \delta} (n^0)^{ 1 + \delta} }{p^0 q^0} \nonumber \\
& \leq C(a_{ij}, b^{ij} ) \frac{ (n^0)^{ 3 - \gamma } }{p^0 q^0} \nonumber \\
& \leq C( a_{ij}, b^{ij} ) \left( \frac{ (p^0)^{ 2- \gamma } }{ q^0 } + \frac{ (q^0)^{ 2- \gamma }}{ p^0 } \right), \label{W20}
\end{align}
where we used $h \leq n^0$ by \eqref{l1}. We obtain the desired results by collecting the estimates \eqref{W17}, \eqref{W18}, \eqref{W19} and \eqref{W20}.
\end{proof}

\begin{lemma}\label{lem H}
Suppose that there exist a time interval $[0,T]$ and positive constants $B_1$ and $B_2$ such that $f$ is defined on $[0,T]$ and satisfy
\begin{align}
\sup_{0 \leq s \leq T} \| f(s) \|_{\langle -1 \rangle} + \sup_{0 \leq s \leq T} \| f(s) \|_{\langle 1 \rangle} & \leq B_1, \\
\sup_{0 \leq s \leq T} \left\| \frac{\partial f}{\partial s} (s) \right\|_{\langle -1 \rangle} + \sup_{0 \leq s \leq T} \left\| \frac{\partial f}{\partial s} (s) \right\|_{\langle 1 \rangle} & \leq B_2.
\end{align}
Then, the function $H$ satisfies
\begin{align}
\sup_{0 \leq s \leq T} \| H (s) \| + \sup_{0 \leq s \leq T} \left\| \frac{\partial H}{\partial x}(s) \right\| + \sup_{0 \leq s \leq T} \left\| \frac{\partial^2 H}{\partial x^2} (s) \right\| & \leq C(a_{ij}, b^{ij}) B_1^2, \\
\sup_{0 \leq s \leq T} \left\| \frac{\partial H}{\partial s} (s) \right\| + \sup_{0 \leq s \leq T} \left\| \frac{\partial^2 H}{\partial x \partial s} (s) \right\| & \leq C (a_{ij}, b^{ij}) B_1 B_2,
\end{align}
where $C(a_{ij}, b^{ij})$ are positive constants independent of $T$.
\end{lemma}
\begin{proof}
The proof is straightforward. Explicit formulas of the derivatives of $H_{ij}$ with respect to $s$ and $a_{st}$ are given in \eqref{H1}--\eqref{H4}, and derivatives with respect to $b^{st}$ are estimated by Lemma \ref{lem W}. Recall that $H_{ij}$ is given by
\begin{equation*}
H_ { i j } = \frac { 8 \pi c_\gamma } { \det a }  \iiint W_{ij}(p_k,q_k,\omega_k,b^{kl}) f(p) f(q) \, d \omega \, d ^ 3 q \, d ^ 3 p.
\end{equation*}
Applying \eqref{lem W1} of Lemma \ref{lem W} we obtain the following estimate:
\begin{equation}
| H_{ i j } | \leq C ( a_{ i j } , b^{ i j } ) \iint \frac{ ( p^0 )^{ 1 - \frac{ \gamma }{2} } }{ (q^0)^{\frac{ \gamma }{2}} } f(p) f(q) \, d^3q \, d^3p \leq C( a_{ij}, b^{ij} ) \| f \|_{L^1_{1- \gamma /2}} \| f \|_{L^1_{- \gamma /2}},
\end{equation}
where we used the symmetry in $ p $ and $ q $ in the integration. By the same arguments, we obtain the same estimates for the following quantities:
\begin{align}
\left| \frac{\partial H_{ij} }{\partial a_{st}} \right| + \left| \frac{\partial H_{ij} }{\partial b^{st}} \right| + \left| \frac{\partial^2 H_{ij} }{\partial b^{uv} \partial a_{st}} \right| + \left| \frac{\partial^2 H_{ij} }{\partial a_{uv} \partial b^{st}} \right| + \left| \frac{\partial^2 H_{ij} }{\partial a_{uv} \partial a_{st}} \right| \leq C( a_{ij}, b^{ij} ) \| f \|_{L^1_{1- \gamma /2}} \| f \|_{L^1_{- \gamma /2}}.
\end{align}
For the second derivatives with respect to $b^{st}$ we use \eqref{lem W2} of Lemma \ref{lem W}:
\begin{align}
\left| \frac{\partial^2 H_{ij} }{\partial b^{uv} \partial b^{st}} \right| \leq C( a_{ij}, b^{ij} ) \left( \| f \|_{L^1_{2- \gamma }} \| f \|_{L^1_{-1}} + \| f \|_{L^1_{1- \gamma /2}} \| f \|_{L^1_{- \gamma /2}} \right).
\end{align}
For the $s$-derivative of $H_{ij}$ we consider \eqref{H1} to obtain
\begin{align}
\left| \frac{\partial H_{ij} }{\partial s} \right| & \leq C(a_{ij}, b^{ij} ) \iint \left( \frac{ (p^0)^{ 1 - \frac{ \gamma }{2} } }{ (q^0)^{\frac{ \gamma }{2}}} + \frac{ (q^0)^{1-\frac{ \gamma }{2}}}{ (p^0)^{\frac{ \gamma }{2}}} \right) \left| \frac{\partial f}{\partial s}(p) \right| \left| f(q) \right| \, d^3q \, d^3p \nonumber \\
& \leq C(a_{ij}, b^{ij} ) \left( \left\| \frac{\partial f}{\partial s} \right\|_{L^1_{1- \gamma /2}} \| f \|_{L^1_{- \gamma /2}} + \left\| \frac{\partial f}{\partial s} \right\|_{L^1_{- \gamma /2}} \| f \|_{L^1_{1- \gamma /2}} \right).
\end{align}
We obtain the same estimates for the quantities $\partial^2 H_{ij} / \partial s \partial a_{st}$ and $\partial^2 H_{ij} / \partial s \partial b^{st}$. By the relations \eqref{p^0 1}--\eqref{p^0 2} we have for any $ m \in \bbr $, 
\begin{equation}
\| f \|_{L^1_m} \leq C(a_{ij}, b^{ij}) \| f \|_{\langle m \rangle},\qquad \left\| \frac{\partial f}{\partial s} \right\|_{L^1_m} \leq C(a_{ij}, b^{ij}) \left\| \frac{\partial f}{\partial s} \right\|_{\langle m \rangle}.
\end{equation}
The lemma is now proved by an interpolation, since $ 1< \gamma <2$. 
\end{proof}

\subsubsection{Eigenvalue conditions for $ N $}\label{eigen}
There are also eigenvalue conditions on the matrix~$ N $ required for Theorem \ref{ODEthm}, which we now verify by computing the eigenvalues.

The vector space in question is the direct sum of two copies of the space of symmetric bilinear forms on $ \mathbb R ^ 3 $, one for $ K $ and one for $ \hat Z $, since those are the variables in the singular part of our system of first order ordinary differential equations.

The theorem only requires the eigenvalue conditions to hold when $ s = 0 $, so we will assume $ s = 0 $ for the remainder of this section.
We therefore have $ \Psi _ { i j } = \frac 1 2 a _ { i j } $. This implies that
\begin{align}
  b ^ { i j } \Psi _ { i j } = \frac 1 2 b ^ { i j } a _ { i j } =
  \frac 3 2 ,
\end{align}
and $N = c_\gamma \, Q$, which are defined by \eqref{N} and \eqref{Q}, respectively. We note that the matrix $Q$ can be written as
\begin{align}
  Q = \left ( \begin{matrix}
    \gamma I + 3 \pi & - 2 I
    \\
    I - \chi + 3 \pi / 2 & c_\gamma^{-1} I
  \end{matrix} \right ),
\end{align}
where $I$, $\pi$ and $\chi$ are all considered as operators on the space of symmetric bilinear forms on $ \mathbb R ^ 3 $. Now
\begin{align}
  \pi ^ { q r } _ { k l } \chi ^ { k l } _ { m n }
  = \frac13 a_{kl} b^{qr} \Psi_{mncd} b^{ck} b^{dl} = \frac13 b^{qr} \Psi_{mnld} b^{dl} = \frac13 b^{qr} \Psi_{mn} = \frac13 \Pi_{mn}^{qr}
  = \frac 1 2 \pi ^ { q r } _ { m n },
\end{align}
where in the last equation we have assumed $s=0$,
while
\begin{align}
  \chi ^ { q r } _ { k l } \pi ^ { k l } _ { m n }
  = \frac13 \Psi_{klcd} b^{cq} b^{dr} a_{mn} b^{kl} = \frac13 \Psi_{cd} b^{cq} b^{dr} a_{mn} = \frac16 a_{cd} b^{cq} b^{dr} 2 \Psi_{mn}= \frac13 b^{qr} \Psi_{mn}
  = \frac 1 2 \pi ^ { q r } _ { m n },
\end{align}
where we also have assumed $s=0$. In other words we have
\begin{align}
  \chi \pi = \frac 1 2 \pi = \pi \chi .
\end{align}
Since $ \pi $ and $ \chi $ commute they are simultaneously diagonalisable.
$ \pi $ is a projection, so its eigenvalues are $ 1 $ and $ 0 $. $ \pi $
of anything is a multiple of $ \Psi _ { i j } $, so it is of rank one and
the eigenspace corresponding to the eigenvalue $ 1 $ is one-dimensional
and the eigenspace corresponding to the eigenvalue $ 0 $ is five-dimensional.
It follows from the equation above that $ \chi $ has eigenvalue
$ \frac 1 2 $ on the space where $ \pi $ has eigenvalue 1.

The equation $ \chi \pi = \frac 1 2 \pi = \pi \chi $
provides no further information on the eigenvalues of $ \chi $,
but we can get bounds for those eigenvalues by other means.
Suppose that $ z $ is an eigenvector of $ \chi $ with eigenvalue~$ \mu $.
Clearly
\begin{align}
  \frac { 4 \pi } { \sqrt{ \det a } }
  \int
  ( b ^ { i j } p _ i p _ j ) ^ { - 3 / 2 }
  ( b ^ { k l } b ^ { m n } z _ { l m } p _ k p _ n ) ^ 2
  f ( p )
  \, d ^ 3 p > 0
\end{align}
since the integrand is non-negative everywhere and positive on an open
set.
We can rewrite this as
\begin{align}
  \frac { 4 \pi } { \sqrt{ \det a } }
  \int
  ( b ^ { i j } p _ i p _ j ) ^ { - 3 / 2 }
  b ^ { k l } b ^ { m n } z _ { l m } p _ k p _ n
  b ^ { k ' l ' } b ^ { m ' n ' } z _ { l ' m ' } p _ { k ' } p _ { n ' }
  f ( p )
  \, d ^ 3 p > 0
\end{align}
or
\begin{align}
  b ^ { k l } b ^ { m n } z _ { l m } z _ { l ' m ' }
  \chi ^ { l ' m ' } _ { k n } > 0
\end{align}
or, using the eigenvalue equation,
\begin{align}
  \mu b ^ { k l } b ^ { m n } z _ { l m } z _ { k n } > 0 .
\end{align}
Since $ b ^ { k l } b ^ { m n } z _ { l m } z _ { k n } \ge 0 $ it
follows that $ \mu > 0 $. So all eigenvalues are positive.
The trace of $ \chi $ is
\begin{align}
  \chi ^ { i j } _ { i j }
  = b ^ { i k } b ^ { j l } \Psi _ { i j k l }
  = b ^ { i k } \Psi _ { i k }
  = \frac 3 2 .
\end{align}
Since we already have an eigenvector with eigenvalue
$ \frac 1 2 $ and all eigenvalues are positive it follows that
all eigenvalues are less than~1.

On a subspace where $ \pi = 1 $ we have $ \chi = \frac 1 2 $ and
\begin{align}
  Q = \left ( \begin{matrix} \gamma + 3 & - 2 \\ 2 & \gamma - 1 \end{matrix} \right ).
\end{align}
The corresponding eigenvalues of $ Q $ are the double eigenvalue $\gamma+1$ and the eigenvalues of $ N $ are thus
\begin{align}
 c_{\gamma}  (\gamma+1)= \frac { \gamma + 1 } { \gamma - 1 } > 1 .
\end{align}
Since $1<\gamma< 2$ by our assumption on the scattering kernel, the eigenvalues are strictly greater than one.

On a subspace where $ \pi = 0 $ and $ \chi = \mu $ we have
\begin{align}
  Q = \left ( \begin{matrix} \gamma & - 2 \\ 1 - \mu & \gamma - 1 \end{matrix} \right ) .
\end{align}
The corresponding eigenvalues of $ Q $ are
\begin{align}
  \frac { 2 \gamma - 1 \pm \sqrt { 1 - 8 ( 1 - \mu ) } } 2
\end{align}
and the corresponding eigenvalues of $ N $ are
\begin{align}
  \frac { 2 \gamma - 1 \pm \sqrt { 1 - 8 ( 1 - \mu ) } } { 2 ( \gamma - 1 ) } .
\end{align}
Now $ \mu < 1 $ so $ 1 - 8 ( 1 - \mu ) < 1 $ and hence the real parts of
both eigenvalues of $ N $ are positive and in fact bigger than one.

We've just seen that all eigenvalues of $ N $ initially have positive
real parts, so $ N $ is invertible. In particular this means that
$ W _ n = 0 $ and $ W = W _ i $ in the notation of the appendix.
We now return to the question of solving Equations (\ref{FF2bis}) and (\ref{FF3bis}) for $ a $, $ b $, $ K $ and $ \hat Z $ in terms of $ f $.
From those two equations it follows that
\begin{align}
  ( \gamma ^ 2 - \gamma + 2 ) K _ { i j } + 3 \gamma \pi ^ { m n } _ { i j } K _ { m n } - 2 \chi ^ { m n } _ { i j } K _ { m n } = 2 ( \gamma - 1 )  H  _ { i j } .
\end{align}
Now
\begin{align*}
  J = \left ( \begin{matrix} 0 & - I \\ I & 0 \end{matrix} \right )
\end{align*}
is invertible and we've just seen that $ N $ is invertible, so
\begin{align}
  - ( \gamma - 1 ) ^ 2 ( N J ) ^ 2
  = \left ( \begin{matrix}
    ( \gamma ^ 2 - \gamma + 2 ) I + 3 \gamma \pi - 2 \chi & 0 \cr
    0 & ( \gamma ^ 2 - \gamma + 2 ) I + 3 \gamma \pi - 2 \chi
    \end{matrix} \right )
\end{align}
is also invertible. From this it follows that
\begin{align}
  ( \gamma ^ 2 - \gamma + 2 ) I + 3 \gamma \pi - 2 \chi
\end{align}
is invertible, so there is a unique choice of $ K $ satisfying the equation
above, and then a unique $ \hat Z $ satisfying Equation~(\ref{FF2bis}). These
$ K $ and $ \hat Z $ then satisfy Equation~(\ref{FF3bis}). In other words,
the initial values of $ a $ uniquely determine those of $ K $ and $ \hat Z $,
as claimed in Section~\ref{ssfc}.

\subsubsection{Proof of Proposition \ref{propeinstein}}\label{sproofprop2}
We are now ready to prove Proposition \ref{propeinstein}. The proof is an application of Theorem \ref{ODEthm} to the equations \eqref{ODEx}--\eqref{ODE}, but we need to slightly modify it as follows. Let us write
\begin{align}
R : = 2 \max \left( \| a_0 \| , \| b_0 \| , \| K_0 \| , \| \hat{Z}_0 \| \right), 
\end{align}
which implies that $x_0 \in \overline { B _ V ( 0 , R/2 ) } $ and $ y_0 \in \overline { B _ W ( 0 , R/2 ) } $, where $V = W = S_3(\bbr) \times S_3(\bbr)$, and introduce $K_2 = [0,T] \times \overline { B _ V ( 0 , R ) }$ and $K_3 = [0,T] \times \overline { B _ V ( 0 , R ) } \times \overline { B _ W ( 0 , R ) }$. Following the proof of Theorem \ref{ODEthm} we observe that the solutions $x$ and $y$ exist on $[ 0 , t_{\min}]$, where $t_{\min}$ is given by
\begin{align}
t_{\min} = \min \left( T, \frac{R}{2Q_x}, \frac{R}{2Q_y}, \frac{1}{2Q_u}, \frac{1}{2Q_v}\right), 
\end{align}
which is slightly different from the one in Theorem \ref{ODEthm}. Here, $Q_x$, $Q_y$, $Q_u$ and $Q_v$ are almost the same as the ones in Theorem \ref{ODEthm}. The only difference is that we do not need the constant $C_y$, which is needed to estimate $y \in \overline{B_W (y_0, R)}$, but in this proof we consider $y \in \overline{B_W(0, R)}$ so that the constant $C_y$ can be replaced simply by $R$. Now, Lemmas \ref{lem F}, \ref{lem N} and \ref{lem H} imply that there exists a constant $C_R$, which depends on $R$, satisfying 
\begin{align}
\max \left( \frac{ 2 Q_x }{ R }, \frac{2 Q_y}{R}, 2 Q_u, 2 Q_v \right) \leq C_R ( 1 + B_1 + B_2 )^2.
\end{align}
We note that $\max_{[0,1]} \| P \|$ depends only on initial data so that the constant $C_P$ can be absorbed into $C_R$. This shows that there exists $0 < T_B \leq t_{\min} \leq T$, which depends on $B_1$ and $B_2$, such that the equations \eqref{ODEx}--\eqref{ODE} have a unique solution on $[0, T_B]$. We have proved in the previous section that $N_0$ satisfies the eigenvalue conditions. This completes the proof of Proposition \ref{propeinstein}.  \qed

\section{Estimates for the Boltzmann equation}\label{sBoltzmann}
In this part we study the Boltzmann equation for a given metric. The space-time will be assumed to be of Bianchi type I, in which case the structure constants $ { C^k }_{ i j } $ vanish so that the Boltzmann equation \eqref{V5} with \eqref{V6} and \eqref{C} reduces to 
\begin{align}
\frac{ \partial f }{ \partial s } = Q ( f , f ) , \label{B}
\end{align}
where the collision term $ Q $ is given as in \eqref{C}. In this section, we will replace the factor $ 1 / \sqrt{ \det a } $ in \eqref{C} with $ ( \det b )^{ \frac12 } $ for the convenience of calculation: 
\begin{align*}
Q ( f , f ) = ( \det b )^{ \frac12 } \int_{ \bbr^3 } \int_{ \bbs^2 } \frac{ h^{ 2 - \gamma } }{ p^0 q^0 }( f ( p' ) f ( q' ) - f ( p ) f ( q ) ) \, d \omega \, d^3 q . 
\end{align*}
Recall that $ h $, $ p^0 $ and $ q^0 $ are defined by \eqref{hp^0q^0}, and $ p' $ and $ q' $ are parametrized by \eqref{p'^0}--\eqref{q'}. We also recall that the post-collision momenta can be written with respect to the orthonormal frame \eqref{ortho_e} and \eqref{ortho_theta} as in \eqref{phat'^0}--\eqref{qhat'}:
\begin{align*}
p'^0 & = \frac12 ( n^0 + \hat{ n } \cdot \omega ) , \\
q'^0 & = \frac12 ( n^0 - \hat{ n } \cdot \omega ) , 
\end{align*}
and
\begin{align*}
\hat{ p }'_j & = \frac12 \left( \hat{ n }_j + h \omega_j + \frac{ ( \hat{ n } \cdot \omega ) \hat{ n }_j }{ n^0 + h } \right) , \\ 
\hat{ q }'_j & = \frac12 \left( \hat{ n }_j - h \omega_j - \frac{ ( \hat{ n } \cdot \omega ) \hat{ n }_j }{ n^0 + h } \right) . 
\end{align*}
Here, we denote 
\begin{equation*}
n^0 = p^0 + q^0 , \qquad \hat{ n }_j = \hat{ p }_j + \hat{ q }_j , 
\end{equation*}
and $ p^0 $, $ q^0 $ and $ h $ are understood as functions of $ \hat{ p }_j $ and $ \hat{ q }_j $ as in \eqref{hp^0q^0hat}:
\begin{equation*}
h = \sqrt{ 2 p^0 q^0 - 2 ( \hat{ p } \cdot \hat{ q } ) } , \qquad p^0 = | \hat{ p } | , \qquad q^0 = | \hat{ q } | ,
\end{equation*}
where $ \hat{ p } \cdot \hat{ q } $ is the usual inner product in three dimensions for $ { \hat p } = ( \hat{ p }_1 , \hat{ p }_2 , \hat{ p }_3 ) $ and $ { \hat q } = ( \hat{ q }_1 , \hat{ q }_2 , \hat{ q }_3 ) $. We remark that the representations \eqref{p'^0}--\eqref{q'} and \eqref{phat'^0}--\eqref{qhat'} can be used interchangeably through \eqref{ortho_e} and \eqref{ortho_theta}, as long as a metric $ b^{ i j } $ exists. We also remark that the following change of variables between pre- and post-collision momenta  will be frequently used: 
\begin{align}
\det b \frac{ d \omega \, d^3 p \, d^3 q }{ p^0 q^0 } = \frac{ d \omega \, d^3 { \hat p } \, d^3 \hat{ q } }{ p^0 q^0 } = \frac{ d \omega \, d^3 { \hat p }' \, d^3 \hat{ q }' }{ p'^0 q'^0 } = \det b \frac{ d \omega \, d^3 p' \, d^3 q' }{ p'^0 q'^0 } . \label{dpdq}
\end{align}

In this section we will obtain the existence for the Boltzmann equation. The existence will be obtained in $ L^1 $ by following the well-known $ L^1 $ framework in the Boltzmann theory \cite{ark1,ark2,CIP, LNT1, LR,SY}, but we use singular weights, so that singular moment estimates will be exploited. We need the following lemmas.

\begin{lemma}\label{lem h}
The relative momentum $ h $ is a collisional invariant and can be written as
\begin{align}
h = \sqrt{ ( n^0 )^2 - | { \hat n } |^2 } = 2 \sqrt{ p^0 q^0 } \sin \frac{ \phi }{ 2 } , \label{l1}
\end{align}
where $ \phi $ is the angle between the three dimensional vectors $ { \hat p } $ and $ { \hat q } $.
\end{lemma}
\begin{proof}
Since $ p_\alpha p^\alpha = q_\alpha q^\alpha = 0 $, we have $ h^2 = - 2 p_\alpha q^\alpha = - n_\alpha n^\alpha $. Then, the energy-momentum conservation \eqref{conservation} shows that $ h $ is a collisional invariant. In an orthonormal frame we have
\begin{align*}
h^2 = - 2 p_\alpha q^\alpha = 2 ( | { \hat p } | | { \hat q } | - { \hat p } \cdot { \hat q } )
= 4 | { \hat p } | | { \hat q } | \sin^2 \frac{ \phi }{ 2 } ,
\end{align*}
where $ \phi $ is the angle between  $ { \hat p } $ and $ { \hat q } $.
\end{proof}

\begin{lemma}\label{lem p'}
Let $ p'^0 $ and $ q'^0 $ be the post-collision (unphysical) momenta for given $ p , q \in \bbr^3 $, $ \omega \in \bbs^2 $ and $ b^{ i j } \in S_3 ( \bbr ) $. Then, for any $ \delta > 0 $ there exists $ C > 0 $ such that the following hold:
\begin{align}
& \int_{ \bbs^2 } \frac{ 1 }{ p'^0 } \, d \omega = \int_{ \bbs^2 } \frac{ 1 }{ q'^0 } \, d \omega \leq \frac{ C }{ h^\delta ( n^0 )^{ 1-\delta } } , \label{l2} \\
& \int_{ \bbs^2 } \frac{ 1 }{ ( p'^0 )^2 } \, d \omega = \int_{ \bbs^2 } \frac{ 1 }{ ( q'^0 )^2 } \, d \omega = \frac{ 16 \pi }{ h^2 } , \label{l3} \\ 
& \int_{ \bbs^2 } \frac{ 1 }{ p'^0 q'^0 } \, d \omega \leq \frac{ C }{ h^\delta ( n^0 )^{ 2-\delta } } , \label{l4} 
\end{align}
where $ C $ does not depend on $ p $, $ q $ and $ b^{ i j } $.
\end{lemma}
\begin{proof}
Let us consider an orthonormal frame, and use \eqref{phat'^0} to obtain
\begin{align}
\int_{ \bbs^2 } \frac{ 1 }{ p'^0 } \, d \omega & = \int_{ \bbs^2 } \frac{ 2 }{ n^0 + \hat{ n } \cdot \omega } \, d \omega \nonumber \\
& = \int_0^\pi \frac{ 4 \pi \sin \phi \, d \phi }{ n^0 + | \hat{ n } | \cos \phi } \nonumber \\
& = \frac{ 4 \pi }{ | \hat{ n } | } \ln \left( \frac{ n^0 + | \hat{ n } | }{ n^0 - | \hat{ n } | } \right) \nonumber \\
& = \frac{ 8 \pi }{ | \hat{ n } | } \ln \left( \frac{ n^0 + | \hat{ n } | }{ h } \right) ,
\end{align}
where we used \eqref{l1}. Now, for any $ \delta > 0 $ we have
\begin{align}\label{l2-ineq}
\frac{ 8 \pi }{ | \hat{ n } | } \ln \left( \frac{ n^0 + | \hat{ n } | }{ h } \right) = \frac{ 8 \pi }{ h^\delta ( n^0 )^{ 1-\delta } } \left( 1+ \frac{ | \hat{ n } |^2 }{ h^2 } \right)^{ \frac{ 1-\delta }{ 2 } } \frac{ \ln \left( \frac{ | \hat{ n } | }{ h } + \sqrt{ 1 + \frac{ | \hat{ n } |^2 }{ h^2 } } \right) }{ \frac{ | \hat{ n } | }{ h } } \leq \frac{ C }{ h^\delta ( n^0 )^{ 1-\delta } } ,
\end{align}
where $ C $ is a positive constant depending on $ \delta $. We obtain the same result for $ q'^0 $ by the symmetry. We skip the proof of \eqref{l3}, which is the same as the proof of Lemma 5 of \cite{LNT1}. For \eqref{l4}, we have
\begin{align}
\int_{ \bbs^2 } \frac{ 1 }{ p'^0 q'^0 } \, d \omega & = \int_{ \bbs^2 } \frac{ 4 }{ ( n^0 )^2 - ( \hat{ n } \cdot \omega )^2 } \, d \omega \nonumber \\
& = \int_0^\pi \frac{ 8 \pi \sin \phi \, d \phi }{ ( n^0 )^2 - | \hat{ n } |^2 \cos^2 \phi } \nonumber \\
& = \frac{ 8 \pi }{ n^0 | \hat{ n } | } \ln \left( \frac{ n^0 + | \hat{ n } | }{ n^0 - | \hat{ n } | } \right) \nonumber \\
& = \frac{16 \pi }{ n^0 | \hat{ n } | } \ln \left( \frac{ n^0 + | \hat{ n } | }{ h } \right) , 
\end{align}
where we used \eqref{l1}. Now, we obtain the desired result by the inequality \eqref{l2-ineq}. 
\end{proof}

\subsection{Existence for the Boltzmann equation with cut-off}\label{ExistBoltzmannCutoff}
We first prove the existence of solutions to a modified Boltzmann equation. We choose a constant $ \varepsilon > 0 $ and consider the modified Boltzmann equation:
\begin{align}
\frac{ \partial f }{ \partial s } = Q_\varepsilon ( f , f ) , \label{Bmod}
\end{align}
where $ Q_\varepsilon $ is the collision operator with cut-off defined by 
\begin{align}
Q_\varepsilon ( f , f ) = ( \det b )^{ \frac12 } \iint_{ h \geq \varepsilon } \frac{ h^{ 2 - \gamma } }{ p^0 q^0 } ( f ( p' ) f ( q' ) - f ( p ) f ( q ) ) \, d \omega \, d^3 q . \label{Cmod} 
\end{align} 
We notice that $ h^{ 2 - \gamma } / ( p^0 q^0 ) \sim h^{ - \gamma } $ by Lemma \ref{lem h}, and the singularity is cut-off by the restriction $ h \geq \varepsilon $. We may write the collision operator in an orthonormal frame and denote it by
\begin{equation*}
Q_\varepsilon ( f , f ) = \widehat{ Q }_\varepsilon ( \hat{ f } , \hat{ f } ) ,
\end{equation*}
where $ \hat{ f } $ is defined by $ \hat{ f } = \hat{ f } ( \hat{ p } ) = f ( p ) $, and 
\begin{equation*}
\widehat{ Q }_\varepsilon ( \hat{ f } , \hat{ f } ) = \iint_{ h \geq \varepsilon } \frac{ h^{ 2 - \gamma } }{ p^0 q^0 } ( \hat{ f } ( \hat{ p }' ) \hat{ f } ( \hat{ q }' ) - \hat{ f } ( \hat{ p } ) \hat{ f } ( \hat{ q } ) ) \, d \omega \, d^3 \hat{ q } ,
\end{equation*}
where the volume form is given by
\begin{align}
d^3 \hat{ q } = ( \det b )^{ \frac12 } \, d^3 q . \label{dqhat}
\end{align}
Now that an orthonormal frame is applied, we can follow the arguments of \cite{LNT1}:
\begin{align*}
\int_{ \bbr^3 } | \widehat{ Q }_\varepsilon ( \hat{ f } , \hat{ f } ) ( \hat{ p } ) | \, d^3 \hat{p} & \leq C_\varepsilon \left( \int_{ \bbr^3 } | \hat{ f } ( \hat{ p } ) | \, d^3 \hat{ p } \right)^2 , \\
\int_{ \bbr^3 } | \widehat{ Q }_\varepsilon ( \hat{ f } , \hat{ f } ) ( \hat{ p } ) - \widehat{ Q }_\varepsilon ( \hat{ g } , \hat{ g } ) ( \hat{ p } ) | \, d^3 \hat{ p } & \leq C_\varepsilon \left( \int_{ \bbr^3 } | \hat{ f } ( \hat{ p } ) + \hat{ g } ( \hat{ p } ) | \, d^3 \hat{ p } \right) \left( \int_{ \bbr^3 } | \hat{ f } ( \hat{ p } ) - \hat{ g } ( \hat{ p } ) | \, d^3 \hat{ p } \right).
\end{align*}
Applying \eqref{dqhat} we have
\begin{align*}
\| Q_\varepsilon ( f , f ) \|_{ L^1 } & \leq C_\varepsilon ( \det b )^{ \frac12 } \| f \|^2_{ L^1 } , \\
\| Q_\varepsilon ( f , f ) - Q_\varepsilon ( g , g ) \|_{ L^1 } & \leq C_\varepsilon ( \det b )^{ \frac12 } \| f + g \|_{ L^1 } \| f - g \|_{ L^1 } .
\end{align*}
This shows that the modified equation \eqref{Bmod} with \eqref{Cmod} has a solution in $ L^1 $ as long as the metric $ b^{ i j } $ exists and $ ( \det b )^{ \frac12 } $ is bounded. Next, we need to show that the solution $ f $ is non-negative. This ensures that the time interval on which the solution $ f $ exists can be extended to the time interval on which the metric $ b^{ i j } $ exists. Since the proof of the non-negativity of $ f $ is exactly the same as the proof of Proposition 1 of \cite{LNT1}, we only refer to \cite{LNT1} for the proof. We obtain the existence part of the following result:

\begin{prop}\label{prop mod}
Let $ \varepsilon > 0 $ be given. Suppose that there exist a time interval $ [ 0 , T ] $ and a positive constant $ A $ such that $ b^{ i j } $ and $ K_{ i j } $ are defined on $ [ 0 , T ] $ and satisfy
\begin{align}
\sup_{ 0 \leq s \leq T } \| b ( s ) \| \leq A , \qquad \sup_{ 0 \leq s \leq T } \| K ( s ) \| \leq A . 
\end{align}
Then, for any initial data $ f_0 \in L^1( \bbr^3 ) $ with $ f_0 \geq 0 $ the modified equation \eqref{Bmod} with \eqref{Cmod} has a unique non-negative solution $ f \in C^1 ( [ 0 , T ] ; L^1( \bbr^3 ) ) $. Moreover, if $ f_0 \in L^1_{ - 2 } ( \bbr^3 ) \cap L^\infty_\eta ( \bbr^3 ) $ for some $ 0 < \eta < 2 / A^2 $, then there exists $ 0 < T_A \leq T $ such that
\begin{align}
\sup_{ 0 \leq s \leq T_A } \| f ( s ) \|_{ L^1_{ - 2 } } \leq C_A ,
\end{align}
where $ T_A $ and $ C_A $ are positive constants depending on $ A $ and $ \eta $.
\end{prop}
\begin{proof}
We first estimate the $ L^\infty_\eta $-norm of $ f $. Multiplying the equation \eqref{Bmod} by $ w_\eta $, we obtain
\begin{align}
\frac{ \partial ( w_\eta f ) }{ \partial s } = \frac{ \partial w_\eta }{ \partial s } f + w_\eta Q_\varepsilon ( f , f ) .\label{L infty 0}
\end{align}
To estimate the first quantity on the right hand side we note that
\begin{align}\label{dp ds}
\left| \frac{ \partial p^0 }{ \partial s } \right| = \left| \frac{ b^{ i m } b^{ j n } K_{ m n } p_i p_j }{ 2 p^0 } \right| \leq \frac{ ( b^{ i m } b^{ j n } K_{ m n } K_{ i j } )^{ \frac12 } ( b^{ k l } p_k p_l ) }{ 2 p^0 } \leq \frac{ A^2 }{ 2 } p^0 ,
\end{align}
and that there exists $ 0 < T_A \leq T $ such that the following holds on $ [ 0 , T_A ] $:
\begin{align}
\frac{ A^2 }{ 2 } - \frac{ \eta }{ s + \eta^2 } \leq 0 ,
\end{align}
since $ 0 < \eta < 2 / A^2 $. We now have
\begin{align}
\frac{ \partial w_\eta }{ \partial s } & = \frac{ \partial p^0 }{ \partial s } \exp \left\{ \frac{ p^0 }{ ( s + \eta^2 )^\eta } \right\} + p^0 \left\{ \frac{ \partial p^0 }{ \partial s } \frac{ 1 }{ ( s + \eta^2 )^\eta } - \eta \frac{ p^0 }{ ( s + \eta^2 )^{ \eta + 1 } } \right\} \exp \left\{ \frac{ p^0 }{ ( s + \eta^2 )^\eta } \right\} \nonumber \\
& \leq \frac{ A^2 }{ 2 } p^0 \exp \left\{ \frac{ p^0 }{ ( s + \eta^2 )^\eta } \right\} + p^0 \left\{ \frac{ A_2 }{ 2 } p^0 \frac{ 1 }{ ( s + \eta^2 )^\eta } - \eta \frac{ p^0 }{ ( s + \eta^2 )^{ \eta + 1 } } \right\} \exp \left\{ \frac{ p^0 }{ ( s + \eta^2 )^\eta } \right\} \nonumber \\
& = \frac{ A^2 }{ 2 } p^0 \exp \left\{ \frac{ p^0 }{ ( s + \eta^2)^\eta } \right\} + \frac{ ( p^0 )^2 }{ ( s + \eta^2 )^\eta } \left\{ \frac{ A_2 }{ 2 } - \frac{ \eta }{ s + \eta^2 } \right\} \exp \left\{ \frac{ p^0 }{ ( s + \eta^2 )^\eta } \right\} \nonumber \\
& \leq \frac{ A^2 }{ 2 } w_\eta , \label{dw ds}
\end{align}
on $ [ 0 , T_A ] $. Hence, we obtain from the equation \eqref{L infty 0} the following inequality:
\begin{align}
\frac{ \partial ( w_\eta f ) }{ \partial s } \leq \frac{ A^2 }{ 2 } w_\eta f + w_\eta Q_\varepsilon ( f , f ) . \label{L infty 1}
\end{align}
The second quantity on the right hand side of \eqref{L infty 0} is estimated as follows. Since the solution is non-negative, we obtain
\begin{align}
w_\eta Q_\varepsilon ( f , f ) & \leq ( \det b )^{ \frac12 } \iint_{ h \geq \varepsilon } \frac{ h^{ 2 - \gamma } } { p^0 q^0 } p^0 \exp ( s^{ - 1 }_\eta p^0 ) f ( p' ) f ( q' ) \, d \omega \, d^3 q \nonumber \\
& \leq ( \det b )^{ \frac12 } \| f \|_{ L^\infty_\eta }^2 \int_{ \bbr^3 } \int_{ \bbs^2 } \frac{ h^{ 2 - \gamma } }{ q^0 } \exp( s^{ - 1 }_\eta p^0 ) \frac{ 1 }{ p'^0 } \exp( - s^{ - 1 }_\eta p'^0 ) \frac{ 1 }{ q'^0 } \exp( - s^{ - 1 }_\eta q'^0 ) \, d \omega\, d^3 q \nonumber \\
& = ( \det b )^{ \frac12 } \| f \|_{ L^\infty_\eta }^2 \int_{ \bbr^3 } \int_{ \bbs^2 } \frac{ h^{ 2 - \gamma } }{ q^0 } \exp( - s^{ - 1 }_\eta q^0 ) \frac{ 1 }{ p'^0 q'^0 } \, d \omega \, d^3 q ,
\end{align}
where we used the conservation of energy \eqref{conservation}. We apply \eqref{l4} of Lemma \ref{lem p'} to the integration on $ \bbs^2 $. We choose $ \delta = 2 - \gamma $ to obtain
\begin{align}
\int_{ \bbs^2 } \frac{ 1 }{ p'^0 q'^0 } \, d \omega \leq \frac{ C }{ h^{ 2 - \gamma } ( n^0 )^\gamma } \leq \frac{ C }{ h^{ 2 - \gamma } ( q^0 )^\gamma } ,
\end{align}
which implies that since $ 1 < \gamma < 2 $,
\begin{align}
w_\eta Q_\varepsilon ( f , f ) & \leq C ( \det b )^{ \frac12 } \| f \|_{ L^\infty_\eta }^2 \int_{ \bbr^3 } \frac{ 1 }{ ( q^0 )^{ 1 + \gamma } } \exp( - s^{ - 1 }_\eta q^0 ) \, d^3 q \nonumber \\
& \leq C s_\eta^{ 2 - \gamma } \| f \|_{ L^\infty_\eta }^2 . \label{L infty 2}
\end{align}
We combine \eqref{L infty 1} and \eqref{L infty 2} to obtain a Gr{\"o}nwall type inequality:
\begin{align}
\frac{ d }{ d s } \| f \|_{ L^\infty_\eta } \leq \frac{ A^2 }{ 2 } \| f \|_{ L^\infty_\eta } + C s_\eta^{ 2 - \gamma } \| f \|^2_{ L^\infty_\eta } .
\end{align}
Hence, the norm $ \| f \|_{ L^\infty_\eta } $ is bounded on a time interval, which we still denote by $ [ 0 , T_A ] $, so that we have
\begin{align}
\sup_{ 0 \leq s \leq T_A } \| f ( s ) \|_{ L^\infty_\eta } \leq C_A , \label{L infty}
\end{align}
where $ C_A $ is a constant depending on $ A $ and $ \eta $.

Next, we consider the $ L^1_r $-norm of $ f $. Multiplying \eqref{Bmod} by $ ( p^0 )^r $, integrating it over $ \bbr^3 $, and using the change of variables \eqref{dpdq}, we obtain
\begin{align}
& \frac{ d }{ d s } \int_{ \bbr^3 } f ( p ) ( p^0 )^r \, d^3 p = r \int_{ \bbr^3 } f ( p ) ( p^0 )^{ r - 1 } \frac{ \partial p^0 }{ \partial s } \, d^3 p \nonumber \\
& \qquad + \frac{ ( \det b )^{ \frac12 } }{ 2 } \iiint_{ h \geq \varepsilon } \frac{ h^{ 2 - \gamma } }{ p^0 q^0 } f ( p ) f ( q ) ( ( p'^0 )^r + ( q'^0 )^r - ( p^0 )^r - ( q^0 )^r ) \, d \omega \, d^3 q \, d^3 p . \label{intQmod}
\end{align}
For $ r = 0 $, we obtain
\begin{align}
\frac{ d }{ d s } \int_{ \bbr^3 } f ( p ) \, d^3 p = 0 .
\end{align}
Since $ f $ is non-negative, we have 
\begin{align}
\| f ( s ) \|_{ L^1 } = \| f_0 \|_{ L^1 } , \qquad 0 \leq s \leq T_A . \label{L10}
\end{align}
In the case $ r = - 2 $ we estimate \eqref{intQmod} as follows:
\begin{align}
& \frac{ d }{ d s } \int_{ \bbr^3 } f ( p ) \frac{ 1 }{ ( p^0 )^2 } \, d^3 p \leq 2 \int_{ \bbr^3 } f ( p ) ( p^0 )^{ - 3 } \left| \frac{ \partial p^0 }{ \partial s } \right| \, d^3 p \nonumber \\
& \qquad + \frac{ ( \det b )^{ \frac12 } }{ 2 } \int_{ \bbr^3 } \int_{ \bbr^3 } \int_{ \bbs^2 } \frac{ h^{ 2 - \gamma } }{ p^0 q^0 } f ( p ) f ( q ) \left( \frac{ 1 }{ ( p'^0 )^2 } + \frac{ 1 }{ ( q'^0 )^2 } \right) \, d \omega \, d^3 q \, d^3 p .
\end{align}
The first quantity on the right hand side is estimated by using \eqref{dp ds}:
\begin{align}
2 \int_{ \bbr^3 } | f ( p ) | ( p^0 )^{ - 3 } \left| \frac{ \partial p^0 }{ \partial s } \right| \, d^3 p \leq A^2 \| f \|_{ L^1_{ - 2 } } .
\end{align}
The second quantity is estimated by \eqref{l3} of Lemma \ref{lem p'}:
\begin{align}
\frac{ ( \det b )^{ \frac12 } }{ 2 } \iiint \frac{ h^{ 2 - \gamma } }{ p^0 q^0 } f ( p ) f ( q ) \left( \frac{ 1 }{ ( p'^0 )^2 } + \frac{ 1 }{ ( q'^0 )^2 } \right) \, d \omega \, d^3 q \, d^3 p \leq C ( \det b )^{ \frac12 } \iint \frac{ h^{ - \gamma } }{ p^0 q^0 } f ( p ) f ( q ) \, d^3 q \, d^3 p ,
\end{align}
and we apply \eqref{L infty} to obtain 
\begin{align}
& ( \det b )^{ \frac12 } \iint \frac{ h^{ - \gamma } }{ p^0 q^0 } f ( p ) f ( q ) \, d^3 q \, d^3 p \nonumber \\
& \leq C_A ( \det b )^{ \frac12 } \| f \|_{ L^\infty_\eta } \iint \frac{ h^{ - \gamma } }{ p^0 q^0 } f ( p ) \frac{ 1 }{ q^0 } \exp ( - s_\eta^{ - 1 } q^0 ) \, d^3 q \, d^3 p \nonumber \\
& \leq C_A ( \det b )^{ \frac12 } \| f \|_{ L^\infty_\eta } \iint \frac{ f ( p ) }{ ( p^0 )^{ 1 + \frac{ \gamma }{ 2 } } ( q^0 )^{ 2 + \frac{ \gamma }{ 2 } } \sin^\gamma ( \phi / 2 ) } \exp ( - s_\eta^{ - 1 } q^0 ) \, d^3 q \, d^3 p ,
\end{align}
where we used \eqref{l1}. The integration over $ \bbr^3_q $ can be explicitly computed as follows:
\begin{align}
& ( \det b )^{ \frac12 } \int_{ \bbr^3 } \frac{ 1 }{ ( q^0 )^{ 2 + \frac{ \gamma }{ 2 } } \sin^\gamma ( \phi / 2 ) } \exp ( - s_\eta^{ - 1 } q^0 ) \, d^3 q \nonumber \\
& = \int_{ \bbr^3 } \frac{ 1 }{ | \hat{ q } |^{ 2 + \frac{ \gamma }{ 2 } } \sin^\gamma ( \phi / 2 ) } \exp ( - s_\eta^{ - 1 } | \hat{ q } | ) \, d^3 \hat{ q } \nonumber \\
& = 2 \pi \int_0^\infty \int_0^\pi \frac{ \sin \phi }{ r^{ \frac{ \gamma }{ 2 } } \sin^\gamma ( \phi / 2 ) } \exp ( - s_\eta^{ - 1 } r ) \, d \phi \, d r \nonumber \\
& \leq C \int_0^\infty \frac{ 1 }{ r^{ \frac{ \gamma }{ 2 } } } \exp ( - s_\eta^{ - 1 } r ) \, d r \nonumber \\
& \leq C s_\eta^{ 1 - \frac{ \gamma }{ 2 } } ,
\end{align}
where $ C $ is a constant depending on $ \gamma $. Hence, we obtain
\begin{align}
\frac{ d }{ d s } \| f \|_{ L^1_{ - 2 } } & \leq A^2 \| f \|_{ L^1_{ - 2 } } + C_A s_\eta^{ 1 - \frac{ \gamma }{ 2 } } \| f \|_{ L^\infty_\eta } \| f \|_{ L^1_{ - 1 - \gamma / 2 } } \leq C_A \left( \| f \|_{ L^1 } + \| f \|_{ L^1_{ - 2 } } \right) .
\end{align}
Together with \eqref{L10} we conclude that the norm $ \| f \|_{ L^1_{ - 2 } } $ is bounded on a time interval, which we still denote by $ [ 0 , T_A ] $, so that we have
\[
\sup_{ 0 \leq s \leq T_A } \| f ( s ) \|_{ L^1_{ -2 } } \leq C_A ,
\]
which completes the proof.
\end{proof}

\subsection{Existence for the Boltzmann equation}\label{ExistBoltzmann}
In this part we prove the existence for the Boltzmann equation \eqref{B}. Proposition \ref{prop mod} shows that the modified equation \eqref{Bmod} with $ \varepsilon > 0 $ has a solution in $ L^1_{ - 2 } ( \bbr^3 ) $. We now take the limit $ \varepsilon \to 0 $ and obtain a solution to the Boltzmann equation \eqref{B}. This will be done in $ L^1_{ - 1 } ( \bbr^3 ) $.

\begin{prop}\label{prop B}
Suppose that there exist a time interval $ [ 0 , T ] $ and a positive constant $ A $ such that $ b^{ i j } $ and $ K_{ i j } $ are defined on $ [ 0 , T ] $ and satisfy
\begin{align}
\sup_{ 0 \leq s \leq T } \| b ( s ) \| \leq A , \qquad \sup_{ 0 \leq s \leq T } \| K ( s ) \| \leq A . 
\end{align}
Then, for any initial data $ 0 \leq f_0 \in L^1_1 ( \bbr^3 ) \cap L^1_{ - 2 } ( \bbr^3 ) \cap L^\infty_\eta ( \bbr^3 ) $ for some $ 0 < \eta < 2 / A^2 $, there exists $ 0 < T_A \leq T $ such that the Boltzmann equation \eqref{B} has a unique non-negative solution $ f \in C^1 ( [ 0 , T_A ] ; L^1 ( \bbr^3 ) ) $ satisfying
\begin{align}
\sup_{ 0 \leq s \leq T_A } \| f ( s ) \|_{ L^1_{ - 1 } } + \sup_{ 0 \leq s \leq T_A } \| f ( s ) \|_{ L^1_1 } + \sup_{ 0 \leq s \leq T_A } \| f ( s ) \|_{ L^\infty_\eta } & \leq C_A , \label{prop B1} \\
\sup_{ 0 \leq s \leq T_A } \left\| \frac{ \partial f }{ \partial s } ( s ) \right\|_{ L^1_{ - 1 } } + \sup_{ 0 \leq s \leq T_A } \left\| \frac{ \partial f }{ \partial s } ( s ) \right\|_{ L^1_1 } & \leq C_{A} , \label{prop B2} 
\end{align}
where $ T_A $ and $ C_A $ are positive constants depending on $ A $ and $ \eta $.
\end{prop}
\begin{proof}
Since the proof is almost the same as the proof of Theorem 1 of \cite{LNT1}, we only give a sketch of it. Let $ f_k $ denote the solution of \eqref{Bmod} with $ \varepsilon = k^{ - 1 } $ for $ k \in \bbn $. We will estimate $ \| f_k - f_m \|_{ L^1 } $ and $ \| f_k - f_m \|_{ L^1_{ - 1 } } $ for $ k < m $. Using \eqref{Bmod} we first obtain
\begin{align}
& \partial_s | f_k - f_m | ( p ) \nonumber \\
& \leq \frac{ ( \det b )^{ \frac12 } }{ 2 } \iint_{ h \geq k^{ - 1 } } \frac{ h^{ 2 - \gamma } }{ p^0 q^0 } \{ ( f_k + f_m ) ( p' ) | f_k - f_m | ( q' ) + ( f_k + f_m ) ( q' ) | f_k - f_m | ( p' ) \nonumber \\
& \qquad \qquad \qquad \qquad \qquad + ( f_k + f_m ) ( p ) | f_k - f_m | ( q ) - ( f_k + f_m ) ( q ) | f_k - f_m | ( p ) \} \, d \omega \, d^3 q \nonumber \\
& \quad + ( \det b )^{ \frac12 } \iint_{ m^{ - 1 } \leq h \leq k^{ - 1 } } \frac{ h^{ 2 - \gamma } } { p^0 q^0 } ( f_m ( p' ) f_m ( q' ) + f_m ( p ) f_m ( q ) ) \, d \omega \, d^3 q ,
\end{align}
and multiplying both sides by $ ( p^0 )^r $ and integrating it over $ \bbr^3_p $ we have
\begin{align}
& \frac{ d }{ d s } \| f_k - f_m \|_{ L^1_r } - \int_{ \bbr^3 } | f_k - f_m | ( p ) \left( r ( p^0 )^{ r - 1 } \frac{ \partial p^0 }{ \partial s } \right) \, d^3 p \nonumber \\
& \leq \frac{ ( \det b )^{ \frac12 } }{ 2 } \iiint_{ h \geq k^{ - 1 } } \frac{ h^{ 2 - \gamma } } { p^0 q^0 } \{ ( f_k + f_m ) ( p' ) | f_k - f_m | ( q' ) + ( f_k + f_m ) ( q' ) | f_k - f_m | ( p' ) \nonumber \\
& \qquad \qquad \qquad \qquad \qquad + ( f_k + f_m ) ( p ) | f_k - f_m | ( q ) - ( f_k + f_m ) ( q ) | f_k - f_m | ( p ) \} ( p^0 )^r \, d \omega \, d^3 q \, d^3 p \nonumber \\
& \quad + ( \det b )^{ \frac12 } \iiint_{ m^{ - 1 } \leq h \leq k^{ - 1 } } \frac{ h^{ 2 - \gamma } } { p^0 q^0 } ( f_m ( p' ) f_m ( q' ) + f_m ( p ) f_m ( q ) ) ( p^0 )^r \, d \omega \, d^3 q \, d^3p \nonumber \\
& = \frac{ ( \det b )^{ \frac12 } }{ 2 } \iiint_{ h \geq k^{ - 1 } } \frac{ h^{ 2 - \gamma } }{ p^0 q^0 } ( f_k + f_m ) ( p ) | f_k - f_m | ( q ) \{ ( p'^0 )^r + ( q'^0 )^r + ( p^0 )^r - ( q^0 )^r \} \, d \omega \, d^3 q \, d^3 p \nonumber \\
& \quad + ( \det b )^{ \frac12 } \iiint_{ m^{ - 1 } \leq h \leq k^{ - 1 } } \frac{ h^{ 2 - \gamma } }{ p^0 q^0 } f_m ( p ) f_m ( q ) ( ( p'^0 )^r + ( p^0 )^r ) \, d \omega \, d^3 q \, d^3 p .
\end{align}
Now, for $ r = 0 $ we have
\begin{align}
& \frac{ d }{ d s } \| f_k - f_m \|_{ L^1 } \leq I_1 + I_2 ,
\end{align}
where
\begin{align}
I_1 & = ( \det b )^{ \frac12 } \int_{ \bbr^3 } \int_{ \bbr^3 } \int_{ \bbs^2 } \frac{ h^{ 2 - \gamma } }{ p^0 q^0 } ( f_k + f_m ) ( p ) | f_k - f_m | ( q ) \, d \omega \, d^3 q \, d^3 p , \\
I_2 & = 2 ( \det b )^{ \frac12 } \iiint_{ h \leq k^{ - 1 } } \frac{ h^{ 2 - \gamma } }{ p^0 q^0 } f_m ( p ) f_m ( q ) \, d \omega \, d^3 q \, d^3 p .
\end{align}
The estimates of $ I_1 $ and $ I_2 $ are the same as the estimates (36) and (37) of \cite{LNT1}:
\begin{align}
I_1 & \leq C ( \det b )^{ \frac12 } \iint \frac{ 1 }{ ( p^0 )^{ \frac{ \gamma }{ 2 } } ( q^0 )^{ \frac{ \gamma }{ 2 } } } ( f_k + f_m ) ( p ) | f_k - f_m | ( q ) \, d^3 q \, d^3 p \nonumber \\
& \leq C_A \sup_k \| f_k \|_{ L^1_{ - \gamma / 2 } } \| f_k - f_m \|_{ L^1_{ - \gamma / 2 } } , \label{L101}
\end{align}
and
\begin{align}
I_2 & \leq C ( \det b )^{ \frac12 } k^{ - 2 + \gamma } \int_{ \bbr^3 } \int_{ \bbr^3 } \frac{ 1 }{ p^0 q^0 } f_m ( p ) f_m ( q ) \, d^3 q \, d^3 p \nonumber \\ 
& \leq C_A k^{ - 2 + \gamma } \sup_k \| f_k \|^2_{ L^1_{ - 1 } } . \label{L102}
\end{align}
In a similar way, for $ r = - 1 $ we obtain
\begin{align}
\frac{ d }{ d s } \| f_k - f_m \|_{ L^1_{ - 1 } } \leq J_0 + J_1 + J_2 + J_3 + J_4 + J_5 ,
\end{align}
where
\begin{align}
J_0 & = - \int_{ \bbr^3 } | f_k - f_m | ( p ) \left( ( p^0 )^{ - 2 } \frac{ \partial p^0 }{ \partial s } \right) \, d^3 p, \\
J_1 & = \frac{ ( \det b )^{ \frac12 } }{ 2 } \int_{ \bbr^3 } \int_{ \bbr^3 } \int_{ \bbs^2 } \frac{ h^{ 2 - \gamma } }{ p^0 q^0 } ( f_k + f_m ) ( p ) | f_k - f_m | ( q ) \frac{ 1 }{ p'^0 } \, d \omega \, d^3 q \, d^3 p , \\
J_2 & = \frac{ ( \det b )^{ \frac12 } }{ 2 } \int_{ \bbr^3 } \int_{ \bbr^3 } \int_{ \bbs^2 } \frac{ h^{ 2 - \gamma } }{ p^0 q^0 } ( f_k + f_m ) ( p ) | f_k - f_m | ( q ) \frac{ 1 }{ q'^0 } \, d \omega \, d^3 q \, d^3 p , \\
J_3 & = \frac{ ( \det b )^{ \frac12 } }{ 2 } \int_{ \bbr^3 } \int_{ \bbr^3 } \int_{ \bbs^2 } \frac{ h^{ 2 - \gamma } }{ p^0 q^0 } ( f_k + f_m ) ( p ) | f_k - f_m | ( q ) \frac{ 1 }{ p^0 } \, d \omega \, d^3 q \, d^3 p , \\
J_4 & = ( \det b )^{ \frac12 } \iiint_{ h \leq k^{ - 1 } } \frac{ h^{ 2 - \gamma } }{ p^0 q^0 } f_m ( p ) f_m ( q ) \frac{ 1 }{ p'^0 } \, d \omega \, d^3 q \, d^3 p , \\
J_5 & = ( \det b )^{ \frac12 } \iiint_{ h \leq k^{ - 1 } } \frac{ h^{ 2 - \gamma } }{ p^0 q^0 } f_m ( p ) f_m ( q ) \frac{ 1 }{ p^0 } \, d \omega \, d^3 q \, d^3 p .
\end{align}
The quantity $\partial p^0 / \partial s$ in $J_0$ is estimated by \eqref{dp ds} so that we have
\begin{align}
| J_0 | \leq \frac{ A^2 }{ 2 } \| f_k - f_m \|_{ L^1_{ - 1 } } . \label{L110}
\end{align}
The estimates of $ J_1 $, $ J_2 $, $ J_3 $, $ J_4 $ and $ J_5 $ are the same as the estimates (38)--(42) of \cite{LNT1}. For $ J_3 $ and $ J_5 $ we have
\begin{align}
J_3 & \leq C ( \det b )^{ \frac12 } \iint \frac{ 1 }{ ( p^0 )^{ 1 + \frac{ \gamma }{ 2 } } ( q^0 )^{ \frac{ \gamma }{ 2 } } } ( f_k + f_m ) ( p ) | f_k - f_m | ( q ) \, d^3 q \, d^3 p \nonumber \\
& \leq C_A \sup_k \| f_k \|_{ L^1_{ - 1 - \gamma / 2 } } \| f_k - f_m \|_{ L^1_{ - \gamma / 2 } } , \label{L113}
\end{align}
and
\begin{align}
J_5 & \leq C_A k^{ - 2 + \gamma } \sup_k \| f_k \|_{ L^1_{ - 2 } } \sup_m \| f_m \|_{ L^1_{ - 1 } } . \label{L115}
\end{align}
For $ J_1 $, $ J_2 $ and $ J_4 $ we need to choose $ 0 < \delta < 1 $ satisfying
\begin{align}
\gamma + \delta < 2 ,
\end{align}
and applying Young's inequality to the the estimate \eqref{l2} of Lemma \ref{lem p'} we obtain
\begin{align}
\int_{ \bbs^2 } \frac{ 1 }{ p'^0 } \, d \omega \leq \frac{ C }{ h^\delta ( n^0 )^{ 1 - \delta } } \leq \frac{ C }{ h^\delta ( p^0 )^{ \frac{ \gamma - \delta }{ 2 } } ( q^0 )^{ \frac{ 2 - \gamma - \delta }{ 2 } } } .
\end{align}
Applying this to $ J_1 $ we obtain the following estimate:
\begin{align}
J_1 & \leq C ( \det b )^{ \frac12 } \iint \frac{ h^{ 2 - \gamma - \delta } }{ ( p^0 )^{ 1 + \frac{ \gamma - \delta }{ 2 } } ( q^0 )^{ 1 + \frac{ 2 - \gamma - \delta }{ 2 } } } ( f_k + f_m ) ( p ) | f_k - f_m | ( q ) \, d^3 q \, d^3 p \nonumber \\
& \leq C_A \iint \frac{ 1 }{ ( p^0 )^\gamma q^0 } ( f_k + f_m ) ( p ) | f_k - f_m | ( q ) \, d^3 q \, d^3 p \nonumber \\
& \leq C_A \sup_k \| f_k \|_{ L^1_{ - \gamma } } \| f_k - f_m \|_{ L^1_{ - 1 } } . \label{L111}
\end{align}
We obtain the same estimate for $ J_2 $. For $ J_4 $ we have the following estimate:
\begin{align}
J_4 & \leq C ( \det b )^{ \frac12 } \iint_{ h \leq k^{ - 1 } } \frac{ h^{ 2 - \gamma - \delta } }{ ( p^0 )^{ 1 + \frac{ \gamma - \delta }{ 2 } } ( q^0 )^{ 1 + \frac{ 2 - \gamma - \delta }{ 2 } } } f_m ( p ) f_m ( q ) \, d^3 q \, d^3 p \nonumber \\
& \leq C_A k^{ - 2 + \gamma + \delta } \int_{ \bbr^3 } \int_{ \bbr^3 } \frac{ 1 }{ ( p^0 )^{ 1 + \frac{ \gamma - \delta }{ 2 } } ( q^0 )^{ 1 + \frac{ 2 - \gamma - \delta }{ 2 } } } f_m ( p ) f_m ( q ) \, d^3 q \, d^3 p \nonumber \\
&\leq C_A k^{ - 2 + \gamma + \delta } \sup_k \| f_k \|_{ L^1_{ - 1 - ( \gamma - \delta ) / 2 } } \sup_m \| f_m \|_{ L^1_{ - 2 + ( \gamma + \delta ) / 2 } } . \label{L114}
\end{align}
Here, we note that $ - 2 < - 1 - ( \gamma - \delta ) / 2 < - 1 $ and $ - 3 / 2 < - 2 + ( \gamma + \delta ) / 2 < - 1 $. We now combine the estimates \eqref{L101}--\eqref{L114} and apply Proposition \ref{prop mod} to obtain the following Gr{\"o}nwall type inequality:
\begin{align}
& \frac{ d }{ d s } \left( \| f_k - f_m \|_{ L^1 } + \| f_k - f_m \|_{ L^1_{ - 1 } } \right) \leq C_A \left( k^{ - 2 + \gamma + \delta } + \| f_k - f_m \|_{ L^1 } + \| f_k - f_m \|_{ L^1_{ - 1 } } \right) ,
\end{align}
and conclude that the following holds on a time interval, which we still denote by $ [ 0 , T_A ] $:
\begin{align}
\sup_{ 0 \leq s \leq T_A } \| ( f_k - f_m ) ( s ) \|_{ L^1 } + \| ( f_k - f_m ) ( s ) \|_{ L^1_{ - 1 } } \leq C_A k^{ - 2 + \gamma + \delta } .
\end{align}
This shows that the sequence $ f_k $ converges in $ L^1 \cap L^1_{ - 1 } $ as $ k \to \infty $, and we obtain a solution of the Boltzmann equation \eqref{B}.

Finally, we consider the estimates \eqref{prop B1} and \eqref{prop B2}. Note that we already obtained the boundedness of $ \| f \|_{ L^1 } $ and $ \| f \|_{ L^1_{ - 1 } } $. For $ \| f \|_{ L^1_1 } $ we use \eqref{intQmod} with the restriction $ h \geq \varepsilon $ removed. By the energy conservation \eqref{conservation} and the estimate \eqref{dp ds} we obtain
\begin{align}
\frac{ d }{ d s } \int_{ \bbr^3 } f ( p ) p^0 \, d^3 p = \int_{ \bbr^3 } f ( p ) \frac{ \partial p^0 }{ \partial s } \, d^3 p \leq \frac{ A^2 }{ 2 } \int_{ \bbr^3 } f ( p ) p^0 \, d^3 p .
\end{align}
This shows that $ \| f \|_{ L^1_1 } $ is bounded on $ [ 0 , T_A ] $. The boundedness of $ \| f \|_{ L^\infty_\eta } $ is proved by the same argument as in \eqref{L infty}, and we obtain the estimates \eqref{prop B1}. For the boundedness of $ \| \partial f / \partial s \|_{ L^1_r } $ we consider $ \| Q_\pm ( f , f ) \|_{ L^1_r } $, where the gain and loss terms should be estimated separately. For $ r = 1 $, the loss term is estimated as follows: 
\begin{align}
\| Q_- ( f , f ) \|_{ L^1_1 } & = ( \det b )^{ \frac12 } \iiint \frac{ h^{ 2 - \gamma } }{ p^0 q^0 } f ( p ) f ( q ) p^0 \, d \omega \, d^3 q \, d^3 p \nonumber \\
& \leq C_A \iint ( p^0 )^{ 1 - \frac{ \gamma }{ 2 } } ( q^0 )^{ - \frac{ \gamma }{ 2 } } f ( p ) f ( q ) \, d^3 q \, d^3 p \nonumber \\ 
& \leq C_A \left( \| f \|_{ L^1_{ - 1 } } + \| f \|_{ L^1_1 } \right)^2 , 
\end{align}
by an interpolation, since $ 1 < \gamma < 2 $. The gain term can be written as
\begin{align}
\| Q_+ ( f , f ) \|_{ L^1_1 } & = ( \det b )^{ \frac12 } \iiint \frac{ h^{ 2 - \gamma } }{ p^0 q^0 } f ( p' ) f ( q' ) p^0 \, d \omega \, d^3 q \, d^3 p \nonumber \\
& = ( \det b )^{ \frac12 } \iiint \frac{ h^{ 2 - \gamma } }{ p^0 q^0 } f ( p ) f ( q ) p'^0 \, d \omega \, d^3 q \, d^3 p .
\end{align}
Since $ p'^0 \leq p^0 + q^0 $, we have the same estimate as for the loss term. For $ r = - 1 $, we need to use \eqref{L infty}, which still holds for the original case \eqref{B}. For the gain term, we apply \eqref{l2} of Lemma \ref{lem p'} with $ \delta = 2 - \gamma $ to obtain
\begin{align}
\| Q_+ ( f , f ) \|_{ L^1_{ - 1 } } & = ( \det b )^{ \frac12 } \iiint \frac{ h^{ 2 - \gamma } }{ p^0 q^0 } f ( p ) f ( q ) \frac{ 1 }{ p'^0 } \, d \omega \, d^3 q \, d^3 p \nonumber \\
& \leq C_A \iint \frac{ 1 }{ p^0 q^0 ( n^0 )^{ \gamma - 1 } } f ( p ) f ( q ) \, d^3 q \, d^3 p \allowdisplaybreaks \nonumber \\
& \leq C_A \| f \|_{ L^1_{ - 1 } } \| f \|_{ L^\infty_\eta } \int \frac{ 1 }{ ( q^0 )^{ \gamma + 1 } } \exp ( - s^{ - 1 }_\eta q^0 ) \, d^3 q \nonumber \\
& \leq C_A s_\eta^{ 2 - \gamma }  \| f \|_{ L^1_{ - 1 } } \| f \|_{ L^\infty_\eta } , 
\end{align}
which is bounded on $ [ 0 , T_A ] $. For the loss term we simply use Lemma \ref{lem h} to obtain 
\begin{align}
\| Q_- ( f , f ) \|_{L^1_{ - 1 } } & = ( \det b )^{ \frac12 } \iiint \frac{ h^{ 2 - \gamma } }{ p^0 q^0 } f ( p ) f ( q ) \frac{ 1 }{ p^0 } \, d \omega \, d^3 q \, d^3 p \nonumber \\
& \leq C_A \iint \frac{ 1 }{ ( p^0 )^{ 1 + \frac{ \gamma }{ 2 } } ( q^0 )^{ \frac{ \gamma }{ 2 } } } f ( p ) f ( q ) \, d^3 q \, d^3 p \nonumber \\
& \leq C_A \| f \|_{ L^\infty_\eta } \| f \|_{ L^1_{ - \gamma / 2 } } \int \frac{ 1 }{ ( p^0 )^{ 2 + \frac{ \gamma }{ 2 } } } \exp ( - s^{ - 1 }_\eta p^0 ) \, d^3 p \nonumber \\
& \leq C_A s_\eta^{ 1 - \frac{ \gamma }{ 2 } } \| f \|_{ L^\infty_\eta } \| f \|_{ L^1_{ - \gamma / 2 } } , 
\end{align}
which is also bounded. We obtain the estimates \eqref{prop B2}, and this completes the proof.
\end{proof}

\subsection{Continuous dependence on the metric}\label{ContiBoltzmann}
In Proposition \ref{prop B} we showed that the Boltzmann equation has a solution for given $ b^{ i j } $ and $ K_{ i j } $. We further need to show that the solution depends continuously on $ b^{ i j } $ and $ K_{ i j } $ in order to obtain the existence result for the coupled Einstein-Boltzmann system.

Suppose that two metrics $ { b_1 }^{ i j } $ and $ { b_2 }^{ i j } $ are given, so that we have two solutions $ f_1 $ and $ f_2 $ corresponding to $ { b_1 }^{ i j } $ and $ { b_2 }^{ i j } $ respectively, but also have two post-collision momenta $ p_1' $ and $ p_2' $, corresponding to $ { b_1 }^{ i j } $ and $ { b_2 }^{ i j } $ respectively. Hence, we have to estimate $ f_1 ( p'_1 ) - f_2 ( p'_2 ) $. This will be estimated by considering
\begin{align}
f_1 ( p_1' ) - f_1 ( p_2' ) + ( f_1 - f_2 ) ( p_2' ),
\end{align}
so that we have to consider $ p $-derivatives of $ f $ in order to estimate the first two terms above in terms of $ { b_1 }^{ i j } - { b_2 }^{ i j } $. This will be done in Lemma \ref{lem dfdp} in Section \ref{sdfdp}. The continuous dependence will be obtained in Proposition \ref{prop C} in Section \ref{scontidepmet}.

\subsubsection{Estimates of $ p $-derivatives of $ f $}\label{sdfdp}
We first estimate $ \partial p' / \partial p $ in Lemma \ref{lem dp'dp}. This will be used in Lemma \ref{lem dfdp} to estimate $ \partial f / \partial p $.

\begin{lemma} \label{lem dp'dp}
Suppose that there exist a time interval $ [ 0 , T ] $ and a positive constant $ A $ such that $ { b }^{ i j } $ and $ { K }_{ i j } $ are defined on $ [ 0 , T ] $ and satisfy
\begin{align}
\sup_{ 0 \leq s \leq T } \| b ( s ) \| \leq A , \qquad \sup_{ 0 \leq s \leq T } \| K ( s ) \| \leq A , \qquad \inf_{ 0 \leq s \leq T } \det b ( s ) \geq \frac{ 1 }{ A } . 
\end{align}
Then, we have the following estimate on $ [ 0 , T ] $: 
\begin{align}
\left| \frac{ \partial p_i' }{ \partial p_l } \right| \leq C_A \left( 1 + \sqrt{ \frac{ q^0 }{ p^0 } } \right) , 
\end{align}
where $ C_A $ is a positive constant depending on $ A $. 
\end{lemma}
\begin{proof}
Recall that $ p' $ is given in \eqref{p'}: 
\begin{align*}
p'_i & = \frac{1}{2}\left(n_i + h \omega_j { e^j }_i + \frac{ n_l \omega_k \delta^{ j k } { e_j }^l n_i }{ n^0 + h}\right) . 
\end{align*}
By direct calculations, we have 
\begin{align}
\frac{ \partial n_i }{ \partial p_l } & = \delta^l_i , \\ 
\frac{ \partial p^0 }{ \partial p_l } & = \frac{ b^{ i l } p_i }{ p^0 } , \\ 
\frac{ \partial h }{ \partial p_l } & = \frac{ q^0 b^{ i l } }{ h } \left( \frac{ p_i }{ p^0 } - \frac{ q_i }{ q^0 } \right) . 
\end{align}
In a similar way we obtain 
\begin{align}
\frac{ \partial }{ \partial p_l } \left( \frac{ 1 }{ n^0 + h } \right) & = - \frac{ 1 }{ ( n^0 + h )^2 } \left( \frac{ \partial p^0 }{ \partial p_l } + \frac{ \partial h }{ \partial p_l } \right) \nonumber \\ 
& = - \frac{ 1 }{ ( n^0 + h )^2 } \left( \frac{ b^{ i l } p_i }{ p^0 } + \frac{ q^0 b^{ i l } }{ h } \left( \frac{ p_i }{ p^0 } - \frac{ q_i }{ q^0 } \right) \right) . 
\end{align}
Now, by the assumptions on $ \| b \| $ and $ \det b $ we have from \eqref{p^0 1} and \eqref{p^0 2} that there exists $ C_A > 0 $ such that 
\begin{align}
\frac{ 1 }{ C_A } p^0 \leq \langle p \rangle \leq C_A p^0 . 
\end{align}
Hence, we have 
\begin{align}
\left| \frac{ \partial n_i }{ \partial p_l } \right| & \leq C_A , \label{dp dp1} \\ 
\left| \frac{ \partial p^0 }{ \partial p_l } \right| & \leq C_A . \label{dp dp2} 
\end{align}
Moreover, we obtain from Lemma \ref{lem e} that $ { e_j }^i $ and $ { e^i }_j $ are bounded on $ [ 0 , T ] $: 
\begin{align}
| { e_j }^i | \leq C_A , \qquad | { e^i }_j | \leq C_A , 
\end{align}
so that we obtain 
\begin{align}
\left| \frac{ \partial h }{ \partial p_l } \right| & \leq C_A \frac{ q^0 }{ h } \frac{ h }{ \sqrt{ p^0 q^0 } } \leq C_A \sqrt{ \frac{ q^0 }{ p^0 } } , \label{dp dp3} 
\end{align}
where we used \eqref{phat} and \eqref{Wsin}. In a similar way we obtain 
\begin{align}
\left| \frac{ \partial }{ \partial p_l } \left( \frac{ 1 }{ n^0 + h } \right) \right| & \leq \frac{ C_A }{ ( n^0 + h )^2 } \left( 1 + \sqrt{ \frac{ q^0 }{ p^0 } } \right) . 
\end{align}
We collect the above estimates to obtain the desired result.
\end{proof}

We now estimate the $ p $-derivatives of $ f $. In Proposition \ref{prop B} we showed that $ f $ is bounded in $ L^1_{ - 1 } $. In Lemma \ref{lem dfdp} below, we will first show that $ f $ is bounded in $ L^1_{ - 2 - \delta / 2 } $, and use this to prove that $ \partial f / \partial p $ is bounded in $ L^1_{ - 1 - \delta / 2 } $ for a small $ \delta > 0 $.

\begin{lemma} \label{lem dfdp} 
Suppose that there exist a time interval $ [ 0 , T ] $ and a positive constant $ A $ such that $ b^{ i j } $ and $ K_{ i j } $ are defined on $ [ 0 , T ] $ and satisfy
\begin{align}
\sup_{ 0 \leq s \leq T } \| b ( s ) \| \leq A , \qquad \sup_{ 0 \leq s \leq T } \| K ( s ) \| \leq A , \qquad \inf_{ 0 \leq s \leq T } \det b ( s ) \geq \frac{ 1 }{ A } . 
\end{align}
Then, for any initial data $ f_0 \geq 0 $ such that 
\begin{align}
f_0 \in L^1_1 ( \bbr^3 ) \cap L^1_{ - 2 - \delta / 2 } ( \bbr^3 ) \cap L^\infty_\eta ( \bbr^3 ) , \qquad \frac{ \partial f_0 }{ \partial p } \in L^1_1 ( \bbr^3 ) \cap L^1_{ - 1 - \delta / 2 } ( \bbr^3 ) , 
\end{align}
where $ 0 < \eta < 2 / A^2 $, and $ \delta > 0 $ is a number satisfying
\begin{align}
\gamma + \delta < 2 , 
\end{align}
there exists $ 0 < T_A \leq T $ such that the Boltzmann equation \eqref{B} has a unique non-negative solution $ f \in C^1 ( [ 0 , T_A ] ; L^1 ( \bbr^3 ) ) $ satisfying \eqref{prop B1}, \eqref{prop B2} and 
\begin{align}
\sup_{ 0 \leq s \leq T_A } \| f ( s ) \|_{ L^1_{ - 2 - \delta / 2 } } + \int_0^{ T_A } \| f ( s ) \|_{ L^1_{ - 2 - ( \gamma + \delta ) / 2 } } \, d s \leq C_A , \label{prop B3} \\ 
\sup_{ 0 \leq s \leq T_A } \left\| \frac{ \partial f }{ \partial p } ( s ) \right\|_{ L^1_{ - 1 - \delta / 2 } } + \sup_{ 0 \leq s \leq T_A } \left\| \frac{ \partial f }{ \partial p } ( s ) \right\|_{ L^1_1 } \leq C_A , \label{prop B4}
\end{align}
where $ T_A $ and $ C_A $ are positive constants depending on $ A $, $ \eta $ and $ \delta $.
\end{lemma}
\begin{proof}
We note that the existence of $ f $ satisfying \eqref{prop B1} and \eqref{prop B2} is given by Proposition \ref{prop B}, so we only need to prove \eqref{prop B3} and \eqref{prop B4}. We use Lemma \ref{lem h} to write the Boltzmann equation \eqref{B} as follows:
\begin{align}
\frac{ \partial f }{ \partial s } + B \frac{ 1 }{ ( p^0 )^{ \frac{ \gamma }{ 2 } } } f ( p ) = ( \det b )^{ \frac12 } \int_{ \bbr^3 } \int_{ \bbs^2 } \frac{ h^{ 2 - \gamma } }{ p^0 q^0 } f ( p' ) f ( q' ) \, d \omega \, d^3 q , \label{B'}
\end{align}
where 
\begin{equation}
B = 2^{ 4 - \gamma } \pi ( \det b )^{ \frac12 } \int_{ \bbr^3 } \frac{ \sin^{ 2 - \gamma } ( \phi / 2 ) }{ ( q^0 )^{ \frac{ \gamma }{ 2 } } } f ( q ) \, d^3 q . \label{B''}
\end{equation}
Note that $ B $ depends on $ p $ through $ \phi $, which is the angle between $ p $ and $ q $, but we have 
\begin{align}
\int f ( q ) \, d^3 q \leq \left( \int \frac{ \sin^{ 2 - \gamma } ( \phi / 2 ) }{ ( q^0 )^{ \frac{ \gamma }{ 2 } } } f ( q ) \, d^3 q \right)^{ \frac12 } \left( \int \frac{ ( q^0 )^{ \frac{ \gamma }{ 2 } } }{ \sin^{ 2 - \gamma } ( \phi / 2 ) } f ( q ) \, d^3 q \right)^{ \frac12 } , 
\end{align}
where the integral on the left hand side is a conserved quantity, and 
\begin{align}
\int \frac{ ( q^0 )^{ \frac{ \gamma }{ 2 } } }{ \sin^{ 2 - \gamma } ( \phi / 2 ) } f ( q ) \, d^3 q \leq C_A \int \frac{ ( q^0 )^{ - 1 + \frac{ \gamma }{ 2 } } e^{ - s^{ - 1 }_\eta q^0 } }{ \sin^{ 2 - \gamma } ( \phi / 2 ) } \, d^3 q \leq C_A . 
\end{align}
Hence, $ B \geq B_* $ for some $ B_* > 0 $ independent of $ p $. It is clear that $ B $ is bounded above by $ \| f \|_{ L^1_{ - \gamma / 2 } } $ so that we have 
\begin{align}
0 < B_* \leq B \leq C_A \| f \|_{ - \gamma / 2 } . \label{B_*}
\end{align}
Let $ \delta > 0 $ be a number satisfying 
\begin{align}
\gamma + \delta < 2 . 
\end{align}
Then, multiplying the equation \eqref{B'} by $ ( p^0 )^{ - 2 - \frac{ \delta }{ 2 } } $ and integrating over $ \bbr^3_p $, we obtain 
\begin{align}
\frac{ d }{ d s } \| f \|_{ L^1_{ - 2 - \delta / 2 } } + B_* \| f \|_{ L^1_{ - 2 - ( \gamma + \delta ) / 2 } } \leq C_A \| f \|_{ L^1_{ - 2 - \delta / 2 } } + C_A \iiint \frac{ h^{ 2 - \gamma } }{ p^0 q^0 } f ( p' ) f ( q' ) \frac{ 1 }{ ( p^0 )^{ 2 + \frac{ \delta }{ 2 } } } \, d \omega \, d^3 q \, d^3 p . 
\end{align}
Here, the integral on the right hand side is estimated as follows: 
\begin{align}
& \iiint \frac{ h^{ 2 - \gamma } }{ p^0 q^0 } f ( p' ) f ( q' ) \frac{ 1 }{ ( p^0 )^{ 2 + \frac{ \delta }{ 2 } } } \, d \omega \, d^3 q \, d^3 p \nonumber \\ 
& \leq C_A \iiint \frac{ h^{ 2 - \gamma } }{ p^0 q^0 } f ( p ) f ( q ) \frac{ 1 }{ ( p'^0 )^{ 2 + \frac{ \delta }{ 2 } } } \, d \omega \, d^3 q \, d^3 p \nonumber \\ 
& \leq C_A \iiint \frac{ h^{ 2 - \gamma } }{ p^0 q^0 } f ( p ) f ( q ) \frac{ 1 }{ ( p'^0 )^{ 2 } } \frac{ ( q'^0 )^{ \frac{ \delta }{ 2 } } }{ ( p'^0 q'^0 )^{ \frac{ \delta }{ 2 } } } \, d \omega \, d^3 q \, d^3 p \nonumber \\ 
& \leq C_A \iint \frac{ h^{ 2 - \gamma } }{ p^0 q^0 } f ( p ) f ( q ) \frac{ 1 }{ ( p'^0 )^{ 2 } } \frac{ ( n^0 )^{ \frac{ \delta }{ 2 } } }{ h^{ \delta } } \, d \omega \, d^3 q \, d^3 p \nonumber \\ 
& \leq C_A \iint \frac{ h^{ 2 - \gamma } }{ p^0 q^0 } f ( p ) f ( q ) \frac{ ( n^0 )^{ \frac{ \delta }{ 2 } } }{ h^{ 2 + \delta } } \, d^3 q \, d^3 p \nonumber \\ 
& \leq C_A \iint \frac{ ( p^0 )^{ \frac{ \delta }{ 2 } } }{ p^0 q^0 h^{ \gamma + \delta } } f ( p ) f ( q ) \, d^3 q \, d^3 p \nonumber \\ 
& \leq C_A \iint \frac{ ( p^0 )^{ \frac{ \delta }{ 2 } } }{ ( p^0 q^0 )^{ 1 + \frac{ \gamma + \delta }{ 2 } } \sin^{ \gamma + \delta } ( \phi / 2 ) } f ( p ) f ( q ) \, d^3 q \, d^3 p \nonumber \\ 
& \leq C_A \iint \frac{ 1 }{ ( p^0 )^{ 1 + \frac{ \gamma }{ 2 } } ( q^0 )^{ 1 + \frac{ \gamma + \delta }{ 2 } } \sin^{ \gamma + \delta } ( \phi / 2 ) } f ( p ) f ( q ) \, d^3 q \, d^3 p \nonumber \\ 
& \leq C_A \int \frac{ 1 }{ ( p^0 )^{ 1 + \frac{ \gamma }{ 2 } } } f ( p ) \int \frac{ e^{ - s^{ - 1 }_\eta q^0 } }{ ( q^0 )^{ 2 + \frac{ \gamma + \delta }{ 2 } } \sin^{ \gamma + \delta } ( \phi / 2 ) } \, d^3 q \, d^3 p \nonumber \\ 
& \leq C_A \| f \|_{ L^1_{ - 1 - \gamma / 2 } } ,
\end{align}
where we used $ \gamma + \delta < 2 $ in the last inequality. Hence, we obtain by an interpolation
\begin{align}
\frac{ d }{ d s } \| f \|_{ L^1_{ - 2 - \delta / 2 } } + B_* \| f \|_{ L^1_{ - 2 - ( \gamma + \delta ) / 2 } } \leq C_A \left( 1 + \| f \|_{ L^1_{ - 2 - \delta / 2 } } \right) , 
\end{align}
which shows that $ \| f \|_{ L^1_{ - 2 - \delta / 2 } } $ is bounded on $ [ 0 , T_A ] $. Moreover, the following is also bounded on $ [ 0 , T_A ] $: 
\begin{align}
\int_0^t  \| f ( s ) \|_{ L^1_{ - 2 - ( \gamma + \delta ) / 2 } } \, d s \leq C_A . 
\end{align}
This completes the proof of \eqref{prop B3}. 

We now consider the $ p $ derivatives of $ f $. We use the original equation \eqref{B} to obtain 
\begin{align}
& \frac{ \partial }{ \partial s } \left( \frac{ \partial f }{ \partial p_l } \right) + ( \det b )^{ \frac12 } \iint \frac{ h^{ 2 - \gamma } }{ p^0 q^0 } \frac{ \partial f }{ \partial p_l } ( p ) f ( q ) \, d \omega \, d^3 q \nonumber \\ 
& = ( \det b )^{ \frac12 } \iint \frac{ \partial }{ \partial p_l } \left( \frac{ h^{ 2 - \gamma } }{ p^0 q^0 } \right) f ( p' ) f ( q' ) \, d \omega \, d^3 q \nonumber \\ 
& \quad + ( \det b )^{ \frac12 } \iint \frac{ h^{ 2 - \gamma } }{ p^0 q^0 } \left( \frac{ \partial p'_i }{ \partial p_l } \frac{ \partial f }{ \partial p_i } ( p' ) f ( q' ) + f ( p' ) \frac{ \partial q'_i }{ \partial p_l } \frac{ \partial f }{ \partial p_i } ( q' ) \right) \, d \omega \, d^3 q \nonumber \\ 
& \quad - ( \det b )^{ \frac12 } \iint \frac{ \partial }{ \partial p_l } \left( \frac{ h^{ 2 - \gamma } }{ p^0 q^0 } \right) f ( p ) f ( q ) \, d \omega \, d^3 q , 
\end{align}
for $ l = 1 , 2 , 3 $. Multiplying the above by $ ( p^0 )^r $, we obtain 
\begin{align}
& \frac{ \partial }{ \partial s } \left( ( p^0 )^r \frac{ \partial f }{ \partial p_l } \right) - \frac{ \partial ( p^0 )^r }{ \partial s } \frac{ \partial f }{ \partial p_l } + ( \det b )^{ \frac12 } \iint \frac{ h^{ 2 - \gamma } }{ p^0 q^0 } \frac{ \partial f }{ \partial p_l } ( p ) f ( q ) ( p^0 )^r \, d \omega \, d^3 q \nonumber \\ 
& = ( \det b )^{ \frac12 } \iint \frac{ \partial }{ \partial p_l } \left( \frac{ h^{ 2 - \gamma } }{ p^0 q^0 } \right) f ( p' ) f ( q' ) ( p^0 )^r \, d \omega \, d^3 q \nonumber \\ 
& \quad + ( \det b )^{ \frac12 } \iint \frac{ h^{ 2 - \gamma } }{ p^0 q^0 } \left( \frac{ \partial p'_i }{ \partial p_l } \frac{ \partial f }{ \partial p_i } ( p' ) f ( q' ) + f ( p' ) \frac{ \partial q'_i }{ \partial p_l } \frac{ \partial f }{ \partial p_i } ( q' ) \right) ( p^0 )^r \, d \omega \, d^3 q \nonumber \\ 
& \quad - ( \det b )^{ \frac12 } \iint \frac{ \partial }{ \partial p_l } \left( \frac{ h^{ 2 - \gamma } }{ p^0 q^0 } \right) f ( p ) f ( q ) ( p^0 )^r \, d \omega \, d^3 q . 
\end{align}
Multiplying the above again by $ ( \partial f / \partial p_l ) / | \partial f / \partial p_l | $, we obtain 
\begin{align}
& \frac{ \partial }{ \partial s } \left| ( p^0 )^r \frac{ \partial f }{ \partial p_l } \right| - \frac{ \partial ( p^0 )^r }{ \partial s } \left| \frac{ \partial f }{ \partial p_l } \right| + ( \det b )^{ \frac12 } \iint \frac{ h^{ 2 - \gamma } }{ p^0 q^0 } \left| \frac{ \partial f }{ \partial p_l } ( p ) \right| f ( q ) ( p^0 )^r \, d \omega \, d^3 q \nonumber \\ 
& \leq ( \det b )^{ \frac12 } \iint \left| \frac{ \partial }{ \partial p_l } \left( \frac{ h^{ 2 - \gamma } }{ p^0 q^0 } \right) \right| f ( p' ) f ( q' ) ( p^0 )^r \, d \omega \, d^3 q \nonumber \\ 
& \quad + ( \det b )^{ \frac12 } \iint \frac{ h^{ 2 - \gamma } }{ p^0 q^0 } \left( \left| \frac{ \partial p'_i }{ \partial p_l } \right| \left| \frac{ \partial f }{ \partial p_i } ( p' ) \right| f ( q' ) + f ( p' ) \left| \frac{ \partial q'_i }{ \partial p_l } \right| \left| \frac{ \partial f }{ \partial p_i } ( q' ) \right| \right) ( p^0 )^r \, d \omega \, d^3 q \nonumber \\ 
& \quad + ( \det b )^{ \frac12 } \iint \left| \frac{ \partial }{ \partial p_l } \left( \frac{ h^{ 2 - \gamma } }{ p^0 q^0 } \right) \right| f ( p ) f ( q ) ( p^0 )^r \, d \omega \, d^3 q . 
\end{align}
Applying Lemma \ref{lem dp'dp} together with \eqref{dp dp2} and \eqref{dp dp3}, we obtain 
\begin{align}
& \frac{ \partial }{ \partial s } \left| ( p^0 )^r \frac{ \partial f }{ \partial p_l } \right| + ( \det b )^{ \frac12 } \iint \frac{ h^{ 2 - \gamma } }{ p^0 q^0 } \left| \frac{ \partial f }{ \partial p_l } ( p ) \right| f ( q ) ( p^0 )^r \, d \omega \, d^3 q \nonumber \\ 
& \leq C_A ( p^0 )^r \left| \frac{ \partial f }{ \partial p_l } \right| + C_A ( \det b )^{ \frac12 } \iint \left( h^{ 1 - \gamma } \sqrt{ \frac{ q^0 }{ p^0 } } \frac{ 1 }{ p^0 q^0 } + h^{ 2 - \gamma } \frac{ 1 }{ ( p^0 )^2 q^0 } \right) f ( p' ) f ( q' ) ( p^0 )^r \, d \omega \, d^3 q \nonumber \\ 
& \quad + C_A ( \det b )^{ \frac12 } \sum_{ i = 1 }^3 \iint \frac{ h^{ 2 - \gamma } }{ p^0 q^0 } \left( 1 + \sqrt{ \frac{ q^0 }{ p^0 } } \right) \left( \left| \frac{ \partial f }{ \partial p_i } ( p' ) \right| f ( q' ) + f ( p' ) \left| \frac{ \partial f }{ \partial p_i } ( q' ) \right| \right) ( p^0 )^r \, d \omega \, d^3 q \nonumber \\ 
& \quad + ( \det b )^{ \frac12 } \iint \left( h^{ 1 - \gamma } \sqrt{ \frac{ q^0 }{ p^0 } } \frac{ 1 }{ p^0 q^0 } + h^{ 2 - \gamma } \frac{ 1 }{ ( p^0 )^2 q^0 } \right) f ( p ) f ( q ) ( p^0 )^r \, d \omega \, d^3 q . 
\end{align}
The second quantity on the left hand side can be written as 
\begin{align}
( \det b )^{ \frac12 } \iint \frac{ h^{ 2 - \gamma } }{ p^0 q^0 } \left| \frac{ \partial f }{ \partial p_l } ( p ) \right| f ( q ) ( p^0 )^r \, d \omega \, d^3 q = B \left| \frac{ \partial f }{ \partial p_l } ( p ) \right| ( p^0 )^{ r - \frac{ \gamma }{ 2 } } , 
\end{align} 
where $ B $ is the quantity defined in \eqref{B''}. Note that 
\begin{align*}
h^{ 2 - \gamma } \frac{ 1 }{ ( p^0 )^2 q^0 } \leq C_A h^{ 1 - \gamma } \frac{ 1 }{ ( p^0 )^{ \frac32 } ( q^0 )^{ \frac12 } } , \qquad 1 + \sqrt{ \frac{ q^0 }{ p^0 } } \leq C_A \sqrt{ \frac{ n^0 }{ p^0 } } . 
\end{align*} 
Then, using \eqref{B_*} and integrating over $ \bbr^3_p $, we have 
\begin{align}
& \frac{ d }{ d s } \left\| \frac{ \partial f }{ \partial p_l } \right\|_{ L^1_r } + B_* \left\| \frac{ \partial f }{ \partial p_l } \right\|_{ L^1_{ r - \gamma / 2 } } \nonumber \\ 
& \leq C_A \left\| \frac{ \partial f }{ \partial p_l } \right\|_{ L^1_r } + C_A \iiint h^{ 1 - \gamma } \frac{ 1 }{ ( p^0 )^{ \frac32 } ( q^0 )^{ \frac12 } } f ( p' ) f ( q' ) ( p^0 )^r \, d \omega \, d^3 q \, d^3 p \nonumber \\ 
& \quad + C_A \sum_{ i = 1 }^3 \iiint \frac{ h^{ 2 - \gamma } }{ p^0 q^0 } \sqrt{ \frac{ n^0 }{ p^0 } } \left( \left| \frac{ \partial f }{ \partial p_i } ( p' ) \right| f ( q' ) + f ( p' ) \left| \frac{ \partial f }{ \partial p_i } ( q' ) \right| \right) ( p^0 )^r \, d \omega \, d^3 q \, d^3 p \nonumber \\ 
& \quad + C_A \iiint h^{ 1 - \gamma } \frac{ 1 }{ ( p^0 )^{ \frac32 } ( q^0 )^{ \frac12 } } f ( p ) f ( q ) ( p^0 )^r \, d \omega \, d^3 q \, d^3 p . 
\end{align}
The last quantity can be estimated as follows: 
\begin{align}
& \iiint h^{ 1 - \gamma } \frac{ 1 }{ ( p^0 )^{ \frac32 } ( q^0 )^{ \frac12 } } f ( p ) f ( q ) ( p^0 )^r \, d \omega \, d^3 q \, d^3 p \nonumber \\ 
& \leq C_A \iint \frac{ ( p^0 )^r }{ ( p^0 )^{ \frac32 + \frac{ \gamma - 1 }{ 2 } } ( q^0 )^{ \frac12 + \frac{ \gamma - 1 }{ 2 } } \sin^{ \gamma - 1 } ( \phi / 2 ) } f ( p ) f ( q ) \, d^3 q \, d^3 p \nonumber \\ 
& \leq C_A \int ( p^0 )^{ r - 1 - \frac{ \gamma }{ 2 } } f ( p ) \int \frac{ e^{ - s^{ - 1 }_\eta q^0 } }{ ( q^0 )^{ \frac{ \gamma }{ 2 } + 1 } \sin^{ \gamma - 1 } ( \phi / 2 ) } \, d^3 q \, d^3 p \nonumber \\ 
& \leq C_A \| f \|_{ L^1_{ r - 1 - \gamma / 2 } } . 
\end{align}
Hence, we obtain 
\begin{align}
& \frac{ d }{ d s } \left\| \frac{ \partial f }{ \partial p_l } \right\|_{ L^1_r } + B_* \left\| \frac{ \partial f }{ \partial p_l } \right\|_{ L^1_{ r - \gamma / 2 } } \nonumber \\ 
& \leq C_A \| f \|_{ L^1_{ r - 1 - \gamma / 2 } } + C_A \left\| \frac{ \partial f }{ \partial p_l } \right\|_{ L^1_r } + C_A \iiint h^{ 1 - \gamma } \frac{ 1 }{ ( p^0 )^{ \frac32 } ( q^0 )^{ \frac12 } } f ( p' ) f ( q' ) ( p^0 )^r \, d \omega \, d^3 q \, d^3 p \nonumber \\ 
& \quad + C_A \sum_{ i = 1 }^3 \iiint \frac{ h^{ 2 - \gamma } }{ p^0 q^0 } \sqrt{ \frac{ n^0 }{ p^0 } } \left( \left| \frac{ \partial f }{ \partial p_i } ( p' ) \right| f ( q' ) + f ( p' ) \left| \frac{ \partial f }{ \partial p_i } ( q' ) \right| \right) ( p^0 )^r \, d \omega \, d^3 q \, d^3 p . \label{df dpr}
\end{align}
We need to consider two cases (1) $ r = - 1 - \delta / 2 $ and (2) $ r = 1 $. \medskip

\noindent (1) In the case $ r = - 1 - \delta / 2 $, we have from \eqref{df dpr} 
\begin{align}
& \frac{ d }{ d s } \left\| \frac{ \partial f }{ \partial p_l } \right\|_{ L^1_{ - 1 - \delta / 2 } } + B_* \left\| \frac{ \partial f }{ \partial p_l } \right\|_{ L^1_{ - 1 - ( \delta + \gamma ) / 2 } } \nonumber \\ 
& \leq C_A \| f \|_{ L^1_{ - 2 - ( \gamma + \delta ) / 2 } } + C_A \left\| \frac{ \partial f }{ \partial p_l } \right\|_{ L^1_{ - 1 - \delta / 2 } } + C_A \iiint h^{ 1 - \gamma } \frac{ 1 }{ ( p^0 )^{ \frac52 + \frac{ \delta }{ 2 } } ( q^0 )^{ \frac12 } } f ( p' ) f ( q' ) \, d \omega \, d^3 q \, d^3 p \nonumber \\ 
& \quad + C_A \sum_{ i = 1 }^3 \iiint \frac{ h^{ 2 - \gamma } }{ p^0 q^0 } \frac{ ( n^0 )^{ \frac12 } }{ ( p^0 )^{ \frac32 + \frac{ \delta }{ 2 } } } \left( \left| \frac{ \partial f }{ \partial p_i } ( p' ) \right| f ( q' ) + f ( p' ) \left| \frac{ \partial f }{ \partial p_i } ( q' ) \right| \right) \, d \omega \, d^3 q \, d^3 p . 
\end{align}
Let us write for simplicity 
\begin{align}
L_1 & = \iiint h^{ 1 - \gamma } \frac{ 1 }{ ( p^0 )^{ \frac52 + \frac{ \delta }{ 2 } } ( q^0 )^{ \frac12 } } f ( p' ) f ( q' ) \, d \omega \, d^3 q \, d^3 p , \\ 
L_2 & = \iiint \frac{ h^{ 2 - \gamma } }{ p^0 q^0 } \frac{ ( n^0 )^{ \frac12 } }{ ( p^0 )^{ \frac32 + \frac{ \delta }{ 2 } } } \left| \frac{ \partial f }{ \partial p_i } ( p' ) \right| f ( q' ) \, d \omega \, d^3 q \, d^3 p , \\ 
L_3 & = \iiint \frac{ h^{ 2 - \gamma } }{ p^0 q^0 } \frac{ ( n^0 )^{ \frac12 } }{ ( p^0 )^{ \frac32 + \frac{ \delta }{ 2 } } } f ( p' ) \left| \frac{ \partial f }{ \partial p_i } ( q' ) \right| \, d \omega \, d^3 q \, d^3 p . 
\end{align}
We first estimate $ L_1 $ as follows: 
\begin{align}
L_1 & \leq C_A \iiint h^{ 1 - \gamma } \frac{ 1 }{ ( p^0 )^{ \frac52 + \frac{ \delta }{ 2 } } ( q^0 )^{ \frac12 } } \frac{ 1 }{ p'^0 } e^{ - s^{ - 1 }_\eta p'^0 } \frac{ 1 }{ q'^0 } e^{ - s^{ - 1 }_\eta q'^0 } \, d \omega \, d^3 q \, d^3 p \nonumber \\ 
& \leq C_A \iiint \frac{ e^{ - s^{ - 1 }_\eta p'^0 } e^{ - s^{ - 1 }_\eta q'^0 } }{ ( p^0 q^0 )^{ \frac{ \gamma - 1 }{ 2 } } \sin^{ \gamma - 1 } ( \phi / 2 ) ( p^0 )^{ \frac52 + \frac{ \delta }{ 2 } } ( q^0 )^{ \frac12 } } \frac{ 1 }{ p'^0 q'^0 } \, d \omega \, d^3 q \, d^3 p \nonumber \\ 
& \leq C_A \iint \frac{ e^{ - s^{ - 1 }_\eta p^0 } e^{ - s^{ - 1 }_\eta q^0 } }{ ( p^0 )^{ 2 + \frac{ \gamma + \delta }{ 2 } } ( q^0 )^{ \frac{ \gamma }{ 2 } } \sin^{ \gamma - 1 } ( \phi / 2 ) } \int_{ \bbs^2 } \frac{ 1 }{ p'^0 q'^0 } \, d \omega \, d^3 q \, d^3 p . 
\end{align}
We use \eqref{l4} of Lemma \ref{lem p'} with $ \epsilon > 0 $ satisfying
\begin{align}
\gamma + \delta + \epsilon < 2 
\end{align}
to obtain 
\begin{align}
L_1 & \leq C_A \iint \frac{ e^{ - s^{ - 1 }_\eta p^0 } e^{ - s^{ - 1 }_\eta q^0 } }{ ( p^0 )^{ 2 + \frac{ \gamma + \delta }{ 2 } } ( q^0 )^{ \frac{ \gamma }{ 2 } } \sin^{ \gamma - 1 } ( \phi / 2 ) } \frac{ 1 }{ h^\epsilon ( n^0 )^{ 2 - \epsilon } } \, d^3 q \, d^3 p \nonumber \\ 
& \leq C_A \iint \frac{ e^{ - s^{ - 1 }_\eta p^0 } e^{ - s^{ - 1 }_\eta q^0 } }{ ( p^0 )^{ 2 + \frac{ \gamma + \delta + \epsilon }{ 2 } } ( q^0 )^{ \frac{ \gamma + \epsilon }{ 2 } + 2 - \epsilon } \sin^{ \gamma + \epsilon - 1 } ( \phi / 2 ) } \, d^3 q \, d^3 p \nonumber \\ 
& \leq C_A , 
\end{align}
since 
\begin{align}
2 + \frac{ \gamma + \delta + \epsilon }{ 2 } < 3 , \qquad \frac{ \gamma + \epsilon }{ 2 } + 2 - \epsilon < 3 , \qquad \gamma + \epsilon - 1 < 1 . 
\end{align}
For $ L_2 $, we use the change of variables \eqref{dpdq} to obtain 
\begin{align}
L_2 & = \iiint \frac{ h^{ 2 - \gamma } }{ p^0 q^0 } \frac{ ( n^0 )^{ \frac12 } }{ ( p'^0 )^{ \frac32 + \frac{ \delta }{ 2 } } } \left| \frac{ \partial f }{ \partial p_i } ( p ) \right| f ( q ) \, d \omega \, d^3 q \, d^3 p . 
\end{align}
Note that 
\begin{align}
\frac{ 1 }{ ( p'^0 )^{ \frac{ 1 + \delta }{ 2 } } } = \frac{ ( q'^0 )^{ \frac{ 1 + \delta }{ 2 } } }{ ( p'^0 q'^0 )^{ \frac{ 1 + \delta }{ 2 } } } \leq \frac{ ( n^0 )^{ \frac{ 1 + \delta }{ 2 } } }{ h^{ 1 + \delta } } . 
\end{align}
Hence, we have 
\begin{align}
L_2 & \leq C_A \iiint \frac{ h^{ 2 - \gamma } }{ p^0 q^0 } \frac{ ( n^0 )^{ \frac12 } }{ p'^0 } \frac{ ( n^0 )^{ \frac{ 1 + \delta }{ 2 } } }{ h^{ 1 + \delta } } \left| \frac{ \partial f }{ \partial p_i } ( p ) \right| f ( q ) \, d \omega \, d^3 q \, d^3 p \nonumber \\ 
& \leq C_A \iint \frac{ h^{ 1 - \gamma - \delta } ( n^0 )^{ 1 + \frac{ \delta }{ 2 } } }{ p^0 q^0 } \left| \frac{ \partial f }{ \partial p_i } ( p ) \right| f ( q ) \int_{ \bbs^2 } \frac{ 1 }{ p'^0 } \, d \omega \, d^3 q \, d^3 p . 
\end{align}
We use \eqref{l2} of Lemma \ref{lem p'} with $ \epsilon > 0 $ satisfying again
\begin{align}
\gamma + \delta + \epsilon < 2
\end{align}
to obtain 
\begin{align}
L_2 & \leq C_A \iint \frac{ h^{ 1 - \gamma - \delta - \epsilon } ( n^0 )^{ \frac{ \delta }{ 2 } + \epsilon } }{ p^0 q^0 } \left| \frac{ \partial f }{ \partial p_i } ( p ) \right| f ( q ) \, d^3 q \, d^3 p \nonumber \\ 
& \leq C_A \iint \frac{ ( n^0 )^{ \frac{ \delta }{ 2 } + \epsilon } }{ ( p^0 q^0 )^{ \frac{ 1 + \gamma + \delta + \epsilon }{ 2 } } \sin^{ - 1 + \gamma + \delta + \epsilon } ( \phi / 2 ) } \left| \frac{ \partial f }{ \partial p_i } ( p ) \right| f ( q ) \, d^3 q \, d^3 p \nonumber \\ 
& \leq C_A \iint \left( \frac{ 1 }{ ( p^0 )^{ \frac{ 1 + \gamma - \epsilon }{ 2 } } ( q^0 )^{ \frac{ 1 + \gamma + \delta + \epsilon }{ 2 } } } + \frac{ 1 }{ ( p^0 )^{ \frac{ 1 + \gamma + \delta + \epsilon }{ 2 } } ( q^0 )^{ \frac{ 1 + \gamma - \epsilon }{ 2 } } } \right) \frac{ 1 }{ \sin^{ - 1 + \gamma + \delta + \epsilon } ( \phi / 2 ) } \left| \frac{ \partial f }{ \partial p_i } ( p ) \right| f ( q ) \, d^3 q \, d^3 p . 
\end{align}
Note that, for sufficiently small $ \epsilon > 0 $, 
\begin{align}
0 < \frac{ 1 + \gamma - \epsilon }{ 2 } < \frac{ 1 + \gamma + \delta + \epsilon }{ 2 } < \frac32 , \qquad 0 < - 1 + \gamma + \delta + \epsilon < 1 . 
\end{align}
Hence, the following are bounded: 
\begin{align}
\int \frac{ 1 }{ ( q^0 )^{ \frac{ 1 + \gamma + \delta + \epsilon }{ 2 } } \sin^{ - 1 + \gamma + \delta + \epsilon } ( \phi / 2 ) } f ( q ) \, d^3 q & \leq C_A , \\ 
\int \frac{ 1 }{ ( q^0 )^{ \frac{ 1 + \gamma - \epsilon }{ 2 } } \sin^{ - 1 + \gamma + \delta + \epsilon } ( \phi / 2 ) } f ( q ) \, d^3 q & \leq C_A . 
\end{align}
Consequently, we obtain by an interpolation
\begin{align}
L_2 & \leq C_A \int \left( \frac{ 1 }{ ( p^0 )^{ \frac{ 1 + \gamma - \epsilon }{ 2 } } } + \frac{ 1 }{ ( p^0 )^{ \frac{ 1 + \gamma + \delta + \epsilon }{ 2 } } } \right) \left| \frac{ \partial f }{ \partial p_i } ( p ) \right| \, d^3 p \nonumber \\ 
& \leq C_A \left\| \frac{ \partial f }{ \partial p_i } \right\|_{ L^1_1 } + C_A \left\| \frac{ \partial f }{ \partial p_i } \right\|_{ L^1_{ - ( 1 + \gamma + \delta + \epsilon ) / 2 } } . 
\end{align}
The estimate for $ L_3 $ is the same as for $ L_2 $. Now, combine all the estimates above to obtain 
\begin{align}
& \frac{ d }{ d s } \left\| \frac{ \partial f }{ \partial p } \right\|_{ L^1_{ - 1 - \delta / 2 } } + B_* \left\| \frac{ \partial f }{ \partial p } \right\|_{ L^1_{ - 1 - ( \delta + \gamma ) / 2 } } \nonumber \\ 
& \leq C_A  + C_A \| f \|_{ L^1_{ - 2 - ( \gamma + \delta ) / 2 } } + C_A \left\| \frac{ \partial f }{ \partial p } \right\|_{ L^1_{ - 1 - \delta / 2 } } + C_A \left\| \frac{ \partial f }{ \partial p } \right\|_{ L^1_1 } + C_A \left\| \frac{ \partial f }{ \partial p } \right\|_{ L^1_{ - ( 1 + \gamma + \delta + \epsilon ) / 2 } } . \label{df dp1-1}
\end{align}
Here, the last term on the right hand side can be controlled by the second term on the left hand side and the third and fourth terms on the right hand side by applying Young's inequality: $ a b \leq \eta a^{ \frac32 } + C_\eta b^3 $ for any small $ \eta > 0 $. To be precise, we apply 
\begin{align}
\frac{ 1}{ ( p^0 )^{ \frac{ 1 + \gamma + \delta + \epsilon }{ 2 } } } = \frac{ 1 }{ ( p^0 )^{ \frac{ 2 + \gamma + \delta }{ 3 } } ( p^0 )^{ \frac{ - 1 + \gamma + \delta + 3 \epsilon }{ 6 } } } \leq \frac{ \eta }{ ( p^0 )^{ \frac{ 2 + \gamma + \delta }{ 2 } } } + \frac{ C_\eta }{ ( p^0 )^{ \frac{ - 1 + \gamma + \delta + 3 \epsilon }{ 2 } } } , 
\end{align}
to obtain 
\begin{align}
\left\| \frac{ \partial f }{ \partial p } \right\|_{ L^1_{ - ( 1 + \gamma + \delta + \epsilon ) / 2 } } \leq \eta \left\| \frac{ \partial f }{ \partial p } \right\|_{ L^1_{ - 1 - ( \gamma + \delta ) / 2 } } + C_\eta \left\| \frac{ \partial f }{ \partial p } \right\|_{ L^1_{ ( 1 - \gamma - \delta - 3 \epsilon ) / 2 } } , 
\end{align} 
and the second term on the right hand side above can be estimated by the third and fourth terms on the right hand side of \eqref{df dp1-1}, since 
\begin{align}
- 1 - \frac{ \delta }{ 2 } < \frac{ 1 - \gamma - \delta - 3 \epsilon }{ 2 } < 1 , 
\end{align}
for small $ \epsilon > 0 $. Finally, we obtain from \eqref{df dp1-1}
\begin{align}
& \frac{ d }{ d s } \left\| \frac{ \partial f }{ \partial p } \right\|_{ L^1_{ - 1 - \delta / 2 } } \leq C_A + C_A \| f \|_{ L^1_{ - 2 - ( \gamma + \delta ) / 2 } } + C_A \left\| \frac{ \partial f }{ \partial p } \right\|_{ L^1_{ - 1 - \delta / 2 } } + C_A \left\| \frac{ \partial f }{ \partial p } \right\|_{ L^1_1 } . \label{df dp1}
\end{align}\medskip

\noindent (2) In the case $ r = 1 $, we have from \eqref{df dpr} 
\begin{align}
& \frac{ d }{ d s } \left\| \frac{ \partial f }{ \partial p_l } \right\|_{ L^1_1 } + B_* \left\| \frac{ \partial f }{ \partial p_l } \right\|_{ L^1_{ 1 - \gamma / 2 } } \nonumber \\  
& \leq C_A \| f \|_{ L^1_{ - \gamma / 2 } } + C_A \left\| \frac{ \partial f }{ \partial p_l } \right\|_{ L^1_1 } + C_A \iiint h^{ 1 - \gamma } \frac{ 1 }{ ( p^0 )^{ \frac12 } ( q^0 )^{ \frac12 } } f ( p' ) f ( q' ) \, d \omega \, d^3 q \, d^3 p \nonumber \\ 
& \quad + C_A \sum_{ i = 1 }^3 \iiint \frac{ h^{ 2 - \gamma } }{ p^0 q^0 } \sqrt{ n^0 p^0 } \left( \left| \frac{ \partial f }{ \partial p_i } ( p' ) \right| f ( q' ) + f ( p' ) \left| \frac{ \partial f }{ \partial p_i } ( q' ) \right| \right) \, d \omega \, d^3 q \, d^3 p . 
\end{align}
Let use write 
\begin{align}
M_1 & = \iiint h^{ 1 - \gamma } \frac{ 1 }{ ( p^0 )^{ \frac12 } ( q^0 )^{ \frac12 } } f ( p' ) f ( q' ) \, d \omega \, d^3 q \, d^3 p , \\ 
M_2 & = \iiint \frac{ h^{ 2 - \gamma } }{ p^0 q^0 } \sqrt{ n^0 p^0 } \left| \frac{ \partial f }{ \partial p_i } ( p' ) \right| f ( q' ) \, d \omega \, d^3 q \, d^3 p , \\ 
M_3 & = \iiint \frac{ h^{ 2 - \gamma } }{ p^0 q^0 } \sqrt{ n^0 p^0 } f ( p' ) \left| \frac{ \partial f }{ \partial p_i } ( q' ) \right| \, d \omega \, d^3 q \, d^3 p . 
\end{align}
For $ M_1 $, we estimate 
\begin{align}
M_1 & = \iiint h^{ 1 - \gamma } \frac{ 1 }{ ( p^0 )^{ \frac12 } ( q^0 )^{ \frac12 } } f ( p' ) f ( q' ) \, d \omega \, d^3 q \, d^3 p \nonumber \\ 
& \leq C_A \iiint h^{ 1 - \gamma } \frac{ n^0 }{ p^0 q^0 } f ( p' ) f ( q' ) \, d \omega \, d^3 q \, d^3 p \nonumber \\ 
& \leq C_A \iiint h^{ 1 - \gamma } \frac{ 1 }{ q^0 } f ( p ) f ( q ) \, d \omega \, d^3 q \, d^3 p \nonumber \\ 
& \leq C_A \iint \frac{ 1 }{ ( p^0 )^{ \frac{ \gamma - 1 }{ 2 } } ( q^0 )^{ \frac{ \gamma + 1 }{ 2 } } \sin^{ \gamma - 1 } ( \phi / 2 ) } f ( p ) f ( q ) \, d^3 q \, d^3 p \nonumber \\ 
& \leq C_A \iint \frac{ e^{ - s^{ - 1 }_\eta p^0 } e^{ - s^{ - 1 }_\eta q^0 } }{ ( p^0 )^{ \frac{ \gamma + 1 }{ 2 } } ( q^0 )^{ \frac{ \gamma + 3 }{ 2 } } \sin^{ \gamma - 1 } ( \phi / 2 ) } \, d^3 q \, d^3 p \nonumber \\ 
& \leq C_A , 
\end{align}
since 
\begin{align}
1 < \frac{ \gamma + 1 }{ 2 } < \frac32 , \qquad 2 < \frac{ \gamma + 3 }{ 2 } < \frac52 , \qquad 0 < \gamma - 1 < 1 . 
\end{align}
For $ M_2 $, we have 
\begin{align}
M_2 & = \iiint \frac{ h^{ 2 - \gamma } }{ p^0 q^0 } \sqrt{ n^0 p^0 } \left| \frac{ \partial f }{ \partial p_i } ( p' ) \right| f ( q' ) \, d \omega \, d^3 q \, d^3 p \nonumber \\ 
& \leq C_A \iiint \frac{ h^{ 2 - \gamma } }{ p^0 q^0 } n^0 \left| \frac{ \partial f }{ \partial p_i } ( p' ) \right| f ( q' ) \, d \omega \, d^3 q \, d^3 p \nonumber \\ 
& \leq C_A \iiint \frac{ h^{ 2 - \gamma } }{ p^0 q^0 } n^0 \left| \frac{ \partial f }{ \partial p_i } ( p ) \right| f ( q ) \, d^3 q \, d^3 p \nonumber \\ 
& \leq C_A \iiint h^{ 2 - \gamma } \left( \frac{ 1 }{ p^0 } + \frac{ 1 }{ q^0 } \right) \left| \frac{ \partial f }{ \partial p_i } ( p ) \right| f ( q ) \, d^3 q \, d^3 p \nonumber \\ 
& \leq C_A \iiint \left( \frac{ ( q^0 )^{ 1 - \frac{ \gamma }{ 2 } } }{ ( p^0 )^{ \frac{ \gamma }{ 2 } } } + \frac{ ( p^0 )^{ 1 - \frac{ \gamma }{ 2 } } }{ ( q^0 )^{ \frac{ \gamma }{ 2 } } } \right) \left| \frac{ \partial f }{ \partial p_i } ( p ) \right| f ( q ) \, d^3 q \, d^3 p \nonumber \\ 
& \leq C_A \left\| \frac{ \partial f }{ \partial p_i } \right\|_{ L^1_{ - \gamma / 2 } } \| f \|_{ L^1_{ 1 - \gamma / 2 } } + C_A  \left\| \frac{ \partial f }{ \partial p_i } \right\|_{ L^1_{ 1 - \gamma / 2 } } \| f \|_{ L^1_{ - \gamma / 2 } } . 
\end{align}
Note that $ \| f \|_{ L^1_{ 1 - \gamma / 2 } } $ and $ \| f \|_{ L^1_{ - \gamma / 2 } } $ are bounded quantities, and the estimate for $ M_3 $ is the same as for $ M_2 $. We combine all the estimates above to obtain 
\begin{align}
& \frac{ d }{ d s } \left\| \frac{ \partial f }{ \partial p } \right\|_{ L^1_1 } + B_* \left\| \frac{ \partial f }{ \partial p } \right\|_{ L^1_{ 1 - \gamma / 2 } } \leq C_A + C_A \left\| \frac{ \partial f }{ \partial p } \right\|_{ L^1_1 } + C_A \left\| \frac{ \partial f }{ \partial p } \right\|_{ L^1_{ - \gamma / 2 } } + C_A  \left\| \frac{ \partial f }{ \partial p } \right\|_{ L^1_{ 1 - \gamma / 2 } } . 
\end{align}
Since $ -1 < - \gamma / 2 < - 1 / 2 $ and $ 0 < 1 - \gamma / 2 < 1 / 2 $, we obtain by an interpolation
\begin{align}
\frac{ d }{ d s } \left\| \frac{ \partial f }{ \partial p } \right\|_{ L^1_1 } \leq C_A + C_A \left\| \frac{ \partial f }{ \partial p } \right\|_{ L^1_1 } + C_A \left\| \frac{ \partial f }{ \partial p } \right\|_{ L^1_{ - 1 - \delta / 2 } } . \label{df dp2}
\end{align}
We now combine \eqref{df dp1} and \eqref{df dp2} to obtain 
\begin{align}
& \left\| \frac{ \partial f }{ \partial p } ( t ) \right\|_{ L^1_{ - 1 - \delta / 2 } } + \left\| \frac{ \partial f }{ \partial p } ( t ) \right\|_{ L^1_1 } \nonumber \\ 
& \leq C_A + C_A \int_0^t \| f ( s ) \|_{ L^1_{ - 2 - ( \gamma + \delta ) / 2 } } \, d s + C_A \int_0^t \left\| \frac{ \partial f }{ \partial p } ( s ) \right\|_{ L^1_{ - 1 - \delta / 2 } } + \left\| \frac{ \partial f }{ \partial p } ( s ) \right\|_{ L^1_1 } \, d s . 
\end{align}
Finally, we apply \eqref{prop B3} and use Gr{\" o}nwall's inequality to obtain the desired result: 
\begin{align}
\sup_{ 0 \leq t \leq T_A } \left\| \frac{ \partial f }{ \partial p } ( t ) \right\|_{ L^1_{ - 1 - \delta / 2 } } + \sup_{ 0 \leq t \leq T_A } \left\| \frac{ \partial f }{ \partial p } ( t ) \right\|_{ L^1_1 } \leq C_A . 
\end{align}
This completes the proof of the lemma. 
\end{proof}

\subsubsection{Continuous dependence on the metric}\label{scontidepmet}
We are now ready to prove the continuous dependence on the metric. Below, we suppose that $ { b_1 }^{ i j } $, $ { K_1 }_{ i j } $, $ { b_2 }^{ i j } $ and $ { K_2 }_{ i j } $ are given. Some quantities like $ p^0 $, $ q^0 $ and $ h $ depend on the metric so that we need to explicitly write the dependence on $ { b_1 }^{ i j } $ or $ { b_2 }^{ i j } $. Hence, we denote
\begin{align}
{ p_1 }^0 = \sqrt{ { b_1 }^{ i j } p_i p_j } , \qquad { q_1 }^0 = \sqrt{ { b_1 }^{ i j } q_i q_j } , \qquad h_1 = \sqrt{ ( { n_1 }^0 )^2 - { b_1 }^{ i j } n_i n_j } , 
\end{align}
and $ { p_1' }^0 $, $ { q_1' }^0 $, $ { p_1' }_i $ and $ { q_1' }_i $ are defined by \eqref{p'^0}--\eqref{q'} using the metric $ { b_1 }^{ i j } $. Similarly, we denote 
\begin{align}
{ p_2 }^0 = \sqrt{ { b_2 }^{ i j } p_i p_j } , \qquad { q_2 }^0 = \sqrt{ { b_2 }^{ i j } q_i q_j } , \qquad h_2 = \sqrt{ ( { n_2 }^0 )^2 - { b_2 }^{ i j } n_i n_j } , 
\end{align}
and $ { p_2' }^0 $, $ { q_2' }^0 $, $ { p_2' }_i $ and $ { q_2' }_i $ are defined by \eqref{p'^0}--\eqref{q'} using the metric $ { b_2 }^{ i j } $. Recall that the norms for $ f $ also depend on the metric. Hence, we need to write 
\begin{align}
& \| f \|_{ L^1_{ 1 , r } } = \int_{ \bbr^3 } | f ( p ) | ( { p_1 }^0 )^r \, d^3 p , \qquad { p_1 }^0 = \sqrt{ { b_1 }^{ i j } p_i p_j } , \\
& \| f \|_{ L^\infty_{ 1 , \eta } } = \sup_{ p \in \bbr^3 } | w_{ 1 , \eta } f ( p ) | , \qquad w_{ 1 , \eta } = { p_1 }^0 \exp ( s^{ - 1 }_\eta { p_1 }^0 ) , \qquad s_\eta = ( s + \eta^2 )^\eta , \qquad \eta > 0 , 
\end{align}
and the norms $ \| f \|_{ L^1_{ 2 , r } } $ and $ \| f \|_{ L^\infty_{ 2 , \eta } } $ are defined similarly.

\begin{lemma}\label{lem HJ}
Let $ A $ and $ B $ be positive definite $ n \times n $ matrices. For any $ 0 < \alpha < 1 $, we have 
\[
\det ( \alpha A + ( 1 - \alpha ) B ) \geq ( \det A )^\alpha ( \det B )^{ 1 - \alpha } , 
\]
where the equality holds, if and only if $ A = B $. 
\end{lemma}
\begin{proof}
See the Corollary 7.6.8.~of Ref.~\cite{HJ} for the proof. 
\end{proof}

\begin{lemma}\label{lem F-F}
Suppose that there exist a time interval $ [ 0 , T ] $ and a positive constant $ A $ such that $ { b_1 }^{ i j } $, $ { K_1 }_{ i j } $, $ { b_2 }^{ i j } $ and $ { K_2 }_{ i j } $ are defined on $ [ 0 , T ] $ and satisfy
\begin{align}
\sup_{ 0 \leq s \leq T } \| b_1 ( s ) \| \leq A , \qquad \sup_{ 0 \leq s \leq T } \| K_1 ( s ) \| \leq A , \qquad \inf_{ 0 \leq s \leq T } \det b_1 ( s ) \geq \frac{ 1 }{ A } , \\ 
\sup_{ 0 \leq s \leq T } \| b_2 ( s ) \| \leq A , \qquad \sup_{ 0 \leq s \leq T } \| K_2 ( s ) \| \leq A , \qquad \inf_{ 0 \leq s \leq T } \det b_2 ( s ) \geq \frac{ 1 }{ A } . 
\end{align}
Then, we have the following estimates on $ [ 0 , T ] $: 
\begin{align}
\left| ( \det b_1 )^{ \frac12 } - ( \det b_2 )^{ \frac12 } \right| & \leq C_A \| b_1 - b_2 \| , \\ 
\left| \frac{ 1 }{ { p_1 }^0 { q_1 }^0 } - \frac{ 1 }{ { p_2 }^0 { q_2 }^0 } \right| & \leq C_A \frac{ \| b_1 - b_2 \| }{ \langle p \rangle \langle q \rangle } , \\ 
\left| { h_1 }^{ 2 - \gamma } - { h_2 }^{ 2 - \gamma } \right| & \leq C_A \| b_1 - b_2 \| \langle p \rangle^{ 1 - \frac{ \gamma }{ 2 } } \langle q \rangle^{ 1 - \frac{ \gamma }{ 2 } } , 
\end{align}
where $ C_A $ is a positive constant depending on $ A $. 
\end{lemma}
\begin{proof}
Let $ F = F ( b^{ i j } ) $ be a differentiable function with respect to $ b^{ i j } $ for $ i , j = 1 , 2 , 3 $. Define 
\begin{align}
{ b_\alpha }^{ i j } = ( \alpha - 1 ) { b_2 }^{ i j } + ( 2 - \alpha ) { b_1 }^{ i j } , \qquad 1 \leq \alpha \leq 2 . 
\end{align}
Then, we have 
\begin{align}
F ( { b_2 }^{ i j } ) - F ( { b_1 }^{ i j } ) & = \int_1^2 \frac{d ( F ( { b_\alpha }^{ i j } )) }{ d \alpha } \, d \alpha \nonumber \\ 
& = \int_1^2 \frac{ \partial F }{ \partial b^{ k l } } ( { b_\alpha }^{ i j } ) \frac{ d { b_\alpha }^{ k l } }{ d \alpha } \, d \alpha \nonumber \\ 
& = \int_1^2 \frac{ \partial F }{ \partial b^{ k l } } ( { b_\alpha }^{ i j } ) ( { b_2 }^{ k l } - { b_1 }^{ k l } ) \, d \alpha \nonumber \\ 
& = ( { b_2 }^{ k l } - { b_1 }^{ k l } ) \int_1^2 \frac{ \partial F }{ \partial b^{ k l } } ( { b_\alpha }^{ i j } ) \, d \alpha . 
\end{align}
Hence, we obtain 
\begin{align}
| F ( { b_2 }^{ i j } ) - F ( { b_1 }^{ i j } ) | \leq \| b_2 - b_1 \| \max_{ k , l } \sup_{ 1 \leq \alpha \leq 2 } \left| \frac{ \partial F }{ \partial b^{ k l } } ( { b_\alpha }^{ i j } ) \right| . \label{F-F}
\end{align}
The lemma is now proved by using \eqref{F-F} with the estimates of the following quantities: 
\begin{align*}
\frac{ \partial ( \det b )^{ \frac12 } }{ \partial b^{ k l } } , \qquad \frac{ \partial }{ \partial b^{ k l } } \left( \frac{ 1 }{ p^0 q^0 } \right) , \qquad \frac{ \partial ( h^{ 2 - \gamma } ) }{ \partial b^{ k l } } . 
\end{align*}
We first note that $ \partial ( \det b ) / \partial b^{ k l } $ is a second order polynomial of $ b^{ i j } $ so that 
\begin{align}
\left| \frac{ \partial ( \det b )^{ \frac12 } }{ \partial b^{ k l } } \right| \leq C \frac{ \| b \|^2 }{ ( \det b )^{ \frac12 } } . 
\end{align}
For the other quantities we use \eqref{Wconti1}, \eqref{Wconti2} and \eqref{Wconti3}, and apply Lemma \ref{lem e} to obtain 
\begin{align}
& \left| \frac{ \partial }{ \partial b^{ k l } } \left( \frac{ 1 }{ p^0 q^0 } \right) \right| \leq C \| e^{ - 1 } \|^2 \frac{ 1 }{ p^0 q^0 } \leq C \frac{ \| b \|^5 }{ ( \det b )^2 } \frac{ 1 }{ p^0 q^0 } , \label{F-F2} \\ 
& \left| \frac{ \partial ( h^{ 2 - \gamma } ) }{ \partial b^{ k l } } \right| \leq C \| e^{ - 1 } \|^2 h^{ 2 - \gamma } \leq C \frac{ \| b \|^5 }{ ( \det b )^2 } h^{ 2 - \gamma } \leq C \frac{ \| b \|^5 }{ ( \det b )^2 } ( p^0 q^0 )^{ 1 - \frac{ \gamma }{ 2 } } . \label{F-F3}
\end{align}
Now, we note that 
\begin{align}
\sup_{ 1 \leq \alpha \leq 2 } \| b_\alpha \| \leq C_A , 
\end{align}
since $ | { b_\alpha }^{ i j } | \leq 2 A $. Moreover, by Lemma \ref{lem HJ} we have for any $ 1 \leq \alpha \leq 2 $, 
\begin{align}
\det b_\alpha \geq ( \det b_2 )^{ \alpha - 1 } ( \det b_1 )^{ 2 - \alpha } \geq \frac{ 1 }{ A } . 
\end{align}
Hence, we obtain 
\begin{align}
\sup_{ 1 \leq \alpha \leq 2 } \frac{ \| b_\alpha \|^2 }{ ( \det b_\alpha )^{ \frac12 } } \leq C_A , 
\end{align}
which implies that by \eqref{F-F} with $ F = ( \det b )^{ 1 / 2 } $, 
\begin{align}
\left| ( \det b_1 )^{ \frac12 } - ( \det b_2 )^{ \frac12 } \right| & \leq C_A \| b_1 - b_2 \| . 
\end{align}
For the second and third estimates of the lemma, we use \eqref{p^0 1} and \eqref{p^0 2} to obtain 
\begin{align}
\sup_{ 1 \leq \alpha \leq 2 } { p_\alpha }^0 \leq C_A \langle p \rangle , \qquad \sup_{ 1 \leq \alpha \leq 2 } { q_\alpha }^0 \leq C_A \langle q \rangle , \qquad \sup_{ 1 \leq \alpha \leq 2 } \frac{ 1 }{ { p_\alpha }^0 } \leq \frac{ C_A }{ \langle p \rangle } , \qquad \sup_{ 1 \leq \alpha \leq 2 } \frac{ 1 }{ { q_\alpha }^0 } \leq \frac{ C_A }{ \langle q \rangle } , 
\end{align}
which imply that 
\begin{align}
\sup_{ 1 \leq \alpha \leq 2 } \frac{ \| b_\alpha \|^5 }{ ( \det b_\alpha )^2 } \frac{ 1 }{ { p_\alpha }^0 { q_\alpha }^0 } \leq \frac{ C_A }{ \langle p \rangle \langle q \rangle } , \qquad \sup_{ 1 \leq \alpha \leq 2 } \frac{ \| b_\alpha \|^5 }{ ( \det b_\alpha )^2 } ( { p_\alpha }^0 { q_\alpha }^0 )^{ 1 - \frac{ \gamma }{ 2 } } \leq C_A \langle p \rangle^{ 1 - \frac{ \gamma }{ 2 } } \langle q \rangle^{ 1 - \frac{ \gamma }{ 2 } } . 
\end{align}
Applying these estimates to \eqref{F-F}--\eqref{F-F3} we obtain the desired results. 
\end{proof}

\begin{lemma} \label{lem dp'db}
Suppose that there exist a time interval $ [ 0 , T ] $ and a positive constant $ A $ such that $ { b }^{ i j } $ and $ { K }_{ i j } $ are defined on $ [ 0 , T ] $ and satisfy
\begin{align}
\sup_{ 0 \leq s \leq T } \| b ( s ) \| \leq A , \qquad \sup_{ 0 \leq s \leq T } \| K ( s ) \| \leq A , \qquad \inf_{ 0 \leq s \leq T } \det b ( s ) \geq \frac{ 1 }{ A } . 
\end{align}
Then, we have the following estimate on $ [ 0 , T ] $: 
\begin{align}
\left| \frac{ \partial p' }{ \partial b^{ k l } } \right| + \left| \frac{ \partial q' }{ \partial b^{ k l } } \right| \leq C_A n^0 , 
\end{align}
where $ C_A $ is a positive constant depending on $ A $. 
\end{lemma}
\begin{proof}
Recall that $ \partial p' / \partial b^{ k l } $ is computed in \eqref{W11}: 
\begin{align*}
\frac{\partial p'_i}{\partial b^{st}} & = \frac{1}{2}\bigg\{ \frac{\partial h}{\partial b^{st}} \omega_j { e^j }_i + h\omega_j (\partial_{b^{st}} { e^j }_i ) + \frac{n_l \omega_k \delta^{ j k } (\partial_{b^{st}} { e_j }^l ) n_i }{ n^0 + h} - \frac{n_l \omega_k \delta^{ j k } { e_j }^l n_i }{ (n^0 + h)^2}\left(\frac{\partial n^0}{\partial b^{st}} + \frac{\partial h}{\partial b^{st}}\right)\bigg\} .
\end{align*}
Here, the quantities $ \omega_j $, $ { e^j }_i $, $ \partial_{ b^{ s t } } { e^j }_i $, $ \partial_{ b^{ s t } } { e_j }^l $ and $ { e_j }^l $ are all bounded by $ C_A $ by the same arguments as in the proof of Lemma \ref{lem F-F}. We use \eqref{Wconti4} to obtain 
\begin{align}
\left| \frac{ \partial h }{ \partial b^{ s t } } \right| \leq C \| e^{ - 1 } \|^2 h \leq C_A n^0 . 
\end{align}
We also easily obtain $ h \leq C n^0 $, $ | n_j | \leq C_A n^0 $ and $ | \partial n^0 / \partial b^{ s t } | \leq C_A n^0 $. We combine the above estimates to obtain the desired result for $ \partial p' / \partial b^{ k l } $. The result for $ \partial q' / \partial b^{ k l } $ is obtained by the same arguments, and this completes the proof. 
\end{proof}

In Proposition \ref{prop C} below, we will show that the solution $ f $ of the Boltzmann equation depends continuously on the metric $ b^{ i j } $. To see this we will suppose that two metrics $ { b_1 }^{ i j } $, $ { b_2 }^{ i j } $ and an initial data $ f_0 $ are given, so that we have two solutions $ f_1 $ and $ f_2 $, corresponding to $ { b_1 }^{ i j } $ and $ { b_2 }^{ i j } $ respectively, having the initial data $ f_0 $. Then, $ f_1 $ is bounded in $ L^1_{ 1 , r } $, and $ f_2 $ is bounded in $ L^1_{ 2 , r } $, for some $ r $, provided that $ f_0 $ is bounded in $ L^1_{ 1 , r } \cap L^1_{ 2 , r } $. However, we will assume that the initial values of $ { b_1 }^{ i j } $ and $ { b_2 }^{ i j } $ are the same, i.e., 
\begin{align}
{ b_1 }^{ i j } ( 0 ) = { b_0 }^{ i j } = { b_2 }^{ i j } ( 0 ) , 
\end{align}
for some suitable $ { b_0 }^{ i j } $, and this implies that $ L^1_{ 1 , r } = L^1_{ 2 , r } $ at $ s = 0 $. Hence, we may denote it by $ L^1_r $ and will assume that the initial data $ f_0 $ is bounded in $ L^1_r $.

\begin{prop}\label{prop C}
Suppose that there exist a time interval $ [ 0 , T ] $ and a positive constant $ A $ such that $ { b_1 }^{ i j } $, $ { K_1 }_{ i j } $, $ { b_2 }^{ i j } $ and $ { K_2 }_{ i j } $ are defined on $ [ 0 , T ] $ and satisfy $ { b_1 }^{ i j } ( 0 ) = { b_2 }^{ i j } ( 0 ) $, $ { K_1 }_{ i j } ( 0 ) = { K_2 }_{ i j } ( 0 ) $ and 
\begin{align}
\sup_{ 0 \leq s \leq T } \| b_1 ( s ) \| \leq A , \qquad \sup_{ 0 \leq s \leq T } \| K_1 ( s ) \| \leq A , \qquad \inf_{ 0 \leq s \leq T } \det b_1 ( s ) \geq \frac{ 1 }{ A } , \\ 
\sup_{ 0 \leq s \leq T } \| b_2 ( s ) \| \leq A , \qquad \sup_{ 0 \leq s \leq T } \| K_2 ( s ) \| \leq A , \qquad \inf_{ 0 \leq s \leq T } \det b_2 ( s ) \geq \frac{ 1 }{ A } . 
\end{align}
Then, for any initial data $ f_0 \geq 0 $ such that 
\begin{align}
f_0 \in L^1_1 ( \bbr^3 ) \cap L^1_{ - 2 - \delta / 2 } ( \bbr^3 ) \cap L^\infty_\eta ( \bbr^3 ) , \qquad \frac{ \partial f_0 }{ \partial p } \in L^1_1 ( \bbr^3 ) \cap L^1_{ - 1 - \delta / 2 } ( \bbr^3 ) , 
\end{align}
where $ 0 < \eta < 2 / A^2 $, and $ \delta > 0 $ is a number satisfying
\begin{align}
\gamma + \delta < 2 , 
\end{align}
there exists $ 0 \leq T_A \leq T $ such that the Boltzmann equation \eqref{B} has unique non-negative solutions $ f_1 $ and $ f_2 $, corresponding to $ { b_1 }^{ i j } $ and $ { b_2 }^{ i j } $ respectively, such that
\begin{align}
\| f_1 ( s ) - f_2 ( s ) \|_{ L^1_{ 1 , 1 } } + \| f_1 ( s ) - f_2 ( s ) \|_{ L^1_{ 1 , - 1 } } & \leq C_A \sup_{ 0 \leq \rho \leq s } \| b_1 ( \rho ) - b_2 ( \rho ) \| , \label{prop B5} \\
\left\| \frac{ \partial f_1 }{ \partial s } ( s ) - \frac{ \partial f_2 }{ \partial s } ( s ) \right\|_{ L^1_{ 1 , 1 } } + \left\| \frac{ \partial f_1 }{ \partial s } ( s ) - \frac{ \partial f_2 }{ \partial s } ( s ) \right\|_{ L^1_{ 1 , - 1 } } & \leq C_A \sup_{ 0 \leq \rho \leq s } \| b_1 ( \rho ) - b_2 ( \rho ) \| , \label{prop B6}
\end{align}
for $ 0 \leq s \leq T_A $, where $ C_A $ is a positive constant depending on $ A $, $ \eta $ and $ \delta $. 
\end{prop}
\begin{proof}
Let us denote by $ f_1 $ and $ f_2 $ the solutions of the Boltzmann equation associated with $ { b_1 }^{ i j } $ and $ { b_2 }^{ i j } $, respectively. Note that the existence of $ f_1 $ and $ f_2 $ is given by Lemma \ref{lem dfdp}. Since $ f_1 $ and $ f_2 $ satisfy 
\begin{align}
\frac{ \partial f_1 }{ \partial s } & = ( \det b_1 )^{ \frac12 } \iint \frac{ { h_1 }^{ 2 - \gamma } }{ { p_1 }^0 { q_1 }^0 } ( f_1 ( p_1' ) f_1 ( q_1' ) - f_1 ( p ) f_1 ( q ) ) \, d \omega \, d^3 q , \label{df ds1} \\ 
\frac{ \partial f_2 }{ \partial s } & = ( \det b_2 )^{ \frac12 } \iint \frac{ { h_2 }^{ 2 - \gamma } }{ { p_2 }^0 { q_2 }^0 } ( f_2 ( p_2' ) f_2 ( q_2' ) - f_2 ( p ) f_2 ( q ) ) \, d \omega \, d^3 q , \label{df ds2}
\end{align}
respectively, we have 
\begin{align}
\frac{ \partial ( f_1 - f_2 ) }{ \partial s } & = ( ( \det b_1 )^{ \frac12 } - ( \det b_2 )^{ \frac12 } ) \iint \frac{ { h_1 }^{ 2 - \gamma } }{ { p_1 }^0 { q_1 }^0 } ( f_1 ( p_1' ) f_1 ( q_1' ) - f_1 ( p ) f_1 ( q ) ) \, d \omega \, d^3 q \nonumber \\ 
& \quad + ( \det b_2 )^{ \frac12 } \iint \left( \frac{ 1 }{ { p_1 }^0 { q_1 }^0 } - \frac{ 1 }{ { p_2 }^0 { q_2 }^0 } \right) { h_1 }^{ 2 - \gamma } ( f_1 ( p_1' ) f_1 ( q_1' ) - f_1 ( p ) f_1 ( q ) ) \, d \omega \, d^3 q \nonumber \\ 
& \quad + ( \det b_2 )^{ \frac12 } \iint \frac{ 1 }{ { p_2 }^0 { q_2 }^0 } ( { h_1 }^{ 2 - \gamma } - { h_2 }^{ 2 - \gamma } ) ( f_1 ( p_1' ) f_1 ( q_1' ) - f_1 ( p ) f_1 ( q ) ) \, d \omega \, d^3 q \nonumber \\ 
& \quad + ( \det b_2 )^{ \frac12 } \iint \frac{ { h_2 }^{ 2 - \gamma } }{ { p_2 }^0 { q_2 }^0 } ( f_1 ( p_1' ) f_1 ( q_1' ) - f_1 ( p_2' ) f_1 ( q_2' ) ) \, d \omega \, d^3 q \nonumber \\ 
& \quad + \frac{ ( \det b_2 )^{ \frac12 } }{ 2 } \iint \frac{ { h_2 }^{ 2 - \gamma } }{ { p_2 }^0 { q_2 }^0 } ( ( f_1 + f_2 ) ( p_2' ) ( f_1 - f_2 ) ( q_2' ) + ( f_1 - f_2 ) ( p_2' ) ( f_1 + f_2 ) ( q_2' ) \nonumber \\ 
& \qquad \qquad \qquad \qquad - ( f_1 + f_2 ) ( p ) ( f_1 - f_2 ) ( q ) - ( f_1 - f_2 ) ( p ) ( f_1 + f_2 ) ( q ) ) \, d \omega \, d^3 q . \label{df ds3}
\end{align}
Multiplying the above by $ ( f_1 - f_2 ) ( p ) / | f_1 - f_2 | ( p ) $, we obtain 
\begin{align}
\frac{ \partial | f_1 - f_2 | }{ \partial s } & \leq | ( \det b_1 )^{ \frac12 } - ( \det b_2 )^{ \frac12 } | \iint \frac{ { h_1 }^{ 2 - \gamma } }{ { p_1 }^0 { q_1 }^0 } ( f_1 ( p_1' ) f_1 ( q_1' ) + f_1 ( p ) f_1 ( q ) ) \, d \omega \, d^3 q \nonumber \\ 
& \quad + ( \det b_2 )^{ \frac12 } \iint \left| \frac{ 1 }{ { p_1 }^0 { q_1 }^0 } - \frac{ 1 }{ { p_2 }^0 { q_2 }^0 } \right| { h_1 }^{ 2 - \gamma } ( f_1 ( p_1' ) f_1 ( q_1' ) + f_1 ( p ) f_1 ( q ) ) \, d \omega \, d^3 q \nonumber \\ 
& \quad + ( \det b_2 )^{ \frac12 } \iint \frac{ 1 }{ { p_2 }^0 { q_2 }^0 } | { h_1 }^{ 2 - \gamma } - { h_2 }^{ 2 - \gamma } | ( f_1 ( p_1' ) f_1 ( q_1' ) + f_1 ( p ) f_1 ( q ) ) \, d \omega \, d^3 q \nonumber \\ 
& \quad + ( \det b_2 )^{ \frac12 } \iint \frac{ { h_2 }^{ 2 - \gamma } }{ { p_2 }^0 { q_2 }^0 } | f_1 ( p_1' ) f_1 ( q_1' ) - f_1 ( p_2' ) f_1 ( q_2' ) | \, d \omega \, d^3 q \nonumber \\ 
& \quad + \frac{ ( \det b_2 )^{ \frac12 } }{ 2 } \iint \frac{ { h_2 }^{ 2 - \gamma } }{ { p_2 }^0 { q_2 }^0 } ( ( f_1 + f_2 ) ( p_2' ) | f_1 - f_2 | ( q_2' ) + | f_1 - f_2 | ( p_2' ) ( f_1 + f_2 ) ( q_2' ) \nonumber \\ 
& \qquad \qquad \qquad \qquad + ( f_1 + f_2 ) ( p ) | f_1 - f_2 | ( q ) - | f_1 - f_2 | ( p ) ( f_1 + f_2 ) ( q ) ) \, d \omega \, d^3 q . \label{df ds4}
\end{align}
Finally, we multiply the above by $ ( { p_1 }^0 )^r $ and integrate over $ \bbr^3_p $ to obtain 
\begin{align}
& \frac{ d }{ d s } \| f_1 - f_2 \|_{ L^1_{ 1 , r } }  - \int_{ \bbr^3 } | f_1 - f_2 | ( p ) \left( r ( { p_1 }^0 )^{ r - 1 } \frac{ \partial { p_1 }^0 }{ \partial s } \right) \, d^3 p \nonumber \\ 
& \leq | ( \det b_1 )^{ \frac12 } - ( \det b_2 )^{ \frac12 } | \iiint \frac{ { h_1 }^{ 2 - \gamma } }{ { p_1 }^0 { q_1 }^0 } ( f_1 ( p_1' ) f_1 ( q_1' ) + f_1 ( p ) f_1 ( q ) ) ( { p_1 }^0 )^r \, d \omega \, d^3 q \, d^3 p \nonumber \\ 
& \quad + ( \det b_2 )^{ \frac12 } \iiint \left| \frac{ 1 }{ { p_1 }^0 { q_1 }^0 } - \frac{ 1 }{ { p_2 }^0 { q_2 }^0 } \right| { h_1 }^{ 2 - \gamma } ( f_1 ( p_1' ) f_1 ( q_1' ) + f_1 ( p ) f_1 ( q ) ) ( { p_1 }^0 )^r \, d \omega \, d^3 q \, d^3 p \nonumber \\ 
& \quad + ( \det b_2 )^{ \frac12 } \iiint \frac{ 1 }{ { p_2 }^0 { q_2 }^0 } | { h_1 }^{ 2 - \gamma } - { h_2 }^{ 2 - \gamma } | ( f_1 ( p_1' ) f_1 ( q_1' ) + f_1 ( p ) f_1 ( q ) ) ( { p_1 }^0 )^r \, d \omega \, d^3 q \, d^3 p \nonumber \\ 
& \quad + ( \det b_2 )^{ \frac12 } \iiint \frac{ { h_2 }^{ 2 - \gamma } }{ { p_2 }^0 { q_2 }^0 } | f_1 ( p_1' ) f_1 ( q_1' ) - f_1 ( p_2' ) f_1 ( q_2' ) | ( { p_1 }^0 )^r \, d \omega \, d^3 q \, d^3 p \nonumber \\ 
& \quad + \frac{ ( \det b_2 )^{ \frac12 } }{ 2 } \iiint \frac{ { h_2 }^{ 2 - \gamma } }{ { p_2 }^0 { q_2 }^0 } ( ( f_1 + f_2 ) ( p_2' ) | f_1 - f_2 | ( q_2' ) + | f_1 - f_2 | ( p_2' ) ( f_1 + f_2 ) ( q_2' ) \nonumber \\ 
& \qquad \qquad \qquad \qquad + ( f_1 + f_2 ) ( p ) | f_1 - f_2 | ( q ) ) ( { p_1 }^0 )^r \, d \omega \, d^3 q \, d^3 p \nonumber \\
& \quad - \frac{ ( \det b_2 )^{ \frac12 } }{ 2 } \iiint \frac{ { h_2 }^{ 2 - \gamma } }{ { p_2 }^0 { q_2 }^0 } | f_1 - f_2 | ( p ) ( f_1 + f_2 ) ( q ) ( { p_1 }^0 )^r \, d \omega \, d^3 q \, d^3 p . \label{f-f}
\end{align}
We can write the above as 
\begin{align}
\frac{ d }{ d s } \| f_1 - f_2 \|_{ L^1_{ 1 , r } } \leq K_{ 0 , r } + K_{ 1 , r } + K_{ 2 , r } + K_{ 3 , r } + K_{ 4 , r } + K_{ 5 , r } - K_{ 6 , r } , \label{f-fs}
\end{align}
where $ K_{ 0 , r } $ corresponds to the second term on the left hand side of \eqref{f-f}, and the other quantities are the six integrals on the right hand side of \eqref{f-f}. Then, for $ r = 1 $ and $ r = - 1 $ we have 
\begin{align}
\frac{ d }{ d s } \left( \| f_1 - f_2 \|_{ L^1_{ 1 , 1 } } + \| f_1 - f_2 \|_{ L^1_{ 1 , - 1 } } \right) + K_6 \leq K_{ 0 } + K_{ 1 } + K_{ 2 } + K_{ 3 } + K_{ 4 } + K_{ 5 } , \label{f-fa} 
\end{align}
where $ K_{ n } = K_{ n , 1 } + K_{ n , - 1 } $ $ ( n = 0 , \dots , 6 ) $. We claim that each $ K_n $ for $ n = 0 , \dots , 5 $ satisfies the following estimates:
\begin{align}
| K_0 | & \leq C_A \left( \| f_1 - f_2 \|_{ L^1_{ 1 , 1 } } + \| f_1 - f_2 \|_{ L^1_{ 1 , - 1 } } \right) , \label{f-f0} \\
K_1 & \leq C_A \| b_1 - b_2 \| , \label{f-f1} \\
K_2 & \leq C_A \| b_1 - b_2 \| , \label{f-f2} \\
K_3 & \leq C_A \| b_1 - b_2 \| , \label{f-f3} \\
K_4 & \leq C_A \| b_1 - b_2 \| , \label{f-f4} \\
K_5 & \leq C_A \left( \| f_1 - f_2 \|_{ L^1_{ 1 , 1 } } + \| f_1 - f_2 \|_{ L^1_{ 1 , - 1 } } \right) . \label{f-f5} 
\end{align}
If the estimates \eqref{f-f0}--\eqref{f-f5} are satisfied, then we have from \eqref{f-fa}
\begin{align}
\frac{ d }{ d s } \left( \| f_1 - f_2 \|_{ L^1_{ 1 , 1 } } + \| f_1 - f_2 \|_{ L^1_{ 1 , - 1 } } \right) \leq C_A \left( \| b_1 - b_2 \| + \| f_1 - f_2 \|_{ L^1_{ 1 , 1 } } + \| f_1 - f_2 \|_{ L^1_{ 1 , - 1 } } \right) ,
\end{align}
since $ K_6 $ is non-negative, and use Gr{\"o}nwall's inequality to obtain the desired result \eqref{prop B5}. We use again \eqref{f-fa} to obtain the estimate for $ K_6 $:
\begin{align}
K_6 \leq C_A \sup_{ 0 \leq \rho \leq s } \| b_1 ( \rho ) - b_2 ( \rho ) \| , \label{f-f6}
\end{align}
which will be used to prove \eqref{prop B6}. Below, we will estimate each $ K_n $ for $ n = 0 , \dots , 5 $ separately. \medskip

\noindent (0) {\it Estimate of $ K_0 $.} Recall that $ K_{ 0 , r } $ is given by 
\begin{align*}
K_{ 0 , r } = \int_{ \bbr^3 } | f_1 - f_2 | ( p ) \left( r ( { p_1 }^0 )^{ r - 1 } \frac{ \partial { p_1 }^0 }{ \partial s } \right) \, d p . 
\end{align*}

\begin{itemize}

\item $ ( r = 1 ) $ For $ r = 1 $, we use \eqref{dp ds} to have
\begin{align}
| K_{ 0 , 1 } | \leq C_A \int | f_1 - f_2 | ( p ) { p_1 }^0 \, d^3 p \leq C_A \| f_1 - f_2 \|_{ L^1_{ 1 , 1 } } .
\end{align}

\item $ ( r = - 1 ) $ Similarly, we have 
\begin{align}
| K_{ 0 , - 1 } | \leq C_A \int | f_1 - f_2 | ( p ) \frac{ 1 }{ { p_1 }^0 } \, d p \leq C_A \| f_1 - f_2 \|_{ L^1_{ 1 , - 1 } } . 
\end{align}

\end{itemize} 
Since $ K_0 = K_{ 0 , 1 } + K_{ 0 , - 1 } $, we obtain \eqref{f-f0}. \medskip

\noindent (1) {\it Estimate of $ K_1 $.} We first use Lemma \ref{lem F-F} to estimate $ K_{ 1 , r } $ as 
\begin{align}
K_{ 1 , r } & = | ( \det b_1 )^{ \frac12 } - ( \det b_2 )^{ \frac12 } | \iiint \frac{ { h_1 }^{ 2 - \gamma } }{ { p_1 }^0 { q_1 }^0 } ( f_1 ( p_1' ) f_1 ( q_1' ) + f_1 ( p ) f_1 ( q ) ) ( { p_1 }^0 )^r \, d \omega \, d q \, d p \nonumber \\ 
& \leq C_A \| b_1 - b_2 \| \iiint \frac{ { h_1 }^{ 2 - \gamma } }{ { p_1 }^0 { q_1 }^0 } f_1 ( p ) f_1 ( q ) ( ( { p_1' }^0 )^r + ( { p_1 }^0 )^r ) \, d \omega \, d^3 q \, d^3 p , 
\end{align}
where we used 
\begin{align}
\det b_1 \frac{ d \omega \, d^3 p \, d^3 q }{ { p_1 }^0 { q_1 }^0 } = \det b_1 \frac{ d \omega \, d^3 p_1' \, d^3 q_1' }{ { p_1' }^0 { q_1' }^0 }, \label{dpdq1}
\end{align}
as in \eqref{dpdq}. 

\begin{itemize}

\item $ ( r = 1 ) $ Applying the energy conservation \eqref{conservation}, we obtain
\begin{align}
K_{ 1 , 1 } & \leq C_A \| b_1 - b_2 \| \iint \frac{ { h_1 }^{ 2 - \gamma } }{ { p_1 }^0 { q_1 }^0 } f_1 ( p ) f_1 ( q ) ( { p_1 }^0 + { q_1 }^0 ) \, d^3 q \, d^3 p \nonumber \\ 
& \leq C_A \| b_1 - b_2 \| \| f_1 \|_{ L^1_{ 1 , 1 - \gamma / 2 } } \| f_1 \|_{ L^1_{ 1 , - \gamma / 2 } } . 
\end{align}

\item $ ( r = - 1 ) $ We apply the same estimates as in \eqref{L111} and \eqref{L113} to obtain 
\begin{align}
K_{ 1 , - 1 } & \leq C_A \| b_1 - b_2 \| \iiint \frac{ { h_1 }^{ 2 - \gamma } }{ { p_1 }^0 { q_1 }^0 } f_1 ( p ) f_1 ( q ) \frac{ 1 }{ { p_1' }^0 } \, d \omega \, d^3 q \, d^3 p \nonumber \\ 
& \quad + C_A \| b_1 - b_2 \| \iint \frac{ { h_1 }^{ 2 - \gamma } }{ { p_1 }^0 { q_1 }^0 } f_1 ( p ) f_1 ( q ) \frac{ 1 }{ { p_1 }^0 } \, d^3 q \, d^3 p \nonumber \\ 
& \leq C_A \| b_1 - b_2 \| \left( \| f_1 \|_{ L^1_{ 1 , - \gamma } } \| f_1 \|_{ L^1_{ 1 , - 1 } } + \| f_1 \|_{ L^1_{ 1 , - 1 - \gamma / 2 } } \| f_1 \|_{ L^1_{ 1 , - \gamma / 2 } } \right) . 
\end{align} 
\end{itemize} 
By Lemma \ref{lem dfdp}, the quantities $ \| f_1 \|_{ L^1_{ 1 , 1 - \gamma / 2 } } $, $ \| f_1 \|_{ L^1_{ 1 , - \gamma / 2 } } $, $ \| f_1 \|_{ L^1_{ 1 , - \gamma } } $, $ \| f_1 \|_{ L^1_{ 1 , - 1 } } $ and $ \| f_1 \|_{ L^1_{ 1 , - 1 - \gamma / 2 } } $ are bounded for $ 1 < \gamma < 2 $ . Hence, we obtain \eqref{f-f1}. \medskip

\noindent (2) {\it Estimate of $ K_2 $.} We use Lemma \ref{lem F-F} to estimate $ K_{ 2 , r } $ as 
\begin{align}
K_{ 2 , r } & \leq C_A \| b_1 - b_2 \| \iiint \frac{ { h_1 }^{ 2 - \gamma } }{ \langle p \rangle \langle q \rangle } ( f_1 ( p_1' ) f_1 ( q_1' ) + f_1 ( p ) f_1 ( q ) ) ( { p_1 }^0 )^r \, d \omega \, d^3 q \, d^3 p \nonumber \\ 
& \leq C_A \| b_1 - b_2 \| \iiint \frac{ { h_1 }^{ 2 - \gamma } }{ { p_1 }^0 { q_1 }^0 } f_1 ( p ) f_1 ( q ) ( ( { p_1' }^0 )^r + ( { p_1 }^0 )^r ) \, d \omega \, d^3 q \, d^3 p , 
\end{align}
where we used \eqref{p^0 1}. Then, we obtain \eqref{f-f2} by the same arguments as in the estimate of $ K_1 $. \medskip

\noindent (3) {\it Estimate of $ K_3 $.} We use Lemma \ref{lem F-F} to estimate $ K_{ 3 , r } $ as 
\begin{align}
K_{ 3 , r } & \leq C_A \| b_1 - b_2 \| \iiint \frac{ \langle p \rangle^{ 1 - \frac{ \gamma }{ 2 } } \langle q \rangle^{ 1 - \frac{ \gamma }{ 2 } } }{ { p_2 }^0 { q_2 }^0 } ( f_1 ( p_1' ) f_1 ( q_1' ) + f_1 ( p ) f_1 ( q ) ) ( { p_1 }^0 )^r \, d \omega \, d^3 q \, d^3 p \nonumber \\ 
& \leq C_A \| b_1 - b_2 \| \iiint \frac{ ( { p_1 }^0 )^{ 1 - \frac{ \gamma }{ 2 } } ( { q_1 }^0 )^{ 1 - \frac{ \gamma }{ 2 } } }{ { p_1 }^0 { q_1 }^0 } ( f_1 ( p_1' ) f_1 ( q_1' ) + f_1 ( p ) f_1 ( q ) ) ( { p_1 }^0 )^r \, d \omega \, d^3 q \, d^3 p , 
\end{align}
where we used \eqref{p^0 1}--\eqref{p^0 2}, so that we can apply \eqref{dpdq1}. Then, we have 
\begin{align}
K_{ 3 , r } & \leq C_A \| b_1 - b_2 \| \iiint \frac{ ( { n_1 }^0 )^{ 2 - \gamma } }{ { p_1 }^0 { q_1 }^0 } f_1 ( p_1' ) f_1 ( q_1' ) ( { p_1 }^0 )^r \, d \omega \, d^3 q \, d^3 p \nonumber \\ 
& \quad + C_A \| b_1 - b_2 \| \iiint \frac{ 1 }{ ( { p_1 }^0 )^{ \frac{ \gamma }{ 2 } } ( { q_1 }^0 )^{ \frac{ \gamma }{ 2 } } } f_1 ( p ) f_1 ( q ) ( { p_1 }^0 )^r \, d \omega \, d^3 q \, d^3 p \nonumber \\ 
& = : K_{ 3 1 , r } + K_{ 3 2 , r } , 
\end{align}
where we simply used $ { p_1 }^0 , { q_1 }^0 \leq { n_1 }^0 $ for the gain term.

\begin{itemize}
\item $ ( r = 1 ) $ For the gain term $ K_{ 3 1 , 1 } $ we apply the change of variables \eqref{dpdq1} and the energy conservation \eqref{conservation} to obtain 
\begin{align}
K_{ 3 1 , 1 } & \leq C_A \| b_1 - b_2 \| \iiint \frac{ ( { n_1 }^0 )^{ 3 - \gamma } }{ { p_1 }^0 { q_1 }^0 } f_1 ( p ) f_1 ( q ) \, d \omega \, d^3 q \, d^3 p \nonumber \\ 
& \leq C_A \| b_1 - b_2 \| \iint \left( \frac{ ( { p_1 }^0 )^{ 2 - \gamma } }{ { q_1 }^0 } + \frac{ ( { q_1 }^0 )^{ 2 - \gamma } }{ { p_1 }^0 } \right) f_1 ( p ) f_1 ( q ) \, d^3 q \, d^3 p \nonumber \\ 
& \leq C_A \| b_1 - b_2 \| \| f_1 \|_{ L^1_{ 1 , 2 - \gamma } } \| f_1 \|_{ L^1_{ 1 , - 1 } } .
\end{align}
For the loss term $ K_{ 3 2 , 1 } $ we easily obtain 
\begin{align}
K_{ 3 2 , 1 } \leq C_A \| b_1 - b_2 \| \| f_1 \|_{ L^1_{ 1 , 1 - \gamma / 2 } } \| f_1 \|_{ L^1_{ 1 , - \gamma / 2 } } . 
\end{align} 

\item $ ( r = - 1 ) $ We apply \eqref{dpdq1} again to obtain 
\begin{align}
K_{ 3 1 , - 1 } & \leq C_A \| b_1 - b_2 \| \iiint \frac{ ( { n_1 }^0 )^{ 2 - \gamma } }{ { p_1 }^0 { q_1 }^0 } f_1 ( p_1' ) f_1 ( q_1' ) \frac{ 1 }{ { p_1 }^0 } \, d \omega \, d^3 q \, d^3 p \nonumber \\ 
& \leq C_A \| b_1 - b_2 \| \iiint \frac{ ( { n_1 }^0 )^{ 2 - \gamma } }{ { p_1 }^0 { q_1 }^0 } f_1 ( p ) f_1 ( q ) \frac{ 1 }{ { p_1' }^0 } \, d \omega \, d^3 q \, d^3 p . 
\end{align}
Applying \eqref{l2} of Lemma \ref{lem p'}, we further estimate as
\begin{align}
K_{ 3 1 , - 1 } & \leq C_A \| b_1 - b_2 \| \iint \frac{ 1 }{ { p_1 }^0 { q_1 }^0 } f_1 ( p ) f_1 ( q ) \frac{ 1 }{ { h_1 }^{ \gamma - 1 } } \, d^3 q \, d^3 p \nonumber \\ 
& \leq C_A \| b_1 - b_2 \| \| f_1 \|^2_{ L^\infty_{ 1 , \eta } } \iint \frac{ e^{ - s_\eta^{ - 1 } { p_1 }^0 } e^{ - s_\eta^{ - 1 } { q_1 }^0 } }{ ( { p_1 }^0 )^2 ( { q_1 }^0 )^2 { h_1 }^{ \gamma - 1 } } \, d^3 q \, d^3 p , 
\end{align}
since $ f_1 $ is bounded in $ L^\infty_{ 1 , \eta } $. Here, the integral over $ \bbr^3_p \times \bbr^3_q $ is finite: 
\begin{align}
\iint \frac{ e^{ - s_\eta^{ - 1 } { p_1 }^0 } e^{ - s_\eta^{ - 1 } { q_1 }^0 } }{ ( { p_1 }^0 )^2 ( { q_1 }^0 )^2 { h_1 }^{ \gamma - 1 } } \, d^3 q \, d^3 p & \leq C_A \iint \frac{ e^{ - s_\eta^{ - 1 } { p_1 }^0 } e^{ - s_\eta^{ - 1 } { q_1 }^0 } }{ ( { p_1 }^0 )^{ 2 + \frac{ \gamma - 1 }{ 2 } } ( { q_1 }^0 )^{ 2 + \frac{ \gamma - 1 }{ 2 } } \sin^{ \gamma - 1 } ( \phi / 2 ) } \, d^3 q \, d^3 p \nonumber \\ 
& \leq C_A \int \frac{ e^{ - s_\eta^{ - 1 } | \hat{ p } | } }{ | \hat{ p } |^{ 2 + \frac{ \gamma - 1 }{ 2 } } } \int_0^\infty \int_0^\pi \frac{ e^{ - s_\eta^{ - 1 } | \hat{ q } | } \sin \phi }{ | \hat{ q } |^{ \frac{ \gamma - 1 }{ 2 } } \sin^{ \gamma - 1 } ( \phi / 2 ) } \, d \phi \, d | \hat{ q } | \, d^3 \hat{ p } \nonumber \\ 
& \leq C_A \int \frac{ e^{ - s_\eta^{ - 1 } | \hat{ p } | } }{ | \hat{ p } |^{ 2 + \frac{ \gamma - 1 }{ 2 } } } \int_0^\infty \frac{ e^{ - s_\eta^{ - 1 } | \hat{ q } | } }{ | \hat{ q } |^{ \frac{ \gamma - 1 }{ 2 } } } \, d | \hat{ q } | \, d^3 \hat{ p } \nonumber \\ 
& \leq C_A , 
\end{align}
since $ 1 < \gamma < 2 $. For $ K_{ 3 2 , - 1 } $ we easily obtain  
\begin{align}
K_{ 3 2 , - 1 } & \leq C_A \| b_1 - b_2 \| \iint \frac{ 1 }{ ( { p_1 }^0 )^{ 1 + \frac{ \gamma }{ 2 } } ( { q_1 }^0 )^{ \frac{ \gamma }{ 2 } } } f_1 ( p ) f_1 ( q ) \, d^3 q \, d^3 p \nonumber \\ 
& \leq C_A \| b_1 - b_2 \| \| f_1 \|_{ L^1_{ 1 , - 1 - \gamma / 2 } } \| f_1 \|_{ L^1_{ 1 , - \gamma / 2 } } . 
\end{align} 
\end{itemize} 
Now, $ \| f_1 \|_{ L^1_{ 1 , 2 - \gamma } } $, $ \| f_1 \|_{ L^1_{ 1 , - 1 } } $, $ \| f_1 \|_{ L^1_{ 1 , 1 - \gamma / 2 } } $, $ \| f_1 \|_{ L^1_{ 1 , - \gamma / 2 } } $, $ \| f_1 \|_{ L^\infty_{ 1 , \eta } } $, $ \| f_1 \|_{ L^1_{ 1 , - 1 - \gamma / 2 } } $ and $ \| f_1 \|_{ L^1_{ 1 , - \gamma / 2 } } $ are all bounded by Lemma \ref{lem dfdp}. Hence, we obtain \eqref{f-f3}. \medskip

\noindent (4) {\it Estimate of $ K_4 $.} Recall that 
\begin{align*}
K_{ 4 , r } = ( \det b_2 )^{ \frac12 } \iiint \frac{ { h_2 }^{ 2 - \gamma } }{ { p_2 }^0 { q_2 }^0 } | f_1 ( p_1' ) f_1 ( q_1' ) - f_1 ( p_2' ) f_1 ( q_2' ) | ( { p_1 }^0 )^r \, d \omega \, d^3 q \, d^3 p , 
\end{align*}
where we denote 
\begin{align}
p_1' = p' ( { b_1 }^{ i j } ) , \qquad q_1' = q' ( { b_1 }^{ i j } ) , \\ 
p_2' = p' ( { b_2 }^{ i j } ) , \qquad q_2' = q' ( { b_2 }^{ i j } ) . 
\end{align}
Hence, we may write 
\begin{align}
f_1 ( p_1' ) f_1 ( q_1' ) = F ( { b_1 }^{ i j } ) , \qquad f_1 ( p_2' ) f_1 ( q_2' ) = F ( { b_2 }^{ i j } ) , 
\end{align}
for some $ F = F ( b^{ k l } ) $. Then, we have 
\begin{align} 
F ( { b_1 }^{ i j } ) - F ( { b_2 }^{ i j } ) & = \int_1^2 \frac{ d ( F ( { b_\alpha }^{ i j } ) ) }{ d \alpha } \, d \alpha \nonumber \\ 
& = ( { b_2 }^{ k l } - { b_1 }^{ k l } ) \int_1^2 \frac{ \partial F }{ \partial b^{ k l } } ( { b_\alpha }^{ i j } ) \, d \alpha , 
\end{align}
where $ { b_\alpha }^{ i j } = ( \alpha - 1 ) { b_2 }^{ i j } + ( 2 - \alpha ) { b_1 }^{ i j } $, and this implies that 
\begin{align}
| F ( { b_1 }^{ i j } ) - F ( { b_2 }^{ i j } ) | \leq \| { b_2 } - { b_1 } \| \max_{ k , l } \int_1^2 \left| \frac{ \partial F }{ \partial b^{ k l } } ( { b_\alpha }^{ i j } ) \right| \, d \alpha , 
\end{align}
where 
\begin{align}
\frac{ \partial F }{ \partial b^{ k l } } ( { b_\alpha }^{ i j } ) = \frac{ \partial f_1 }{ \partial p } ( p_\alpha' ) \cdot \frac{ \partial p' }{ \partial b^{ k l } } ( { b_\alpha }^{ i j } ) f_1 ( q_\alpha' ) + f_1 ( p_\alpha' ) \frac{ \partial f_1 }{ \partial p } ( q_\alpha' ) \cdot \frac{ \partial q' }{ \partial b^{ k l } } ( { b_\alpha }^{ i j } ) . 
\end{align}
Now, applying Lemma \ref{lem dp'db} with respect to $ { b_\alpha }^{ i j } $, i.e., 
\begin{align}
\left| \frac{ \partial p' }{ \partial b^{ k l } } ( { b_\alpha }^{ i j } ) \right| + \left| \frac{ \partial q' }{ \partial b^{ k l } } ( { b_\alpha }^{ i j } ) \right| \leq C_A { n_\alpha }^0 , 
\end{align}
we obtain 
\begin{align}
K_{ 4 , r } \leq C_A \iiint \frac{ { h_2 }^{ 2 - \gamma } }{ { p_2 }^0 { q_2 }^0 } \| { b_2 } - { b_1 } \| \int_1^2 \left( \left| \frac{ \partial f_1 }{ \partial p } ( p_\alpha' ) \right| f_1 ( q_\alpha' ) + f_1 ( p_\alpha' ) \left| \frac{ \partial f_1 }{ \partial p } ( q_\alpha' ) \right| \right) { n_\alpha }^0 ( { p_1 }^0 )^r \, d \alpha \, d \omega \, d^3 q \, d^3 p . 
\end{align}
Note that the estimates \eqref{p^0 1} and \eqref{p^0 2} with respect to $ { b_1 }^{ i j } $ and $ { b_2 }^{ i j } $ imply 
\begin{align}
{ p_1 }^0 \leq C \| b_1 \|^{ \frac12 } \langle p \rangle, & \qquad \langle p \rangle \leq \frac{ C \| b_1 \|^{ \frac52 } }{ \det b_1 } { p_1 }^0 , \\ 
{ p_2 }^0 \leq C \| b_2 \|^{ \frac12 } \langle p \rangle, & \qquad \langle p \rangle \leq \frac{ C \| b_2 \|^{ \frac52 } }{ \det b_2 } { p_2 }^0 , 
\end{align}
so that $ { p_1 }^0 $ and $ { p_2 }^0 $ are equivalent by the assumptions on $ { b_1 }^{ i j } $, $ { b_2 }^{ i j } $, $ \det b_1 $ and $ \det b_2 $. Moreover, as in the proof of Lemma \ref{lem F-F}, we have 
\begin{align}
\| b_\alpha \| \leq C_A , \qquad \det b_\alpha \geq \frac{ 1 }{ A } , 
\end{align}
which imply that $ { p_1 }^0 $, $ { p_\alpha }^0 $ and $ { p_2 }^0 $ are all equivalent: 
\begin{align}
\frac{ 1 }{ C_A } \langle p \rangle \leq { p_\alpha }^0 \leq C_A \langle p \rangle , \qquad 1 \leq \alpha \leq 2 . 
\end{align}
Similarly, $ { q_1 }^0 $, $ { q_\alpha }^0 $ and $ { q_2 }^0 $ are all equivalent. Now, applying $ h_2 \leq { n_2 }^0 $ and the above equivalence properties, we obtain 
\begin{align}
K_{ 4 , r } & \leq C_A \| b_1 - b_2 \| \int_1^2 \iiint \frac{ ( { n_\alpha }^0 )^{ 3 - \gamma } }{ { p_\alpha }^0 { q_\alpha }^0 } \left( \left| \frac{ \partial f_1 }{ \partial p } ( p_\alpha' ) \right| f_1 ( q_\alpha' ) + f_1 ( p_\alpha' ) \left| \frac{ \partial f_1 }{ \partial p } ( q_\alpha' ) \right| \right) ( { p_\alpha }^0 )^r \, d \omega \, d^3 q \, d^3 p \, d \alpha . 
\end{align}

\begin{itemize}

\item $ ( r = 1 ) $ For $ r = 1 $, we have 
\begin{align}
K_{ 4 , 1 } & \leq C_A \| b_1 - b_2 \| \int_1^2 \iiint \frac{ ( { n_\alpha }^0 )^{ 3 - \gamma } }{ { p_\alpha }^0 { q_\alpha }^0 } \left( \left| \frac{ \partial f_1 }{ \partial p } ( p_\alpha' ) \right| f_1 ( q_\alpha' ) + f_1 ( p_\alpha' ) \left| \frac{ \partial f_1 }{ \partial p } ( q_\alpha' ) \right| \right) { p_\alpha }^0 \, d \omega \, d^3 q \, d^3 p \, d \alpha \nonumber \\ 
& \leq C_A \| b_1 - b_2 \| \int_1^2 \iiint \frac{ ( { n_\alpha }^0 )^{ 2 - \gamma } }{ { p_\alpha }^0 { q_\alpha }^0 } \left( \left| \frac{ \partial f_1 }{ \partial p } ( p ) \right| f_1 ( q ) + f_1 ( p ) \left| \frac{ \partial f_1 }{ \partial p } ( q ) \right| \right) \, d \omega \, d^3 q \, d^3 p \, d \alpha \nonumber \\ 
& \leq C_A \| b_1 - b_2 \| \int_1^2 \iiint \left( \frac{ 1 }{ ( { p_\alpha }^0 )^{ \gamma - 1 } { q_\alpha }^0 } + \frac{ 1 }{ { p_\alpha }^0 ( { q_\alpha }^0 )^{ \gamma - 1 } } \right) \left( \left| \frac{ \partial f_1 }{ \partial p } ( p ) \right| f_1 ( q ) + f_1 ( p ) \left| \frac{ \partial f_1 }{ \partial p } ( q ) \right| \right) \, d^3 q \, d^3 p \, d \alpha \nonumber \\ 
& \leq C_A \| b_1 - b_2 \| \int_1^2 \left\| \frac{ \partial f_1 }{ \partial p } \right\|_{ L^1_{ \alpha , - \gamma + 1 } } \| f_1 \|_{ L^1_{ \alpha , - 1 } } + \left\| \frac{ \partial f_1 }{ \partial p } \right\|_{ L^1_{ \alpha , - 1 } } \| f_1 \|_{ L^1_{ \alpha , - \gamma + 1 } } \, d \alpha \nonumber \\ 
& \leq C_A \| b_1 - b_2 \| \left( \left\| \frac{ \partial f_1 }{ \partial p } \right\|_{ L^1_{ 1 , - \gamma + 1 } } \| f_1 \|_{ L^1_{ 1 , - 1 } } + \left\| \frac{ \partial f_1 }{ \partial p } \right\|_{ L^1_{ 1 , - 1 } } \| f_1 \|_{ L^1_{ 1 , - \gamma + 1 } } \right) , 
\end{align}
where we used \eqref{dpdq1} with $ b_1 $ replaced by $ b_\alpha $, and $ { p_\alpha' }^0 \leq { n_\alpha }^0 $, in the second inequality. 

\item $ ( r = - 1 ) $ Recall that we have chosen $ \delta > 0 $ satisfying $ \gamma + \delta < 2 $. We first estimate $ K_{ 4 , - 1 } $ as follows: 
\begin{align}
K_{ 4 , - 1 } & \leq C_A \| b_1 - b_2 \| \int_1^2 \iiint \frac{ ( { n_\alpha }^0 )^{ 3 - \gamma } }{ { p_\alpha }^0 { q_\alpha }^0 } \left( \left| \frac{ \partial f_1 }{ \partial p } ( p_\alpha' ) \right| f_1 ( q_\alpha' ) + f_1 ( p_\alpha' ) \left| \frac{ \partial f_1 }{ \partial p } ( q_\alpha' ) \right| \right) \frac{ 1 }{ { p_\alpha }^0 } \, d \omega \, d^3 q \, d^3 p \, d \alpha \nonumber \\ 
& \leq C_A \| b_1 - b_2 \| \int_1^2 \iiint \frac{ ( { n_\alpha }^0 )^{ 3 - \gamma } }{ { p_\alpha }^0 { q_\alpha }^0 } \left( \left| \frac{ \partial f_1 }{ \partial p } ( p ) \right| f_1 ( q ) + f_1 ( p ) \left| \frac{ \partial f_1 }{ \partial p } ( q ) \right| \right) \frac{ 1 }{ { p_\alpha' }^0 } \, d \omega \, d^3 q \, d^3 p \, d \alpha \nonumber \\ 
& \leq C_A \| b_1 - b_2 \| \int_1^2 \iint \frac{ ( { n_\alpha }^0 )^{ 3 - \gamma } }{ { p_\alpha }^0 { q_\alpha }^0 } \left( \left| \frac{ \partial f_1 }{ \partial p } ( p ) \right| f_1 ( q ) + f_1 ( p ) \left| \frac{ \partial f_1 }{ \partial p } ( q ) \right| \right) \frac{ 1 }{ ( h_\alpha )^\delta ( { n_\alpha }^0 )^{ 1 - \delta } } \, d^3 q \, d^3 p \, d \alpha \nonumber \\ 
& \leq C_A \| b_1 - b_2 \| \int_1^2 \iint \frac{ ( { n_\alpha }^0 )^{ 2 - \gamma + \delta } }{ { p_\alpha }^0 { q_\alpha }^0 } \left| \frac{ \partial f_1 }{ \partial p } ( p ) \right| f_1 ( q ) \frac{ 1 }{ ( h_\alpha )^\delta } \, d^3 q \, d^3 p \, d \alpha \nonumber \\ 
& \leq C_A \| b_1 - b_2 \| \int_1^2 \iint \frac{ ( { n_\alpha }^0 )^{ 2 - \gamma + \delta } }{ ( { p_\alpha }^0 )^{ 1 + \frac{ \delta }{ 2 } } ( { q_\alpha }^0 )^{ 1 + \frac{ \delta }{ 2 } } \sin^\delta ( \phi / 2 ) } \left| \frac{ \partial f_1 }{ \partial p } ( p ) \right| f_1 ( q ) \, d^3 q \, d^3 p \, d \alpha \nonumber \\ 
& \leq C_A \| b_1 - b_2 \| \int_1^2 \iint \frac{ 1 }{ ( { p_\alpha }^0 )^{ \gamma - 1 - \frac{ \delta }{ 2 } } ( { q_\alpha }^0 )^{ 1 + \frac{ \delta }{ 2 } } \sin^\delta ( \phi / 2 ) } \left| \frac{ \partial f_1 }{ \partial p } ( p ) \right| f_1 ( q ) \, d^3 q \, d^3 p \, d \alpha \nonumber \\ 
& \quad + C_A \| b_1 - b_2 \| \int_1^2 \iint \frac{ 1 }{ ( { p_\alpha }^0 )^{ 1 + \frac{ \delta }{ 2 } } ( { q_\alpha }^0 )^{ \gamma - 1 - \frac{ \delta }{ 2 } } \sin^\delta ( \phi / 2 ) } \left| \frac{ \partial f_1 }{ \partial p } ( p ) \right| f_1 ( q ) \, d^3 q \, d^3 p \, d \alpha \nonumber \\ 
& = : K_{ 4 1 , - 1 } + K_{ 4 2 , - 1 } . 
\end{align}
Note that $ f_1 $ is bounded in $ L^\infty_{ 1 , \eta } $, which implies that 
\begin{align}
f_1 ( q ) \leq \frac{ C_A }{ { q_1 }^0 } e^{ - s_\eta^{ - 1 } { q_1 }^0 } \leq \frac{ C_A }{ { q_\alpha }^0 } e^{ - s^{ - 1 }_\eta { q_\alpha }^0 / C_A } . 
\end{align}
Then, $ K_{ 4 1 , - 1 } $ is estimated as 
\begin{align}
K_{ 4 1 , - 1 } & \leq C_A \| b_1 - b_2 \| \iint \frac{ 1 }{ ( { p_\alpha }^0 )^{ \gamma - 1 - \frac{ \delta }{ 2 } } } \left| \frac{ \partial f_1 }{ \partial p } ( p ) \right| \int \frac{ e^{ - s_\eta^{ - 1 } { q_\alpha }^0 / C_A } }{ ( { q_\alpha }^0 )^{ 2 + \frac{ \delta }{ 2 } } \sin^\delta ( \phi / 2 ) } \, d^3 q \, d^3 p \, d \alpha , 
\end{align}
where the integral over $ \bbr^3_q $ is finite, since $ 2 + \delta / 2 < 3 $ and $ \delta < 1 $. Hence, we have 
\begin{align}
K_{ 4 1 , - 1 } & \leq C_A \| b_1 - b_2 \| \left\| \frac{ \partial f_1 }{ \partial p } \right\|_{ L^1_{ 1 , - \gamma + 1 + \delta / 2 } } , 
\end{align}
where we used the equivalence between $ { p_\alpha }^0 $ and $ { p_1 }^0 $. The term $ K_{ 4 2 , - 1 } $ is estimated in a similar way: 
\begin{align}
K_{ 4 2 , - 1 } & \leq C_A \| b_1 - b_2 \| \iint \frac{ 1 }{ ( { p_\alpha }^0 )^{ 1 + \frac{ \delta }{ 2 } } } \left| \frac{ \partial f_1 }{ \partial p } ( p ) \right| \int \frac{ e^{ - s_\eta^{ - 1 } { q_\alpha }^0 / C_A } }{ ( { q_\alpha }^0 )^{ \gamma - \frac{ \delta }{ 2 } } \sin^\delta ( \phi / 2 ) } \, d^3 q \, d^3 p \, d \alpha \nonumber \\ 
& \leq C_A \| b_1 - b_2 \| \left\| \frac{ \partial f_1 }{ \partial p } \right\|_{ L^1_{ 1 , - 1 - \delta / 2 } } . 
\end{align}
\end{itemize}
We combine the estimates for $ K_{ 4 , 0 } $, $ K_{ 4 1 , - 1 } $ and $ K_{ 4 2 , - 1 } $ and use Lemma \ref{lem dfdp} to obtain \eqref{f-f4}. \medskip

\noindent (5) {\it Estimate of $ K_5 $.} Using \eqref{p^0 1}--\eqref{p^0 2}, we can replace the weight $ ( { p_1 }^0 )^r $ with $ ( { p_2 }^0 )^r $: 
\begin{align}
K_{ 5 , r } & \leq C_A \iiint \frac{ { h_2 }^{ 2 - \gamma } }{ { p_2 }^0 { q_2 }^0 } ( ( f_1 + f_2 ) ( p_2' ) | f_1 - f_2 | ( q_2' ) + | f_1 - f_2 | ( p_2' ) ( f_1 + f_2 ) ( q_2' ) \nonumber \\ 
& \qquad \qquad \qquad \qquad + ( f_1 + f_2 ) ( p ) | f_1 - f_2 | ( q ) ) ( { p_2 }^0 )^r \, d \omega \, d^3 q \, d^3 p \nonumber \\ 
& \leq C_A \iiint \frac{ { h_2 }^{ 2 - \gamma } }{ { p_2 }^0 { q_2 }^0 } ( f_1 + f_2 ) ( p ) | f_1 - f_2 | ( q ) \left\{ ( { p_2' }^0 )^r + ( { q_2' }^0 )^r + ( { p_2 }^0 )^r \right\} \, d \omega \, d^3 q \, d^3 p . 
\end{align} 

\begin{itemize}
\item $ ( r = 1 ) $ We apply $ h_2 \leq { n_2 }^0 $ and use the energy conservation \eqref{conservation} to obtain
\begin{align}
K_{ 5 , 1 } \leq C_A \left( \| f_1 - f_2 \|_{ L^1_{ 2 , - 1 } } + \| f_1 - f_2 \|_{ L^1_{ 2 , 2 - \gamma } } \right) . 
\end{align}

\item $ ( r = - 1 ) $ We apply the same estimates as in \eqref{L113} and \eqref{L111} to obtain 
\begin{align}
K_{ 5 , - 1 } \leq C_A \left( \| f_1 - f_2 \|_{ L^1_{ 2 , - \gamma / 2 } } + \| f_1 - f_2 \|_{ L^1_{ 2 , - 1 } } \right) . 
\end{align} 
\end{itemize} 
Since $ 1 < \gamma < 2 $, we obtain 
\begin{align}
K_5 \leq C_A \left( \| f_1 - f_2 \|_{ L^1_{ 2 , 1 } } + \| f_1 - f_2 \|_{ L^1_{ 2 , - 1 } } \right) .
\end{align}
By applying the equivalence between $ { p_1 }^0 $ and $ { p_2 }^0 $ we can replace the norms $ \| \cdot \|_{ L^1_{ 2 , 1 } } $ and $ \| \cdot \|_{ L^1_{ 2 , - 1 } } $ with $ \| \cdot \|_{ L^1_{ 1 , 1 } } $ and $ \| \cdot \|_{ L^1_{ 1 , - 1 } } $, respectively. Hence, we obtain \eqref{f-f5}, and this completes the proof of \eqref{prop B5}. \medskip

The proof of \eqref{prop B6} is almost the same. We have the equations \eqref{df ds1}--\eqref{df ds2} for $ f_1 $ and $ f_2 $, and obtain the expression \eqref{df ds3} for $ \partial f_1 / \partial s - \partial f_2 / \partial s $. Now, we have
\begin{align}
\left| \frac{ \partial f_1 }{ \partial s } - \frac{ \partial f_2 }{ \partial s } \right| & \leq | ( \det b_1 )^{ \frac12 } - ( \det b_2 )^{ \frac12 } | \iint \frac{ { h_1 }^{ 2 - \gamma } }{ { p_1 }^0 { q_1 }^0 } ( f_1 ( p_1' ) f_1 ( q_1' ) + f_1 ( p ) f_1 ( q ) ) \, d \omega \, d^3 q \nonumber \\ 
& \quad + ( \det b_2 )^{ \frac12 } \iint \left| \frac{ 1 }{ { p_1 }^0 { q_1 }^0 } - \frac{ 1 }{ { p_2 }^0 { q_2 }^0 } \right| { h_1 }^{ 2 - \gamma } ( f_1 ( p_1' ) f_1 ( q_1' ) + f_1 ( p ) f_1 ( q ) ) \, d \omega \, d^3 q \nonumber \\ 
& \quad + ( \det b_2 )^{ \frac12 } \iint \frac{ 1 }{ { p_2 }^0 { q_2 }^0 } | { h_1 }^{ 2 - \gamma } - { h_2 }^{ 2 - \gamma } | ( f_1 ( p_1' ) f_1 ( q_1' ) + f_1 ( p ) f_1 ( q ) ) \, d \omega \, d^3 q \nonumber \\ 
& \quad + ( \det b_2 )^{ \frac12 } \iint \frac{ { h_2 }^{ 2 - \gamma } }{ { p_2 }^0 { q_2 }^0 } | f_1 ( p_1' ) f_1 ( q_1' ) - f_1 ( p_2' ) f_1 ( q_2' ) | \, d \omega \, d^3 q \nonumber \\ 
& \quad + \frac{ ( \det b_2 )^{ \frac12 } }{ 2 } \iint \frac{ { h_2 }^{ 2 - \gamma } }{ { p_2 }^0 { q_2 }^0 } ( ( f_1 + f_2 ) ( p_2' ) | f_1 - f_2 | ( q_2' ) + | f_1 - f_2 | ( p_2' ) ( f_1 + f_2 ) ( q_2' ) \nonumber \\ 
& \qquad \qquad \qquad \qquad + ( f_1 + f_2 ) ( p ) | f_1 - f_2 | ( q ) + | f_1 - f_2 | ( p ) ( f_1 + f_2 ) ( q ) ) \, d \omega \, d^3 q ,
\end{align}
where the right hand side is the same as that of \eqref{df ds4} except the last quantity. Multiplying the above by $ ( { p_1 }^0 )^r $ and integrating over $ \bbr^3_p $, we obtain 
\begin{align}
& \left\| \frac{ \partial f_1 }{ \partial s } - \frac{ \partial f_2 }{ \partial s } \right\|_{ L^1_{ 1 , r } } \nonumber \\ 
& \leq | ( \det b_1 )^{ \frac12 } - ( \det b_2 )^{ \frac12 } | \iiint \frac{ { h_1 }^{ 2 - \gamma } }{ { p_1 }^0 { q_1 }^0 } ( f_1 ( p_1' ) f_1 ( q_1' ) + f_1 ( p ) f_1 ( q ) ) ( { p_1 }^0 )^r \, d \omega \, d^3 q \, d^3 p \nonumber \\ 
& \quad + ( \det b_2 )^{ \frac12 } \iiint \left| \frac{ 1 }{ { p_1 }^0 { q_1 }^0 } - \frac{ 1 }{ { p_2 }^0 { q_2 }^0 } \right| { h_1 }^{ 2 - \gamma } ( f_1 ( p_1' ) f_1 ( q_1' ) + f_1 ( p ) f_1 ( q ) ) ( { p_1 }^0 )^r \, d \omega \, d^3 q \, d^3 p \nonumber \\ 
& \quad + ( \det b_2 )^{ \frac12 } \iiint \frac{ 1 }{ { p_2 }^0 { q_2 }^0 } | { h_1 }^{ 2 - \gamma } - { h_2 }^{ 2 - \gamma } | ( f_1 ( p_1' ) f_1 ( q_1' ) + f_1 ( p ) f_1 ( q ) ) ( { p_1 }^0 )^r \, d \omega \, d^3 q \, d^3 p \nonumber \\ 
& \quad + ( \det b_2 )^{ \frac12 } \iiint \frac{ { h_2 }^{ 2 - \gamma } }{ { p_2 }^0 { q_2 }^0 } | f_1 ( p_1' ) f_1 ( q_1' ) - f_1 ( p_2' ) f_1 ( q_2' ) | ( { p_1 }^0 )^r \, d \omega \, d^3 q \, d^3 p \nonumber \\ 
& \quad + \frac{ ( \det b_2 )^{ \frac12 } }{ 2 } \iiint \frac{ { h_2 }^{ 2 - \gamma } }{ { p_2 }^0 { q_2 }^0 } ( ( f_1 + f_2 ) ( p_2' ) | f_1 - f_2 | ( q_2' ) + | f_1 - f_2 | ( p_2' ) ( f_1 + f_2 ) ( q_2' ) \nonumber \\ 
& \qquad \qquad \qquad \qquad + ( f_1 + f_2 ) ( p ) | f_1 - f_2 | ( q ) + | f_1 - f_2 | ( p ) ( f_1 + f_2 ) ( q ) ) ( { p_1 }^0 )^r \, d \omega \, d^3 q \, d^3 p \nonumber \\
& = K_{ 1 , r } + K_{ 2 , r } + K_{ 3 , r } + K_{ 4 , r } + K_{ 5 , r } + K_{ 6 , r } ,
\end{align}
where $ K_{ 1 , r } , \dots , K_{ 6 , r } $ are the same as in \eqref{f-fs}. Again, we write $ K_n = K_{ n , 1 } + K_{ n , - 1 } $ to have
\begin{align}
\left\| \frac{ \partial f_1 }{ \partial s } - \frac{ \partial f_2 }{ \partial s } \right\|_{ L^1_{ 1 , 1 } } + \left\| \frac{ \partial f_1 }{ \partial s } - \frac{ \partial f_2 }{ \partial s } \right\|_{ L^1_{ 1 , - 1 } } \leq K_1 + K_2 + K_3 + K_4 + K_5 + K_6 ,
\end{align}
and obtain the desired result \eqref{prop B6} by combining the estimates \eqref{f-f1}--\eqref{f-f6} together with \eqref{prop B5}. This completes the proof of the proposition.
\end{proof}

\section{Proofs of the main results}\label{sMain}
We are now ready to prove the main results of the paper. We first prove Theorem \ref{conformalprop}, where we need to combine Propositions \ref{propeinstein}, \ref{prop B} and \ref{prop C}, which we recall are given as follows:

\setcounter{prop}{0}
\begin{prop}
Suppose that there exist a time interval $[0,T]$ and positive constants $B_1$ and $B_2$ such that $f$ is defined on $[0,T]$ and satisfy
\begin{align*}
\sup_{0 \leq s \leq T} \| f(s) \|_{\langle -1 \rangle} + \sup_{0 \leq s \leq T} \| f(s) \|_{\langle 1 \rangle} & \leq B_1, \\
\sup_{0 \leq s \leq T} \left\| \frac{\partial f}{\partial s} (s) \right\|_{\langle -1 \rangle} + \sup_{0 \leq s \leq T} \left\| \frac{\partial f}{\partial s} (s) \right\|_{\langle 1 \rangle} & \leq B_2.
\end{align*}
Then, for any initial data $a_{0}, {b_0}, K_{0}, \hat{Z}_{0} \in S_3(\bbr)$ satisfying the Fuchsian conditions \eqref{FF1bis} and \eqref{FF3bis}, there exists $0 < T_B \leq T$ such that the Einstein equations \eqref{ODEx}--\eqref{ODE} have a unique solution $a_{ij}, b^{ij}, K_{ij}, \hat{Z}_{ij} \in C^1([0, T_B]; S_3(\bbr) )$ satisfying
\[
\sup_{ [0, T_B] } \max \left( \| a \| , \| b \| , \| K \| , \| \hat{Z} \| \right) \leq 2 \max \left( \| a_0 \| , \| b_0 \| , \| K_0 \| , \| \hat{Z}_0 \| \right), 
\]
where $T_B$ depends on $B_1$ and $B_2$. 
\end{prop}

\setcounter{prop}{2}
\begin{prop}
Suppose that there exist a time interval $ [ 0 , T ] $ and a positive constant $ A $ such that $ b^{ i j } $ and $ K_{ i j } $ are defined on $ [ 0 , T ] $ and satisfy
\[
\sup_{ 0 \leq s \leq T } \| b ( s ) \| \leq A , \qquad \sup_{ 0 \leq s \leq T } \| K ( s ) \| \leq A . 
\]
Then, for any initial data $ 0 \leq f_0 \in L^1_1 ( \bbr^3 ) \cap L^1_{ - 2 } ( \bbr^3 ) \cap L^\infty_\eta ( \bbr^3 ) $ for some $ 0 < \eta < 2 / A^2 $, there exists $ 0 < T_A \leq T $ such that the Boltzmann equation \eqref{B} has a unique non-negative solution $ f \in C^1 ( [ 0 , T_A ] ; L^1 ( \bbr^3 ) ) $ satisfying
\begin{align*}
\sup_{ 0 \leq s \leq T_A } \| f ( s ) \|_{ L^1_{ - 1 } } + \sup_{ 0 \leq s \leq T_A } \| f ( s ) \|_{ L^1_1 } + \sup_{ 0 \leq s \leq T_A } \| f ( s ) \|_{ L^\infty_\eta } & \leq C_A , \\
\sup_{ 0 \leq s \leq T_A } \left\| \frac{ \partial f }{ \partial s } ( s ) \right\|_{ L^1_{ - 1 } } + \sup_{ 0 \leq s \leq T_A } \left\| \frac{ \partial f }{ \partial s } ( s ) \right\|_{ L^1_1 } & \leq C_{A} ,
\end{align*}
where $ T_A $ and $ C_A $ are positive constants depending on $ A $ and $ \eta $.
\end{prop}

\setcounter{prop}{3}
\begin{prop}
Suppose that there exist a time interval $ [ 0 , T ] $ and a positive constant $ A $ such that $ { b_1 }^{ i j } $, $ { K_1 }_{ i j } $, $ { b_2 }^{ i j } $ and $ { K_2 }_{ i j } $ are defined on $ [ 0 , T ] $ and satisfy $ { b_1 }^{ i j } ( 0 ) = { b_2 }^{ i j } ( 0 ) $, $ { K_1 }_{ i j } ( 0 ) = { K_2 }_{ i j } ( 0 ) $ and 
\begin{align*}
\sup_{ 0 \leq s \leq T } \| b_1 ( s ) \| \leq A , \qquad \sup_{ 0 \leq s \leq T } \| K_1 ( s ) \| \leq A , \qquad \inf_{ 0 \leq s \leq T } \det b_1 ( s ) \geq \frac{ 1 }{ A } , \\ 
\sup_{ 0 \leq s \leq T } \| b_2 ( s ) \| \leq A , \qquad \sup_{ 0 \leq s \leq T } \| K_2 ( s ) \| \leq A , \qquad \inf_{ 0 \leq s \leq T } \det b_2 ( s ) \geq \frac{ 1 }{ A } . 
\end{align*}
Then, for any initial data $ f_0 \geq 0 $ such that 
\begin{align*}
f_0 \in L^1_1 ( \bbr^3 ) \cap L^1_{ - 2 - \delta / 2 } ( \bbr^3 ) \cap L^\infty_\eta ( \bbr^3 ) , \qquad \frac{ \partial f_0 }{ \partial p } \in L^1_1 ( \bbr^3 ) \cap L^1_{ - 1 - \delta / 2 } ( \bbr^3 ) , 
\end{align*}
where $ 0 < \eta < 2 / A^2 $, and $ \delta > 0 $ is a number satisfying
\[
\gamma + \delta < 2 , 
\]
there exists $ 0 \leq T_A \leq T $ such that the Boltzmann equation \eqref{B} has unique non-negative solutions $ f_1 $ and $ f_2 $, corresponding to $ { b_1 }^{ i j } $ and $ { b_2 }^{ i j } $ respectively, such that 
\begin{align*}
\| f_1 ( s ) - f_2 ( s ) \|_{ L^1_{ 1 , 1 } } + \| f_1 ( s ) - f_2 ( s ) \|_{ L^1_{ 1 , - 1 } } & \leq C_A \sup_{ 0 \leq \rho \leq s } \| b_1 ( \rho ) - b_2 ( \rho ) \| , \\
\left\| \frac{ \partial f_1 }{ \partial s } ( s ) - \frac{ \partial f_2 }{ \partial s } ( s ) \right\|_{ L^1_{ 1 , 1 } } + \left\| \frac{ \partial f_1 }{ \partial s } ( s ) - \frac{ \partial f_2 }{ \partial s } ( s ) \right\|_{ L^1_{ 1 , - 1 } } & \leq C_A \sup_{ 0 \leq \rho \leq s } \| b_1 ( \rho ) - b_2 ( \rho ) \| ,
\end{align*}
for $ 0 \leq s \leq T_A $, where $ C_A $ is a positive constant depending on $ A $, $ \eta $ and $ \delta $. 
\end{prop}

\subsection{Proof of Theorem \ref{conformalprop}}\label{sMain1}
We now prove the existence of a solution to the rescaled Einstein-Boltzmann system. The system consists of the rescaled Einstein equations \eqref{F1}--\eqref{F4}, which are rewritten as \eqref{ODEx}--\eqref{ODE}, and the Boltzmann equation \eqref{V5}, which reduces to \eqref{B} in the Bianchi I case. Initial data will be given at $ s = 0 $, which we assume satisfy the Fuchsian conditions \eqref{FF1}--\eqref{FF3}, or equivalently \eqref{FF1bis}--\eqref{FF3bis}, and the Einstein constraint \eqref{C3}.

The proposition is proved by an iteration scheme with geometric convergence, the limit of which is our solution. At each stage we have functions $ a $, $ b $, $ K $, $ \hat Z $ and $ f $ defined on a rescaled time interval whose length we will specify in a moment. For our zeroth approximation we take $ a $, $ b $, $ K $, $ \hat Z $ and $ f $ equal to the given initial values, subject to the constraints \eqref{FF1bis}--\eqref{FF3bis} and \eqref{C3}. These do not, of course, satisfy the Einstein equations \eqref{ODEx}--\eqref{ODE} in general, or the Boltzmann equation \eqref{B}. At each stage in the iteration we form the Einstein part, $ a $, $ b $, $ K $, $ \hat Z $ to be the unique solutions to \eqref{ODEx}--\eqref{ODE} corresponding to the Boltzmann part, $ f $, at the previous stage, and the given initial data. Similarly, we take the Boltzmann part, $ f $, to be the unique solution to the Boltzmann equation~(\ref{B}) corresponding to the Einstein part, $ a $, $ b $, $ K $, $ \hat Z $, at the previous stage and the given initial data.

We need to make precise the domain of definition of these solutions. Assume the conditions of Theorem \ref{conformalprop}. Let $ a_k $, $ b_k $, $ K_k $, $ \hat{ Z }_k $ and $ f_k $ be the approximation described above, with initial data $ a_0 $, $ b_0 $, $ K_0 $, $ \hat{ Z }_0 $ and $ f_0 $, which are equal to the zeroth approximation. Let $ A > 0 $ denote
\begin{equation}
2 \max \left( \| a_0 \| , \| b_0 \| , \| K_0 \| , \| \hat{Z}_0 \| \right) = A.
\end{equation}
Then, $ b_0 $ and $ K_0 $ satisfy the conditions of Proposition \ref{prop B} on $ [ 0 , \infty ) $. Suppose that we have for some $ k \geq 0 $ a time interval $ [ 0 , T ] $ on which $ a_k $, $ b_k $, $ K_k $ and $ \hat{ Z }_k $ satisfy
\begin{align}
\sup_{ [0, T ] } \max \left( \| a_k \| , \| b_k \| , \| K_k \| , \| \hat{Z}_k \| \right) \leq 2 \max \left( \| a_0 \| , \| b_0 \| , \| K_0 \| , \| \hat{Z}_0 \| \right) = A .
\end{align}
Then, we obtain by Proposition \ref{prop B} an interval $ [ 0 , T_A ] \subset [ 0 , T ] $ on which $ f_{ k + 1 } $ satisfies
\begin{align}
\sup_{ 0 \leq s \leq T_A } \| f_{ k + 1 } ( s ) \|_{ L^1_{ - 1 } } + \sup_{ 0 \leq s \leq T_A } \| f_{ k + 1 } ( s ) \|_{ L^1_1 } & \leq C_A , \\
\sup_{ 0 \leq s \leq T_A } \left\| \frac{ \partial f_{ k + 1 } }{ \partial s } ( s ) \right\|_{ L^1_{ - 1 } } + \sup_{ 0 \leq s \leq T_A } \left\| \frac{ \partial f_{ k + 1 } }{ \partial s } ( s ) \right\|_{ L^1_1 } & \leq C_{A} .
\end{align}
Applying \eqref{p^0 1}--\eqref{p^0 2}, we have
\begin{align}
\sup_{0 \leq s \leq T_A } \| f_{ k + 1 } (s) \|_{\langle -1 \rangle} + \sup_{0 \leq s \leq T_A } \| f_{ k + 1 } (s) \|_{\langle 1 \rangle} & \leq B_A , \\
\sup_{0 \leq s \leq T_A } \left\| \frac{\partial f_{ k + 1 } }{\partial s} (s) \right\|_{\langle -1 \rangle} + \sup_{0 \leq s \leq T_A } \left\| \frac{\partial f_{ k + 1 } }{\partial s} (s) \right\|_{\langle 1 \rangle} & \leq B_A ,
\end{align}
for some different $ B_A $. Now, we can apply Proposition \ref{propeinstein} to obtain an interval $ [ 0 , T_B ] \subset [ 0 , T_A ] $ on which $ a_{ k + 1 } $, $ b_{ k + 1 } $, $ K_{ k + 1 } $ and $ \hat{ Z }_{ k + 1 } $ satisfy
\begin{align}
\sup_{ [0, T_B] } \max \left( \| a_{ k + 1 } \| , \| b_{ k + 1 } \| , \| K_{ k + 1 } \| , \| \hat{Z}_{ k + 1 } \| \right) \leq 2 \max \left( \| a_0 \| , \| b_0 \| , \| K_0 \| , \| \hat{Z}_0 \| \right) = A .
\end{align}
We notice that $ b_{ k + 1 } $ and $ K_{ k + 1 } $ satisfy the conditions of Proposition \ref{prop B} with the same number $ A $ as in the previous step, which implies that $ f_{ k + 2 } $ exists on the same interval $ [ 0 , T_B ] $ and satisfies
\begin{align}
\sup_{ 0 \leq s \leq T_B } \| f_{ k + 2 } ( s ) \|_{ L^1_{ - 1 } } + \sup_{ 0 \leq s \leq T_B } \| f_{ k + 2 } ( s ) \|_{ L^1_1 } & \leq C_A , \\
\sup_{ 0 \leq s \leq T_B } \left\| \frac{ \partial f_{ k + 2 } }{ \partial s } ( s ) \right\|_{ L^1_{ - 1 } } + \sup_{ 0 \leq s \leq T_B } \left\| \frac{ \partial f_{ k + 2 } }{ \partial s } ( s ) \right\|_{ L^1_1 } & \leq C_{A} ,
\end{align}
with the same constant $ C_A $. We conclude that there exists an interval $ [ 0 , T ] $, on which $ a_k $, $ b_k $, $ K_k $ and $ \hat{ Z }_k $ are bounded in $ S_3 ( \bbr ) $, and $ f_k $ and $ \partial f_k / \partial s $ are bounded in $ L^1_{ - 1 } ( \bbr^3 ) \cap L^1_1 ( \bbr^3 ) $.

Next, we need to show that the sequences $ a_k $, $ b_k $, $ K_k $, $ \hat{ Z }_k $ and $ f_k $ converge. We first notice that $ f_k $ depends continuously on $ b_k $ by Proposition \ref{prop C}. To be precise, we have
\begin{align}
\| f_{ k + 1 } ( s ) - f_{ k + 2 } ( s ) \|_{ \langle 1 \rangle } + \| f_{ k + 1 } ( s ) - f_{ k + 2 } ( s ) \|_{ \langle - 1 \rangle } & \leq C \sup_{ 0 \leq \rho \leq s } \| b_{ k } ( \rho ) - b_{ k + 1 } ( \rho ) \| , \\
\left\| \frac{ \partial f_{ k + 1 } }{ \partial s } ( s ) - \frac{ \partial f_{ k + 2 } }{ \partial s } ( s ) \right\|_{ \langle 1 \rangle } + \left\| \frac{ \partial f_{ k + 1 } }{ \partial s } ( s ) - \frac{ \partial f_{ k + 2 } }{ \partial s } ( s ) \right\|_{ \langle - 1 \rangle } & \leq C \sup_{ 0 \leq \rho \leq s } \| b_{ k } ( \rho ) - b_{ k + 1 } ( \rho ) \| ,
\end{align}
for $ 0 \leq s \leq T $, where we used \eqref{p^0 1}--\eqref{p^0 2}. On the other hand, we need to show that $ a_k $, $ b_k $, $ K_k $ and $ \hat{ Z }_k $ depend continuously on $ f_k $. By Theorem \ref{ODEthm}, we observe that $ a_k $, $ b_k $, $ K_k $ and $ \hat{ Z }_k $ depend continuously on the coefficient functions. To be precise, let us write the $ k $-th approximation as
\begin{equation}\begin{aligned}
  x_k ' & = F ( x_k , y_k ) , \qquad & x_k ( 0 ) = x _ 0 , \\
  s y_k ' + N_k ( s , x_k ) y_k & = s G ( s , x_k , y_k ) + H_k ( s , x_k ) , \qquad & y_k ( 0 ) = y _ 0 ,
\end{aligned}\end{equation}
where we used the notations in \eqref{ODEx}--\eqref{N}. We notice that $ F $ and $ G $ do not depend on $ k $, but $ N_k $ and $ H_k $ change as the iteration step increases, since they depend on $ f_k $ through $ \chi_{ i j }^{ m n } $, $ \Pi_{ i j }^{ m n } $ and $ H_{ i j } $. Now, applying Theorem \ref{ODEthm} we have
\begin{align}
\| x_k ( s ) - x_{ k + 1 } ( s ) \| \leq 2 L s , \qquad \| y_k ( s ) - y_{ k + 1 } ( s ) \| \leq 2 L s ,
\end{align}
where $ L $ should be understood as
\begin{align}
L & = C_1 \sup_{ 0 \leq \rho \leq s } \left\| \frac{ \partial H_k }{ \partial s } ( \rho ) - \frac{ \partial H_{ k + 1 } }{ \partial s } ( \rho ) \right\| + C_2 \sup_{ 0 \leq \rho \leq s } \left\| \frac{ \partial H_k }{ \partial x } ( \rho ) - \frac{ \partial H_{ k + 1 } }{ \partial x } ( \rho ) \right\| \nonumber \\
& \quad + C_3 \sup_{ 0 \leq \rho \leq s } \left\| \frac{ \partial N_k }{ \partial s } ( \rho ) - \frac{ \partial N_{ k + 1 } }{ \partial s } ( \rho ) \right\| + C_4 \sup_{ 0 \leq \rho \leq s } \left\| \frac{ \partial N_k }{ \partial x } ( \rho ) - \frac{ \partial N_{ k + 1 } }{ \partial x } ( \rho ) \right\| ,
\end{align}
for some finite $ C_1 , \dots , C_4 $. Let us first consider the $ N_k - N_{ k + 1 } $ terms. We only need to estimate the derivatives of $ \chi_k - \chi_{ k + 1 } $ and $ \Pi_k- \Pi_{ k + 1 } $ at $ ( s , x ) $, which are given by
\begin{align}
\chi_k = { \chi_k }_{ i j }^{ q r } = { \Psi }_{ k \; i j m n } { b }^{ m q } { b }^{ n r } , \qquad \Pi_k = { \Pi_k }_{ i j }^{ m n } = \Psi_{ k \; i j } b^{ m n } ,
\end{align}
where $ { \Psi }_{ k \; i j m n } $ and $ \Psi_{ k \; i j } $ are defined by
\begin{align}
{ \Psi }_{ k \; i j m n } & = \frac { 4 \pi } { \sqrt { \det a } } \int \frac{ p_i p_j p_m p_n }{ ( { p }^0 )^3 } f_k ( p ) \, d^3 p , \\
{ \Psi }_{ k \; i j } & = \frac { 4 \pi } { \sqrt { \det a } } \int \frac{ p_i p_j }{ { p }^0 } f_k ( p ) \, d^3 p .
\end{align}
We can use the computations in Section \ref{sDN} to obtain
\begin{align}
\left\| \frac{ \partial \chi_k }{ \partial s } - \frac{ \partial \chi_{ k + 1 } }{ \partial s } \right\| \leq C \| \Phi_k^{ ( 4 ) } - \Phi_{ k + 1 }^{ ( 4 ) } \| ,
\end{align}
where $ \Phi_k^{ ( 4 ) } $ is defined by \eqref{Phi} using $ f_k $, and the superscript $ ( 4 ) $ denotes the number of indices it has. Similarly, we obtain
\begin{align}
& \left\| \frac{ \partial \chi_k }{ \partial a } - \frac{ \partial \chi_{ k + 1 } }{ \partial a } \right\| \leq C \| \Psi_k^{ ( 4 ) } - \Psi_{ k + 1 }^{ ( 4 ) } \| , \\
& \left\| \frac{ \partial \chi_k }{ \partial b } - \frac{ \partial \chi_{ k + 1 } }{ \partial b } \right\| \leq C ( \| \Psi_k^{ ( 4 ) } - \Psi_{ k + 1 }^{ ( 4 ) } \| + \| \Psi_k^{ ( 6 ) } - \Psi_{ k + 1 }^{ ( 6 ) } \| ) ,
\end{align}
and
\begin{align}
& \left\| \frac{ \partial \Pi_k }{ \partial s } - \frac{ \partial \Pi_{ k + 1 } }{ \partial s } \right\| \leq C \| \Phi_k^{ ( 2 ) } - \Phi_{ k + 1 }^{ ( 2 ) } \| , \\
& \left\| \frac{ \partial \Pi_k }{ \partial a } - \frac{ \partial \Pi_{ k + 1 } }{ \partial a } \right\| \leq C \| \Psi_k^{ ( 2 ) } - \Psi_{ k + 1 }^{ ( 2 ) } \| , \\
& \left\| \frac{ \partial \Pi_k }{ \partial b } - \frac{ \partial \Pi_{ k + 1 } }{ \partial b } \right\| \leq C ( \| \Psi_k^{ ( 2 ) } - \Psi_{ k + 1 }^{ ( 2 ) } \| + \| \Psi_k^{ ( 4 ) } - \Psi_{ k + 1 }^{ ( 4 ) } \| ) .
\end{align}
We can easily estimate the quantities on the right hand sides above as
\begin{align}
\| \Phi_k^{ ( n ) } - \Phi_{ k + 1 }^{ ( n ) } \| \leq C \left\| \frac{ \partial f_k }{ \partial s } - \frac{ \partial f_{ k + 1 } }{ \partial s } \right\|_{ \langle 1 \rangle } , \qquad \| \Psi_k^{ ( n ) } - \Psi_{ k + 1 }^{ ( n ) } \| \leq C \| f_k - f_{ k + 1 } \|_{ \langle 1 \rangle } ,
\end{align}
for any $ n $. Hence, we obtain
\begin{align}
& \sup_{ 0 \leq \rho \leq s } \left\| \frac{ \partial N_k }{ \partial s } ( \rho ) - \frac{ \partial N_{ k + 1 } }{ \partial s } ( \rho ) \right\| + \sup_{ 0 \leq \rho \leq s } \left\| \frac{ \partial N_k }{ \partial x } ( \rho ) - \frac{ \partial N_{ k + 1 } }{ \partial x } ( \rho ) \right\| \nonumber \\
& \leq C \sup_{ 0 \leq \rho \leq s } \| f_k ( \rho ) - f_{ k + 1 } ( \rho ) \|_{ \langle 1 \rangle } + C \sup_{ 0 \leq \rho \leq s } \left\| \frac{ \partial f_k }{ \partial s } ( \rho ) - \frac{ \partial f_{ k + 1 } }{ \partial s } ( \rho ) \right\|_{ \langle 1 \rangle } .
\end{align}
Next, we recall that $ H_k $ is given by \eqref{H} using $ f_k $, so that we have
\begin{align}
H_k - H_{ k + 1 } & = ( H_k - H_{ k + 1 } )_{ i j } \nonumber \\
& = \frac { 8 \pi c_\gamma  } { \det a }  \iiint  \frac{ { h } ^ { 2 - \gamma }}{ { p }^0 {q }^0} f_k (p) ( f_k (q) - f_{ k + 1 } ( q ) ) \left( \frac{ p' _ i p'_ j}{ p'^0} - \frac{p_i p_j}{ { p }^0}\right) \, d \omega \, d ^ 3 q \, d ^ 3 p \nonumber \\
& \quad + \frac { 8 \pi c_\gamma  } { \det a }  \iiint  \frac{ { h } ^ { 2 - \gamma }}{ { p }^0 {q }^0} ( f_k ( p ) - f_{ k + 1 } ( p ) ) f_{ k + 1 } ( q ) \left( \frac{ p' _ i p'_ j}{ p'^0} - \frac{p_i p_j}{ { p }^0}\right) \, d \omega \, d ^ 3 q \, d ^ 3 p .
\end{align}
We use the computations in Section \ref{sDH} to obtain
\begin{align}
& \sup_{ 0 \leq \rho \leq s } \left\| \frac{ \partial H_k }{ \partial s } ( \rho ) - \frac{ \partial H_{ k + 1 } }{ \partial s } ( \rho ) \right\| + \sup_{ 0 \leq \rho \leq s } \left\| \frac{ \partial H_k }{ \partial x } ( \rho ) - \frac{ \partial H_{ k + 1 } }{ \partial x } ( \rho ) \right\| \nonumber \\
& \leq C \sup_{ 0 \leq \rho \leq s } \| f_k ( \rho ) - f_{ k + 1 } ( \rho ) \|_{ \langle - 1 \rangle } + C \sup_{ 0 \leq \rho \leq s } \| f_k ( \rho ) - f_{ k + 1 } ( \rho ) \|_{ \langle 1 \rangle } \nonumber \\
& \quad + C \sup_{ 0 \leq \rho \leq s } \left\| \frac{ \partial f_k }{ \partial s } ( \rho ) - \frac{ \partial f_{ k + 1 } }{ \partial s } ( \rho ) \right\|_{ \langle - 1 \rangle } + C \sup_{ 0 \leq \rho \leq s } \left\| \frac{ \partial f_k }{ \partial s } ( \rho ) - \frac{ \partial f_{ k + 1 } }{ \partial s } ( \rho ) \right\|_{ \langle 1 \rangle } .
\end{align}
We combine the above estimates to obtain
\begin{align}
& \| a_k ( s ) - a_{ k + 1 } ( s ) \| + \| b_k ( s ) - b_{ k + 1 } ( s ) \| + \| K_k ( s ) - K_{ k + 1 } ( s ) \| + \| { \hat Z }_k ( s ) - { \hat Z }_{ k + 1 } ( s ) \| \nonumber \\
& \leq C s \bigg( \sup_{ 0 \leq \rho \leq s } \| f_k ( \rho ) - f_{ k + 1 } ( \rho ) \|_{ \langle - 1 \rangle } + \sup_{ 0 \leq \rho \leq s } \| f_k ( \rho ) - f_{ k + 1 } ( \rho ) \|_{ \langle 1 \rangle } \nonumber \\
& \qquad \quad + \sup_{ 0 \leq \rho \leq s } \left\| \frac{ \partial f_k }{ \partial s } ( \rho ) - \frac{ \partial f_{ k + 1 } }{ \partial s } ( \rho ) \right\|_{ \langle - 1 \rangle } + C \sup_{ 0 \leq \rho \leq s } \left\| \frac{ \partial f_k }{ \partial s } ( \rho ) - \frac{ \partial f_{ k + 1 } }{ \partial s } ( \rho ) \right\|_{ \langle 1 \rangle } \bigg) .
\end{align}
Let us write
\begin{align}
& \| f_{ k } ( s ) - f_{ k + 1 } ( s ) \| \nonumber \\
& = \| f_k ( s ) - f_{ k + 1 } ( s ) \|_{ \langle - 1 \rangle } + \| f_k ( s ) - f_{ k + 1 } ( s ) \|_{ \langle 1 \rangle } + \left\| \frac{ \partial f_k }{ \partial s } ( s ) - \frac{ \partial f_{ k + 1 } }{ \partial s } ( s ) \right\|_{ \langle - 1 \rangle } + \left\| \frac{ \partial f_k }{ \partial s } ( s ) - \frac{ \partial f_{ k + 1 } }{ \partial s } ( s ) \right\|_{ \langle 1 \rangle } , \\
& \| ( a , b , K , \hat Z ) _ { k } ( s ) - ( a , b , K , \hat Z ) _ { k + 1 } ( s ) \| \nonumber \\
& = \| a_k ( s ) - a_{ k + 1 } ( s ) \| + \| b_k ( s ) - b_{ k + 1 } ( s ) \| + \| K_k ( s ) - K_{ k + 1 } ( s ) \| + \| { \hat Z }_k ( s ) - { \hat Z }_{ k + 1 } ( s ) \| .
\end{align}
Now, we can find a time interval, which we still denote by $ [ 0 , T ] $, on which we have
\begin{align}
\sup_{ 0 \leq s \leq T }  \| f_{ k + 1 } ( s ) - f_{ k + 2 } ( s ) \| & \leq \frac 1 2 \sup_{ 0 \leq s \leq T } \| f_k ( s ) - f_{ k + 1 } ( s ) \| , \\
\sup_{ 0 \leq s \leq T }  \| ( a , b , K , \hat Z ) _ { k } ( s ) - ( a , b , K , \hat Z ) _ { k + 1 } ( s ) \| & \leq \frac 1 2 \sup_{ 0 \leq s \leq T } \| ( a , b , K , \hat Z ) _ { k - 1 } ( s ) - ( a , b , K , \hat Z ) _ { k } ( s ) \| .
\end{align}
We therefore have geometric convergence of the sequences. The limits satisfy the coupled system consisting of the Einstein equations~\eqref{ODEx}--\eqref{ODE} and the Boltzmann equation~\eqref{B}. This finishes the proof of Theorem \ref{conformalprop}.  \qed

\subsection{Proof of Theorem \ref{phystheorem}}\label{sMain2}
Let $ f_0 \geq 0 $ be a smooth function with compact support in $ \bbr^3 \setminus \{ 0 \} $ such that it is not identically zero and satisfies the constraint \eqref{C3}. By Theorem \ref{fuchsian} we obtain a unique set of initial data $ a_0 $, $ b_0 $, $ K_0 $, $ \hat{ Z}_0 $ and $ f_0 $ for the rescaled Einstein-Boltzmann system \eqref{V5}, \eqref{F1}--\eqref{F4} satisfying the Fuchsian conditions \eqref{FF1}--\eqref{FF3}. Note that $ f_0 $ satisfies the conditions of Theorem \ref{conformalprop}, and we obtain a unique solution $ a_{ i j } $, $ b^{ i j } $, $ K_{ i j } $, $ { \hat Z }_{ i j } $ and $ f $ by Theorem \ref{conformalprop}. Finally, using the relations \eqref{physicalk}, \eqref{ttau}, \eqref{aa} and \eqref{stau}, we obtain the existence of the physical metric $ \tilde{ a }_{ i j } $, the physical rate of change $ \tilde{ k }_{ i j } $ and the physical distribution function $ f $ on a time interval $ ( 0 , T ] $ such that $ f $ converges to $ f_0 $ in $ L^1 $ as $ t \to 0^+ $.

It remains to prove the asymptotics \eqref{1}--\eqref{2}. From Proposition \ref{propeinstein} we have that $a$ and $K$ are $C^1$ functions of $s$ and by \eqref{F1} we have:
\begin{align}
\frac{d a_{ij}}{ds} = K_{ij}.
\end{align}
It therefore follows that $ a _ { i j } $ is $ C ^ 2 $,
$ d a _ { i j } / d s $ is $ C ^ 1 $ and $ d ^ 2a _ { i j } / d s ^ 2 $
is $ C ^ 0 $, all as functions of $ s $.
By Taylor's Theorem we have for small $s$ that
\begin{align}
  \label{taylor1} a _ { i j }  &= a _ { 0 i j } + a _ { 1 i j } s + a _ { 2 i j } s ^ 2
  + o ( s ^ 2 ) , \\
  \label{taylor2}  \frac { d a _ { i j } } { d s } & = a _ { 1 i j } + 2 a _ { 2 i j } s + o ( s ) ,\\
   \label{taylor3}  \frac { d ^ 2 a _ { i j } } { d s ^ 2 }& = 2 a _ { 2 i j } + o ( 1 ) ,
\end{align}
where $a _ { 0 i j }$, $a _ { 1 i j }$ and $a _ { 2 i j } $ are some constants.

Since we have a unique $a_{ij}$ by \eqref{aa} and using \eqref{ttau} we have that
\begin{align}\label{phys1}
\tilde{a}_{ij} = \tau^2 a_{ij} = 2 t  a_{ij}.
\end{align}
Now using the relations between $t$ and $\tau$ \eqref{ttau} and $s$ and $\tau$ \eqref{stau} we obtain the relation between $s$ and $t$:
\begin{align}\label{st}
s = \frac{1}{\gamma-1} (2t)^{\frac{\gamma-1}{2}},
\end{align}
and 
\begin{align}\label{dsdt}
\frac{ds}{dt} = (2t)^{\frac{\gamma - 3}{2}}.
\end{align}
As a consequence of \eqref{phys1} using \eqref{taylor1} and \eqref{st} we can express $\tilde{a}_{ij}$ in terms of $t$ as follows:
\begin{align}\label{finaltildea}
 \tilde a _ { i j } = 2 t a _ { 0 i j }
    + ( 2 t ) ^ { \frac { \gamma + 1 } 2 } \frac { 1 } { \gamma - 1 } a _ { 1 i j }
    + ( 2 t ) ^ \gamma \frac { 1 } { ( \gamma - 1 ) ^ 2 } a _ { 2 i j }
    + o ( t ^ \gamma ).
    \end{align}
Now we consider the $t$-derivative of \eqref{phys1} using \eqref{dsdt}
   \begin{align}
   \frac{d \tilde{a}_{ij}}{dt}  = 2a_{ij} +2t \frac{ds}{dt} \frac{d a_{ij}}{ds} = 2a_{ij} +(2t)^{\frac{\gamma-1}{2}}  \frac{d a_{ij}}{ds} ,
      \end{align}
which using  \eqref{taylor1}--\eqref{taylor2} becomes
    \begin{align}\label{finaltildek}
     \frac{d \tilde{a}_{ij}}{dt} =  2 a _ { 0 i j } + ( 2 t ) ^ { \frac { \gamma - 1 } 2 } \frac { \gamma + 1 } { \gamma - 1 } a _ { 1 i j }
  + ( 2 t ) ^ { \gamma - 1 } \frac { 2 \gamma } { ( \gamma - 1 ) ^ 2 } a _ { 2 i j }
  + o ( t ^ { \gamma - 1 } ).
    \end{align}
Similarly using \eqref{taylor2}--\eqref{taylor3}
   \begin{align}
      \frac{d ^2\tilde{a}_{ij}}{dt^2} = ( 2 t ) ^ { \frac { \gamma - 3 } 2 } ( \gamma + 1 ) a _ { 1 i j }
  + ( 2 t ) ^ { \gamma - 2 } \frac { 4 \gamma } { \gamma - 1 } a _ { 2 i j } + o(t^{\gamma-2}).
     \end{align}
This implies that $ d ^2\tilde{a}_{ij} / dt^2 $ will not be continuous as a function of $t$ in general, whereas $\tilde{a}$ and $\tilde{k}$ are continuous as functions of $t$. Absorbing constants in \eqref{finaltildea} and \eqref{finaltildek} using the notation
     \begin{align}
  \label{not}   \mathcal{A}_{ij} = 2 a_{0ij},\quad  \mathcal{B}_{ij} = \frac{1}{\gamma-1} 2^{\frac{\gamma+1}{2}} a_{1ij},\quad  \mathcal{C}_{ij}  = \frac{1}{(\gamma-1)^2} 2^{\gamma} a_{2ij},
     \end{align}
together with \eqref{physicalk} we obtain the asymptotics \eqref{1}--\eqref{2}. \qed

\section*{Appendix}
\newcommand{\tmin}{t _ { \min }}
We are interested in the existence and uniqueness of solutions
to the initial value problem
\begin{equation}\label{IVP}\begin{aligned}
  x ' ( t ) & = F ( t , x ( t ) , y ( t ) ) ,
  \qquad & x ( 0 ) = x _ 0 ,
  \\
  t y ' ( t ) + N ( t , x ( t ) ) y ( t ) & = t G ( t , x ( t ) , y ( t ) ) + H ( t , x ( t ) ) ,
  \qquad
&  y ( 0 ) = y _ 0 .
\end{aligned}\end{equation}
An obvious necessary condition for the existence of solutions is the
compatibility condition
\begin{equation}\label{compatibility}
  N _ 0 y _ 0 = H _ 0 ,
\end{equation}
obtained by substituting $ t = 0 $ in the equations above.
Here
\begin{align}
  N _ 0 = N ( 0 , x _ 0 ) , \quad H _ 0 = H ( 0 , x _ 0 ) ,
\end{align}

We suppose that $ V $ and $ W $ are vector spaces,
$ U $ is an open subset of $ V $, $ x _ 0 \in U $,
$ y _ 0 \in W $, $ T > 0 $,
and
$ F \colon [ 0 , T ] \times U \times W \to V $,
$ G \colon [ 0 , T ] \times U \times W \to W $,
$ H \colon [ 0 , T ] \times U $
and
$ N \colon [ 0 , T ] \times U \to \hom ( W , W ) $
are continuous along with
$ \partial F / \partial x $,
$ \partial F / \partial y $,
$ \partial G / \partial x $,
$ \partial G / \partial y $,
$ \partial H / \partial t $,
$ \partial H / \partial x $,
$ \partial ^ 2 H / \partial t \partial x $,
$ \partial ^ 2 H / \partial x ^ 2 $,
$ \partial N / \partial t $,
$ \partial N / \partial x $,
$ \partial ^ 2 N / \partial t \partial x $ and
$ \partial ^ 2 N / \partial x ^ 2 $.
Here and elsewhere $ t $~derivatives
denote derivatives with respect to the first
argument, $ x $~derivatives denote derivatives with respect to the
second argument and $ y $~derivatives denote derivatives with respect
to the third argument of these functions, regardless of where those
derivatives happen to be evaluated.
By a solution we of course mean continuously differentiable functions
$ x \colon [ 0 , \tmin ] \to U $ and $ y \colon [ 0 , \tmin ] \to W $,
for some $ \tmin \in ( 0 , T ] $,
which satisfy~(\ref{IVP}).

There's a further condition on $ N _ 0 $ which will be needed to
ensure existence and uniqueness.
There are, by Linear Algebra, unique subspaces $ W _ n $ and $ W _ i $
such that $ W = W _ n \oplus W _ i $ with $ N _ 0 W _ n \subseteq W _ n $ and
$ N _ 0 W _ i \subseteq W _ i $, with $ N _ n = N _ 0 | _ { W _ n } $ nilpotent
and $ N _ i = N _ 0 | _ { W _ i } $ invertible.
We assume that $ N _ n = 0 $ and that all complex eigenvalues of
$ N _ i $ have positive real part.

Note that it is not necessary to compute the subspaces $ W _ n $ or
$ W _ i $ in order to check this. We simply need to apply the
Hurwitz criterion to $ x ^ { - d } p ( x ) $ where $ d $ is the
dimension of the nullspace of $ N _ 0 $ and $ p $ is the characteristic
polynomial of~$ N _ 0 $.
Note that our assumptions on $ N _ 0 $ are exactly the necessary and sufficient
conditions for $ P $, defined initially on $ ( 0 , \infty ) $ by
\begin{align}
  P ( t ) = \exp ( ( \log t ) N _ 0 ) ,
\end{align}
to extend continously to $ [ 0 , \infty ) $. The value of this extension
at $ 0 $ is easily seen to be the unique projection with image $ W _ n $
and nullspace $ W _ i $.
From the definition
and the properties of the matrix exponential it follows that
\begin{gather}
  P ( t ) W _ n \subseteq W _ n ,
  \quad
  P ( t ) W _ i \subseteq W _ i ,
  \\
  P ( 1 ) = I ,
  \quad
  P ( s t ) = P ( s ) P ( t ) = P ( t ) P ( s ) ,
  \quad
  t P ' ( t ) = N _ 0 P ( t ) = P ( t ) N _ 0
\end{gather}
for all $ s , t \ge 0 $.

\begin{thm}\label{ODEthm}
Assuming that the coefficient functions $ F $, $ G $, $ H $ and $ N $
have the differentiability properties assumed above and that $ N _ 0 $
satisfies the eigenvalue condition described above, there is a unique
solution to the initial value problem~(\ref{IVP}) for all initial
data satisfying the compatibility condition~(\ref{compatibility}).
Furthermore, this solution depends continuously on $ F $, $ G $,
$ H $, $ N $, $ x _ 0 $ and $ y _ 0 $.

The solution is continuous for $ t \ge 0 $ and is continuously
differentiable and satisfies the differential equation equation for $ t > 0 $.
If the eigenvalues of $ N _ i $ have real parts which are not merely
positive but also greater than 1 then the solution is continuously
differentiable and satisfies the equation for $ t \ge 0 $.

More precisely, suppose $ R > 0 $ is such that $ \overline { B _ V ( x _ 0 , R ) } \subseteq U $.
Let $ K _ 2 $ and $ K _ 3 $ be the compact sets
$ K _ 2 = [ 0 , T ] \times \overline { B _ V ( x _ 0 , R ) } $
and $ K _ 3 = [ 0 , T ] \times \overline { B _ V ( x _ 0 , R ) }
\times \overline { B _ W ( y _ 0 , R ) } $. The solutions $ x $ and $ y $
then exist at least up until the time
\begin{align}
  \tmin = \min \left ( T , \frac R { Q _ x } , \frac R { Q _ y } , \frac 1 { 2 Q _ u } , \frac 1 { 2 Q _ { v } } \right ) ,
\end{align}
where
\begin{align}
  Q _ x & = C _ F , \\
  Q _ u & = C _ { F x } + C _ { F y } , \\
  Q _ y & = C _ P \left [ C _ G + C _ { H t } + C _ { H x } C _ F + \left ( C _ { N t } + C _ { N x } C _ F \right ) C _ y \right ] , \\
  Q _ v & =
    C _ P C _ { G x }
    + C _ P C _ { G y }
    + C _ P C _ { H t x }
    + C _ P C _ { H x x } C _ F
    + C _ P C _ { H x } C _ { F x } \nonumber
   \\ &\quad
    + C _ P C _ { H x } C _ { F y }
    + C _ P C _ { N t x } C _ y
    + C _ P C _ { N t } 
    + C _ P C _ { N x x } C _ F C _ y \nonumber
   \\ & \quad
    + C _ P C _ { N x } C _ { F x } C _ y
    + C _ P C _ { N x } C _ { F y } C _ y
    + C _ P C _ { N x } C _ F ,
\end{align}
where the constants $ C $ with various subscripts are such that
\begin{align}
  \max _ { K _ 3 } \| F \|
  & \le
  C _ F
  &
  \max _ { K _ 3 } \| G \|
  & \le
  C _ G
  &
  \max _ { K _ 3 } \| \partial F / \partial x \|
  & \le
  C _ { F x }
  \\
  \max _ { K _ 3 } \| \partial G / \partial x \|
  & \le
  C _ { G x }
  &
  \max _ { K _ 3 } \| \partial F / \partial y \|
  & \le
  C _ { F y }
  &
  \max _ { K _ 3 } \| \partial G / \partial y \|
  & \le
  C _ { G y }
  \\
  \max _ { K _ 2 } \left \| \partial H / \partial t \right \|
  & \le
  C _ { H t }
  &
  \max _ { K _ 2 } \left \| \partial N / \partial t \right \|
  & \le
  C _ { N t }
  &
  \max _ { K _ 2 } \left \| \partial H / \partial x \right \|
  & \le
  C _ { H x }
  \\
  \max _ { K _ 2 } \left \| \partial N / \partial x \right \|
  & \le
  C _ { N x }
  &
  \max _ { K _ 2 } \left \| \partial ^ 2 H / \partial t \partial x \right \|
  & \le
  C _ { H t x }
  & 
  \max _ { K _ 2 } \left \| \partial ^ 2 N / \partial t \partial x \right \|
  & \le
  C _ { N t x }
  \\
  \max _ { K _ 2 } \left \| \partial ^ 2 H / \partial x ^ 2 \right \|
  & \le
  C _ { H x x }
  &
  \max _ { K _ 2 } \left \| \partial ^ 2 N / \partial x ^ 2 \right \|
  & \le
  C _ { N x x }
  \\
  \max _ { [ 0 , 1 ] } \| P \|
  & \le
  C _ P
  &
  \| y _ 0 \| + R
  & \le
  C _ y .
\end{align}

Also, this solution depends in a continuous way on $ x _ 0 $, $ y _ 0 $,
$ N $, $ F $, $ G $ and~$ H $.
Suppose that $ \tilde N $, $ \tilde F $, $ \tilde G $
and $ \tilde H $ satisfy the same differentiability
hypotheses as we have assumed for
$ N $, $ F $, $ G $ and $ H $, the same bounds,
and the compatibility condition
\begin{align}
  \tilde N ( 0 , x _ 0 ) y _ 0
  = \tilde H ( 0 , x _ 0 ) .
\end{align}
Then there is a unique solution
$ ( \tilde x , \tilde y ) $ to the
initial value problem
\begin{equation}\label{IVP2}\begin{aligned}
  \tilde x ' ( t ) & = \tilde F ( t , \tilde x ( t ) , \tilde y ( t ) ) ,
  \qquad & \tilde x ( 0 ) = x _ 0 ,
  \\
  t \tilde y ' ( t ) + \tilde N ( t , \tilde x ( t ) ) \tilde y ( t ) & = t \tilde G ( t , \tilde x ( t ) , \tilde y ( t ) ) + \tilde H ( t , \tilde x ( t ) ) ,
  \qquad
 & \tilde y ( 0 ) = y _ 0 .
\end{aligned}\end{equation}
on $ [ 0 , \tmin ] $. For $ t $ in this interval we have
\begin{align}
  \| \tilde x ( t ) - x ( t ) \| \le 2 L t
\end{align}
and
\begin{align}
  \| \tilde y ( t ) - y ( t ) \| \le 2 L t
\end{align}
where
\begin{gather}
  L = \max ( L _ x , L _ y ) ,
\end{gather}
and
\begin{align}
  L _ x & = \max _ { K _ 3 } \| \tilde F - F \| , \\
  L _ y
  & =
  L _ P \max _ { [ 0 , 1 ] } \| \tilde P - P \|
  + L _ F \max _ { K _ 3 } \| \tilde F - F \|
  + L _ G \max _ { K _ 3 } \| \tilde G - G \|
  + L _ { H t } \max _ { K _ 2 } \left \| \frac { \partial \tilde H } { \partial t } - \frac { \partial H } { \partial t } \right \| \nonumber
  \\ & \quad
  + L _ { H x } \max _ { K _ 2 } \left \| \frac { \partial \tilde H } { \partial x } - \frac { \partial H } { \partial x } \right \|
  + L _ { N t } \max _ { K _ 2 } \left \| \frac { \partial \tilde N } { \partial t } - \frac { \partial N } { \partial t } \right \|
  + L _ { N x } \max _ { K _ 2 } \left \| \frac { \partial \tilde N } { \partial x } - \frac { \partial N } { \partial x } \right \| ,
\end{align}
with
\begin{gather}
  L _ P = C _ G + C _ { H t } + C _ { H x } C _ F + C _ { N t } C _ y + C _ { N x } C _ F C _ y ,
  \qquad
  L _ F = C _ P C _ { H x } + C _ P C _ { N x } C _ y ,
  \\
  L _ G = L _ { H t } = C _ P ,
  \qquad
  L _ { H x } = C _ P C _ F ,
  \qquad
  L _ { N t } = C _ P C _ y ,
  \qquad
  L _ { N x } = C _ P C _ F C _ y .
\end{gather}
\end{thm}
\begin{proof}
We start with the proof of uniqueness. Suppose then that $ x $ and
$ y $ are solutions of the initial value problem~(\ref{IVP}) and the
compatibility condition~(\ref{compatibility}).
Integrating the first equation.
\begin{align}
  x ( t ) = x _ 0 + \int _ 0 ^ t x ' ( s ) \, d s
  = x _ 0 + \int _ 0 ^ t F ( s , x ( s ) , y ( s ) ) \, d s .
\end{align}
The second equation requires more work.
We rewrite it as
\begin{align}
  \left ( t \frac d { d t } + N _ 0 \right ) y ( t ) - H _ 0
  = t G ( t , x ( t ) , y ( t ) )
  + \Delta H ( t , x ( t ) ) - \Delta N ( t , x ( t ) ) y ( t )
\end{align}
where
\begin{align}
  \Delta H ( t , x ) = H ( t , x ) - H _ 0 ,
  \quad
  \Delta N ( t , x ) = N ( t , x ) - N _ 0 .
\end{align}
Now
\begin{align}
  \left ( t \frac d { d t } + N _ 0 \right ) y _ 0
  = N _ 0 y _ 0 = H _ 0
\end{align}
so we can rewrite the differential equation as
\begin{align}
  \left ( t \frac d { d t } + N _ 0 \right ) \left ( y ( t ) - y _ 0 \right )
  = t G ( t , x ( t ) , y ( t ) )
  + \Delta H ( t , x ( t ) ) - \Delta N ( t , x ( t ) ) y ( t ) .
\end{align}
We then multiply from the left by $ P ( t ) $, using the fact that as operators
\begin{align}
  t \frac d { d t } P ( t )
  = P ( t ) \left ( t \frac d { d t } + N _ 0 \right ) ,
\end{align}
to obtain
\begin{align}
  t \frac d { d t } \left [ P ( t )
  \left ( y ( t ) - y _ 0 \right ) \right ]
  = P ( t ) \left [ t G ( t , x ( t ) , y ( t ) )
  + \Delta H ( t , x ( t ) ) - \Delta N ( t , x ( t ) ) y ( t ) \right ] .
\end{align}
This is where it matters whether the real parts of the eigenvalues
are merely positive or are greater than 1. In the former case the
calculation above is valid for $ t > 0 $, but $ d P / d t $
does not extend continuously to $ t = 0 $. In the latter case
it does.

We now examine in more detail the $ \Delta H $ and $ \Delta N $ terms.
\begin{align}
  \Delta H ( t , x ( t ) )
  & = \int _ 0 ^ t \frac d { d r } H ( r , x ( r ) ) \, d r
  = \int _ 0 ^ t \left ( \frac { \partial H } { \partial t } ( r , x ( r ) ) + \frac { \partial H } { \partial x } ( r , x ( r ) ) x ' ( r ) \right ) \, d r \nonumber 
  \\ & = \int _ 0 ^ t \left ( \frac { \partial H } { \partial t } ( r , x ( r ) ) + \frac { \partial H } { \partial x } ( r , x ( r ) ) F ( r , x ( r ) , y ( r ) ) \right ) \, d r \nonumber 
  \\ & = t \int _ 0 ^ 1 \left ( \frac { \partial H } { \partial t } ( \rho t , x ( \rho t ) ) + \frac { \partial H } { \partial x } ( \rho t , x ( \rho t ) ) F ( \rho t , x ( \rho t ) , y ( \rho t ) ) \right ) \, d \rho
\end{align}
and similarly for $ \Delta N $.
Substituting this into the differential equation gives
\begin{align}
  \frac d { d t } \left [ P ( t )
  \left ( y ( t ) - y _ 0 \right ) \right ]
  & = P ( t ) G ( t , x ( t ) , y ( t ) )
  + P ( t ) \int _ 0 ^ 1 \frac { \partial H } { \partial t } ( \rho t , x ( \rho t ) ) \, d \rho \nonumber 
  \\ & \quad + P ( t ) \int _ 0 ^ 1 \frac { \partial H } { \partial x } ( \rho t , x ( \rho t ) ) F ( \rho t , x ( \rho t ) , y ( \rho t ) ) \, d \rho \nonumber 
  \\ & \quad {} - P ( t ) \int _ 0 ^ 1 \frac { \partial N } { \partial t } ( \rho t , x ( \rho t ) ) y ( t ) \, d \rho \nonumber 
  \\ & \quad {} - P ( t ) \int _ 0 ^ 1 \frac { \partial N } { \partial x } ( \rho t , x ( \rho t ) ) F ( \rho t , x ( \rho t ) , y ( \rho t ) ) y ( t ) \, d \rho .
\end{align}
Integrating this,
\begin{align}
  P ( t )
  \left ( y ( t ) - y _ 0 \right )
  & = \int _ 0 ^ t P ( s ) G ( s , x ( s ) , y ( s ) ) \, d s
  + \int _ 0 ^ t P ( s ) \int _ 0 ^ 1 \frac { \partial H } { \partial t } ( \rho s , x ( \rho s ) ) \, d \rho \, d s \nonumber 
  \\ & \quad {} + \int _ 0 ^ t P ( s ) \int _ 0 ^ 1 \frac { \partial H } { \partial x } ( \rho s , x ( \rho s ) ) F ( \rho s , x ( \rho s ) , y ( \rho s ) ) \, d \rho \, d s \nonumber 
  \\ & \quad {} - \int _ 0 ^ t P ( s ) \int _ 0 ^ 1 \frac { \partial N } { \partial t } ( \rho s , x ( \rho s ) ) y ( s ) \, d \rho \, d s \nonumber 
  \\ & \quad {} - \int _ 0 ^ t P ( s ) \int _ 0 ^ 1 \frac { \partial N } { \partial x } ( \rho s , x ( \rho s ) ) F ( \rho s , x ( \rho s ) , y ( \rho s ) ) y ( s ) \, d \rho \, d s .
\end{align}
Multiplying from the left by $ P ( 1 / t ) $,
\begin{align}
  y ( t )
  & = y _ 0 + \int _ 0 ^ t P ( s / t ) G ( s , x ( s ) , y ( s ) ) \, d s
  + \int _ 0 ^ t P ( s / t ) \int _ 0 ^ 1 \frac { \partial H } { \partial t } ( \rho s , x ( \rho s ) ) \, d \rho \, d s \nonumber 
  \\ & \quad + \int _ 0 ^ t P ( s / t ) \int _ 0 ^ 1 \frac { \partial H } { \partial x } ( \rho s , x ( \rho s ) ) F ( \rho s , x ( \rho s ) , y ( \rho s ) ) \, d \rho \, d s
  - \int _ 0 ^ t P ( s / t ) \int _ 0 ^ 1 \frac { \partial N } { \partial t } ( \rho s , x ( \rho s ) ) \, d \rho \, d s \nonumber 
  \\ & \quad {} - \int _ 0 ^ t P ( s / t ) \int _ 0 ^ 1 \frac { \partial N } { \partial x } ( \rho s , x ( \rho s ) ) F ( \rho s , x ( \rho s ) , y ( \rho s ) ) y ( s ) \, d \rho \, d s \nonumber 
  \\ & = y _ 0 + t \int _ 0 ^ 1 P ( \sigma ) G ( \sigma t , x ( \sigma t ) , y ( \sigma t ) ) \, d \sigma
  + t \int _ 0 ^ 1 \int _ 0 ^ 1 P ( \sigma ) \frac { \partial H } { \partial t } ( \rho \sigma t , x ( \rho \sigma t ) ) \, d \rho \, d \sigma \nonumber 
  \\ & \quad + t \int _ 0 ^ 1 \int _ 0 ^ 1 P ( \sigma ) \frac { \partial H } { \partial x } ( \rho \sigma t , x ( \rho \sigma t ) ) F ( \rho \sigma t , x ( \rho \sigma t ) , y ( \rho \sigma t ) ) \, d \rho \, d \sigma \nonumber 
  \\ & \quad {} - t \int _ 0 ^ 1 \int _ 0 ^ 1 P ( \sigma ) \frac { \partial N } { \partial t } ( \rho \sigma t , x ( \rho \sigma t ) ) y ( \sigma t ) \, d \rho \, d \sigma \nonumber 
  \\ & \quad {} - t \int _ 0 ^ 1 \int _ 0 ^ 1 P ( \sigma ) \frac { \partial N } { \partial x } ( \rho \sigma t , x ( \rho \sigma t ) ) F ( \rho \sigma t , x ( \rho \sigma t ) , y ( \rho \sigma t ) ) y ( \sigma t ) \, d \rho \, d \sigma .
\end{align}

Suppose $ \tau \in ( 0 , \tmin ) $.
Let $ M $ be the set of functions from $ [ 0 , \tau ] $ to
$ \overline { B _ V ( x _ 0 , R ) } \times
\overline { B _ W ( y _ 0 , R ) } $, with the metric
\begin{align}
  d ( ( x _ 1 , y _ 1 ) , ( x _ 2 , y _ 2 ) )
  = \sup _ { [ 0 , \tau ] } \max ( \| x _ 1 - x _ 2 \| , \| y _ 1 - y _ 2 \| ) .
\end{align}
We define a function $ \Phi \colon M \to M $ by
\begin{align}
  \Phi ( x , y ) = ( u , v )
\end{align}
where
\begin{align}
  u ( t )
&  = x _ 0 + t \int _ 0 ^ 1 F ( \sigma t , x ( \sigma t ) , y ( \sigma t ) ) \, d \sigma , \\
  v ( t )
  & = y _ 0 + t \int _ 0 ^ 1 P ( \sigma ) G ( \sigma t , x ( \sigma t ) , y ( \sigma t ) ) \, d \sigma
  + t \int _ 0 ^ 1 \int _ 0 ^ 1 P ( \sigma ) \frac { \partial H } { \partial t } ( \rho \sigma t , x ( \rho \sigma t ) ) \, d \rho \, d \sigma \nonumber
  \\ & \quad {} + t \int _ 0 ^ 1 \int _ 0 ^ 1 P ( \sigma ) \frac { \partial H } { \partial x } ( \rho \sigma t , x ( \rho \sigma t ) ) F ( \rho \sigma t , x ( \rho \sigma t ) , y ( \rho \sigma t ) ) \, d \rho \, d \sigma \nonumber
  \\ & \quad {} - t \int _ 0 ^ 1 \int _ 0 ^ 1 P ( \sigma ) \frac { \partial N } { \partial t } ( \rho \sigma t , x ( \rho \sigma t ) ) y ( \sigma t ) \, d \rho \, d \sigma \nonumber
  \\ & \quad {} - t \int _ 0 ^ 1 \int _ 0 ^ 1 P ( \sigma ) \frac { \partial N } { \partial x } ( \rho \sigma t , x ( \rho \sigma t ) ) F ( \rho \sigma t , x ( \rho \sigma t ) , y ( \rho \sigma t ) ) y ( \sigma t ) \, d \rho \, d \sigma .
\end{align}
Note that
\begin{align}
  \| u ( t ) - x _ 0 \| & \le C _ F \tau = Q _ x \tau \le R ,
  \\
  \| v ( t ) - y _ 0 \| & \le C _ P \left ( C _ G + C _ { H t } + C _ { H x } C _ F + C _ { N t } C _ y + C _ { N x } C _ { F } C _ y \right ) \tau
  = Q _ y \tau \le R ,
\end{align}
so $ \Phi $ does indeed map $ M $ to itself.
We have just seen that a solution of the initial value problem~(\ref{IVP})
is a fixed point of $ \Phi $. If we can show that $ \Phi $ is a contraction
mapping then the Banach Fixed Point Theorem will imply the uniqueness of
solutions of the initial value problem.

Suppose that
\begin{align}
  ( u _ 1 , v _ 1 ) = \Phi ( x _ 1 , y _ 1 )
  \qquad
  ( u _ 2 , v _ 2 ) = \Phi ( x _ 2 , y _ 2 )
\end{align}
for $ ( x _ 1 , y _ 1 ) , ( x _ 2 , y _ 2 ) \in M $.
Then
\begin{align}
  u _ 1 ( t ) - u _ 2 ( t )
  & =
  t \int _ 0 ^ 1
  \left [
    F ( \sigma t , x _ 1 ( \sigma t ) , y _ 1 ( \sigma t ) )
    - F ( \sigma t , x _ 2 ( \sigma t ) , y _ 2 ( \sigma t ) )
  \right ]
  \, d \sigma \nonumber
  \\ & =
  t \int _ 0 ^ 1
  \int _ 0 ^ 1
  \frac d { d \kappa } F ( \sigma t , \xi ( \kappa , \sigma t ) , \eta ( \kappa , \sigma t ) )
  \, d \kappa \, d \sigma \nonumber
  \\ & =
  t \int _ 0 ^ 1
  \int _ 0 ^ 1
    \frac { \partial F } { \partial x }
      ( \sigma t , \xi ( \kappa , \sigma t ) , \eta ( \kappa , \sigma t ) )
      ( x _ 1 ( \sigma t ) - x _ 2 ( \sigma t ) )
  \, d \kappa \, d \sigma \nonumber
  \\ & \quad +
  t \int _ 0 ^ 1
  \int _ 0 ^ 1
    \frac { \partial F } { \partial y }
      ( \sigma t , \xi ( \kappa , \sigma t ) , \eta ( \kappa , \sigma t ) )
      ( y _ 1 ( \sigma t ) - y _ 2 ( \sigma t ) )
  \, d \kappa \, d \sigma
\end{align}
where
\begin{align}
  \xi ( \kappa , s ) = \kappa x _ 1 ( s ) + ( 1 - \kappa ) x _ 2 ( s ) ,
  \qquad
  \eta ( \kappa , s ) = \kappa y _ 1 ( s ) + ( 1 - \kappa ) y _ 2 ( s ) .
\end{align}
Hence
\begin{align}
  \sup _ { [ 0 , \tau ] } \| u _ 1 - u _ 2 \|
  & \le C _ { F x } \tau \sup _ { [ 0 , \tau ] } \| x _ 1 - x _ 2 \|
  + C _ { F y } \tau \sup _ { [ 0 , \tau ] } \| y _ 1 - y _ 2 \| \nonumber
  \\ & \le ( C _ { F x } + C _ { F y } ) \tau \max \left (
    \sup _ { [ 0 , \tau ] } \| x _ 1 - x _ 2 \| ,
    \sup _ { [ 0 , \tau ] } \| y _ 1 - y _ 2 \|
  \right ) \nonumber
  \\ & = Q _ u \tau d ( ( x _ 1 , y _ 1 ) , ( x _ 2 , y _ 2 ) ) \nonumber
  \\ & \le \frac 1 2 d ( ( x _ 1 , y _ 1 ) , ( x _ 2 , y _ 2 ) ) .
\end{align}
The calculation for $ v _ 1 - v _ 2 $ is similar and yields
\begin{align}
  \sup _ { [ 0 , \tau ] } \| v _ 1 - v _ 2 \|
  \le Q _ v \tau
  d ( ( x _ 1 , y _ 1 ) , ( x _ 2 , y _ 2 ) )
  \le \frac 1 2
  d ( ( x _ 1 , y _ 1 ) , ( x _ 2 , y _ 2 ) ) .
\end{align}
So
\begin{gather}
  d ( \Phi ( x _ 1 , y _ 1 ) , \Phi ( x _ 2 , y _ 2 ) )
  = d ( ( u _ 1 , v _ 1 ) , ( u _ 2 , v _ 2 ))
  \le \frac 1 2 d ( ( x _ 1 , y _ 1 ) , ( x _ 2 , y _ 2 ) )
\end{gather}
and $ \Phi $ is a contraction with contraction constant $ 1 / 2 $.
The Banach Fixed Point Theorem then shows that $ \Phi $ has a unique
fixed point and hence that there is at most one solution to the
initial value problem~(\ref{IVP}). To show that there is at least
one solution we need to show that every fixed point of $ \Phi $
is indeed a solution of the initial value problem.

Suppose then that $ \Phi ( x , y ) = ( x , y ) $, i.e. that
\begin{align}
  x ( t )
&  = x _ 0 + t \int _ 0 ^ 1 F ( \sigma t , x ( \sigma t ) , y ( \sigma t ) ) \, d \sigma , \\
  y ( t )
  & = y _ 0 + t \int _ 0 ^ 1 P ( \sigma ) G ( \sigma t , x ( \sigma t ) , y ( \sigma t ) ) \, d \sigma
  + t \int _ 0 ^ 1 \int _ 0 ^ 1 P ( \sigma ) \frac { \partial H } { \partial t } ( \rho \sigma t , x ( \rho \sigma t ) ) \, d \rho \, d \sigma \nonumber
  \\ & \quad {} + t \int _ 0 ^ 1 \int _ 0 ^ 1 P ( \sigma ) \frac { \partial H } { \partial x } ( \rho \sigma t , x ( \rho \sigma t ) ) F ( \rho \sigma t , x ( \rho \sigma t ) , y ( \rho \sigma t ) ) \, d \rho \, d \sigma \nonumber
  \\ & \quad {} - t \int _ 0 ^ 1 \int _ 0 ^ 1 P ( \sigma ) \frac { \partial N } { \partial t } ( \rho \sigma t , x ( \rho \sigma t ) ) y ( \sigma t ) \, d \rho \, d \sigma \nonumber
  \\ & \quad {} - t \int _ 0 ^ 1 \int _ 0 ^ 1 P ( \sigma ) \frac { \partial N } { \partial x } ( \rho \sigma t , x ( \rho \sigma t ) ) F ( \rho \sigma t , x ( \rho \sigma t ) , y ( \rho \sigma t ) ) y ( \sigma t ) \, d \rho \, d \sigma .
\end{align}
Then
\begin{align}
  x ( t ) = x _ 0 + \int _ 0 ^ t F ( s , x ( s ) , y ( s ) ) \, d s
\end{align}
and hence $ x ( 0 ) = x _ 0 $ and
\begin{align}
  x ' ( t ) = F ( t , x ( t ) , y ( t ) ) .
\end{align}
For $ y $ we have to work a bit harder. Clearly, $ y ( 0 ) = y _ 0 $. Also
\begin{align}
  y ( t )
  & = y _ 0 + \int _ 0 ^ t P ( s / t ) G ( s , x ( s ) , y ( s ) ) \, d s
  + \int _ 0 ^ t P ( s / t ) \int _ 0 ^ 1 \frac { \partial H } { \partial t } ( \rho s , x ( \rho s ) ) \, d \rho \, d s \nonumber
  \\ & \quad + \int _ 0 ^ t P ( s / t ) \int _ 0 ^ 1 \frac { \partial H } { \partial x } ( \rho s , x ( \rho s ) ) F ( \rho s , x ( \rho s ) , y ( \rho s ) ) \, d \rho \, d s \nonumber
  \\ & \quad {} - \int _ 0 ^ t P ( s / t ) \int _ 0 ^ 1 \frac { \partial N } { \partial t } ( \rho s , x ( \rho s ) ) \, d \rho \, d s \nonumber
  \\ & \quad {} - \int _ 0 ^ t P ( s / t ) \int _ 0 ^ 1 \frac { \partial N } { \partial x } ( \rho s , x ( \rho s ) ) F ( \rho s , x ( \rho s ) , y ( \rho s ) ) y ( s ) \, d \rho \, d s
\end{align}
and so, multiplying from the left by $ P ( t ) $,
\begin{align}
  P ( t )
  \left ( y ( t ) - y _ 0 \right )
  & = \int _ 0 ^ t P ( s ) G ( s , x ( s ) , y ( s ) ) \, d s
  + \int _ 0 ^ t P ( s ) \int _ 0 ^ 1 \frac { \partial H } { \partial t } ( \rho s , x ( \rho s ) ) \, d \rho \, d s \nonumber
  \\ & \quad {} + \int _ 0 ^ t P ( s ) \int _ 0 ^ 1 \frac { \partial H } { \partial x } ( \rho s , x ( \rho s ) ) F ( \rho s , x ( \rho s ) , y ( \rho s ) ) \, d \rho \, d s \nonumber
  \\ & \quad {} - \int _ 0 ^ t P ( s ) \int _ 0 ^ 1 \frac { \partial N } { \partial t } ( \rho s , x ( \rho s ) ) y ( s ) \, d \rho \, d s \nonumber
  \\ & \quad {} - \int _ 0 ^ t P ( s ) \int _ 0 ^ 1 \frac { \partial N } { \partial x } ( \rho s , x ( \rho s ) ) F ( \rho s , x ( \rho s ) , y ( \rho s ) ) y ( s ) \, d \rho \, d s .
\end{align}
Differentiating, and multiplying by $ t $,
\begin{align}
  t \frac d { d t } \left [ P ( t ) ( y ( t ) - y _ 0 ) \right ]
  & = t P ( t ) G ( t , x ( t ) , y ( t ) )
  + t \int _ 0 ^ 1 P ( t ) \frac { \partial H } { \partial t } ( \rho t , x ( \rho t ) ) \, d \rho \nonumber
  \\ & \quad {} + t \int _ 0 ^ 1 P ( t ) \frac { \partial H } { \partial x } ( \rho t , x ( \rho t ) ) F ( \rho t , x ( \rho t ) , y ( \rho t ) ) \, d \rho \nonumber
  \\ & \quad {} - t \int _ 0 ^ 1 P ( t ) \frac { \partial N } { \partial t } ( \rho t , x ( \rho t ) ) y ( t ) \, d \rho \nonumber
  \\ & \quad {} - t \int _ 0 ^ 1 P ( t ) \frac { \partial N } { \partial x } ( \rho t , x ( \rho t ) ) F ( \rho t , x ( \rho t ) , y ( \rho t ) ) y ( t ) \, d \rho .
\end{align}
Using the fact that
\begin{gather}
  t \frac d { d t } P ( t )
  = P ( t ) \left ( t \frac d { d t } + N _ 0 \right )
\end{gather}
and then multiplying from the left by $ P ( 1 / t ) $ gives
\begin{align}
  \left ( t \frac d { d t } + N _ 0 \right ) ( y ( t ) - y _ 0 )
  & = t G ( t , x ( t ) , y ( t ) )
  + t \int _ 0 ^ 1 \frac { \partial H } { \partial t } ( \rho t , x ( \rho t ) ) \, d \rho \nonumber 
  \\ & \quad {} + t \int _ 0 ^ 1 \frac { \partial H } { \partial x } ( \rho t , x ( \rho t ) ) F ( \rho t , x ( \rho t ) , y ( \rho t ) ) \, d \rho
  - t \int _ 0 ^ 1 \frac { \partial N } { \partial t } ( \rho t , x ( \rho t ) ) y ( t ) \, d \rho \nonumber 
  \\ & \quad {} - t \int _ 0 ^ 1 \frac { \partial N } { \partial x } ( \rho t , x ( \rho t ) ) F ( \rho t , x ( \rho t ) , y ( \rho t ) ) y ( t ) \, d \rho .
\end{align}
Now
\begin{gather}
  t \int _ 0 ^ 1 \left ( \frac { \partial H } { \partial t } ( \rho s , x ( \rho s ) ) + \frac { \partial H } { \partial x } ( \rho s , x ( \rho s ) ) F ( \rho s , x ( \rho s ) , y ( \rho s ) ) \right ) \, d \rho
\end{gather}
can be rewritten as
\begin{gather}
  \int _ 0 ^ t \left ( \frac { \partial H } { \partial t } ( r , x ( r ) ) + \frac { \partial H } { \partial x } ( r , x ( r ) ) F ( r , x ( r ) , y ( r ) ) \right ) \, d r
\end{gather}
or
\begin{gather}
  \int _ 0 ^ t \left ( \frac { \partial H } { \partial t } ( r , x ( r ) ) + \frac { \partial H } { \partial x } ( r , x ( r ) ) x ' ( r ) \right ) \, d r
\end{gather}
or
\begin{gather}
  \int _ 0 ^ t \frac d { d r } H ( r , x ( r ) ) \, d r = H ( t , x ( t ) ) - H _ 0
\end{gather}
and similarly for
\begin{gather}
  t \int _ 0 ^ 1 \left ( \frac { \partial N } { \partial t } ( \rho s , x ( \rho s ) ) + \frac { \partial N } { \partial x } ( \rho s , x ( \rho s ) ) F ( \rho s , x ( \rho s ) , y ( \rho s ) ) \right ) \, d \rho .
\end{gather}
Also,
\begin{gather}
  \left ( t \frac d { d t } + N _ 0 \right ) ( y ( t ) - y _ 0 )
  = t y ' ( t ) + N _ 0 y ( t ) - N _ 0 y _ 0
  = t y ' ( t ) + N _ 0 y ( t ) - H _ 0
\end{gather}
by the compatibility condition~(\ref{compatibility}).
So
\begin{gather}
  t y ' ( t ) + N _ 0 y ( t ) - H _ 0
  = t G ( t , x ( t ) , y ( t ) )
  + H ( t , x ( t ) ) - H _ 0
  - ( N ( t , x ( y ) ) - N _ 0 ) y ( t )
\end{gather}
or, equivalently,
\begin{gather}
  t y ' ( t ) + N ( t , x ( t ) ) y ( t )
  = t G ( t , x ( t ) , y ( t ) )
  + H ( t , x ( t ) ) .
\end{gather}
So $ ( x , y ) $ is a fixed point of $ \Phi $ if and only if it is
a solution to the initial value problem~(\ref{IVP}).
The solutions we obtain in this way are solutions in $ [ 0 , \tau ] $,
where $ \tau $ is an arbitrary element of $ ( 0 , \tmin ] $, so
we get a solution on $ [ 0 , \tmin ] $. The flexibility afforded
by allowing $ \tau < \tmin $ is not needed here, but will be useful
in the continuity result which we will now prove.

To establish the continuity result we note that
just as the unique solution to~(\ref{IVP}) is a fixed point of~$ \Phi $,
the unique solution to~(\ref{IVP2}) is a fixed point of~$ \tilde \Phi $,
defined by
\begin{gather}
  \tilde \Phi ( \tilde x , \tilde y ) = ( \tilde u , \tilde v )
\end{gather}
where
\begin{align}
  \tilde u ( t )
&  = x _ 0 + t \int _ 0 ^ 1 \tilde F ( \sigma t , \tilde x ( \sigma t ) , \tilde y ( \sigma t ) ) \, d \sigma , \\
  \tilde v ( t )
  & = y _ 0 + t \int _ 0 ^ 1 \tilde P ( \sigma ) \tilde G ( \sigma t , \tilde x ( \sigma t ) , \tilde y ( \sigma t ) ) \, d \sigma
  + t \int _ 0 ^ 1 \int _ 0 ^ 1 \tilde P ( \sigma ) \frac { \partial \tilde H } { \partial t } ( \rho \sigma t , \tilde x ( \rho \sigma t ) ) \, d \rho \, d \sigma \nonumber
  \\ & \quad {} + t \int _ 0 ^ 1 \int _ 0 ^ 1 \tilde P ( \sigma ) \frac { \partial \tilde H } { \partial x } ( \rho \sigma t , \tilde x ( \rho \sigma t ) ) \tilde F ( \rho \sigma t , \tilde x ( \rho \sigma t ) , \tilde y ( \rho \sigma t ) ) \, d \rho \, d \sigma \nonumber
  \\ & \quad {} - t \int _ 0 ^ 1 \int _ 0 ^ 1 \tilde P ( \sigma ) \frac { \partial \tilde N } { \partial t } ( \rho \sigma t , \tilde x ( \rho \sigma t ) ) \tilde y ( \sigma t ) \, d \rho \, d \sigma \nonumber
  \\ & \quad {} - t \int _ 0 ^ 1 \int _ 0 ^ 1 \tilde P ( \sigma ) \frac { \partial \tilde N } { \partial x } ( \rho \sigma t , \tilde x ( \rho \sigma t ) ) \tilde F ( \rho \sigma t , \tilde x ( \rho \sigma t ) , \tilde y ( \rho \sigma t ) )  \tilde y ( \sigma t ) \, d \rho \, d \sigma .
\end{align}
Also, set
\begin{gather}
  ( \hat x , \hat y ) = \Phi ( \tilde x , \tilde y ) .
\end{gather}
Note that it is $ \Phi $ which appears on the right hand side, not $ \tilde \Phi $.
Then
\begin{gather}
  \tilde x ( t ) - \hat x ( t )
  =
  t \int _ 0 ^ 1 \left [
    \left ( \tilde F - F \right ) ( \sigma t , \tilde x ( \sigma t ) , \tilde y ( \sigma t ) )
  \right ] \, d \sigma ,
\end{gather}
so
\begin{gather}
  \| \tilde x ( t ) - \hat x ( t ) \|
  \le \left ( \max _ { K _ 3 } \| \tilde F - F \| \right ) t
  = L _ x t
\end{gather}
for $ t \in [ 0 , \tau ] $.
Similarly,
\begin{gather}
  \tilde y ( t ) - \hat y ( t )
\end{gather}
is the sum of twelve terms,
\begin{gather}
  t \int _ 0 ^ 1 \left ( \tilde P - P \right ) ( \sigma ) G ( \sigma t , \tilde x ( \sigma t ) , \tilde y ( \sigma t ) ) \, d \sigma , \\
  t \int _ 0 ^ 1 \tilde P ( \sigma ) \left ( \tilde G - G \right ) ( \sigma t , \tilde x ( \sigma t ) , \tilde y ( \sigma t ) )
  \, d \sigma , \\
  t \int _ 0 ^ 1 \int _ 0 ^ 1 \left ( \tilde P - P \right ) ( \sigma ) \frac { \partial H } { \partial t } ( \rho \sigma t , \tilde x ( \rho \sigma t ) ) \, d \rho \, d \sigma , \\
  t \int _ 0 ^ 1 \int _ 0 ^ 1 \tilde P ( \sigma ) \left ( \frac { \partial \tilde H } { \partial t } - \frac { \partial H } { \partial t } \right ) ( \rho \sigma t , \tilde x ( \rho \sigma t ) ) \, d \rho \, d \sigma , \\
  t \int _ 0 ^ 1 \int _ 0 ^ 1 \left ( \tilde P - P \right ) ( \sigma ) \frac { \partial H } { \partial x } ( \rho \sigma t , \tilde x ( \rho \sigma t ) ) F ( \rho \sigma t , \tilde x ( \rho \sigma t ) , \tilde y ( \rho \sigma t ) ) \, d \rho \, d \sigma , \\
  t \int _ 0 ^ 1 \int _ 0 ^ 1 \tilde P ( \sigma ) \left ( \frac { \partial \tilde H } { \partial x } - \frac { \partial H } { \partial x } \right ) ( \rho \sigma t , \tilde x ( \rho \sigma t ) ) F ( \rho \sigma t , \tilde x ( \rho \sigma t ) , \tilde y ( \rho \sigma t ) ) \, d \rho \, d \sigma , \\
  t \int _ 0 ^ 1 \int _ 0 ^ 1 \tilde P ( \sigma ) \frac { \partial \tilde H } { \partial x } ( \rho \sigma t , \tilde x ( \rho \sigma t ) ) \left ( \tilde F - F \right ) ( \rho \sigma t , \tilde x ( \rho \sigma t ) , \tilde y ( \rho \sigma t ) ) \, d \rho \, d \sigma , \\
  t \int _ 0 ^ 1 \int _ 0 ^ 1 \left ( P - \tilde P \right )  ( \sigma ) \frac { \partial N } { \partial t } ( \rho \sigma t , \tilde x ( \rho \sigma t ) ) \tilde y ( \sigma t ) \, d \rho \, d \sigma , \\
  t \int _ 0 ^ 1 \int _ 0 ^ 1 \tilde P ( \sigma ) \left ( \frac { \partial N } { \partial t } - \frac { \partial \tilde N } { \partial t } \right ) ( \rho \sigma t , \tilde x ( \rho \sigma t ) ) \tilde y ( \sigma t ) \, d \rho \, d \sigma , \\
  t \int _ 0 ^ 1 \int _ 0 ^ 1 \left ( P - \tilde P \right ) ( \sigma ) \frac { \partial N } { \partial x } ( \rho \sigma t , \tilde x ( \rho \sigma t ) ) F ( \rho \sigma t , \tilde x ( \rho \sigma t ) , \tilde y ( \rho \sigma t ) )  \tilde y ( \sigma t ) \, d \rho \, d \sigma , \\
  t \int _ 0 ^ 1 \int _ 0 ^ 1 \tilde P ( \sigma ) \left ( \frac { \partial N } { \partial x } - \frac { \partial \tilde N } { \partial x } \right ) ( \rho \sigma t , \tilde x ( \rho \sigma t ) ) F ( \rho \sigma t , \tilde x ( \rho \sigma t ) , \tilde y ( \rho \sigma t ) ) \tilde y ( \sigma t ) \, d \rho \, d \sigma
\end{gather}
and
\begin{gather}
  t \int _ 0 ^ 1 \int _ 0 ^ 1 \tilde P ( \sigma ) \frac { \partial \tilde N } { \partial x } ( \rho \sigma t , \tilde x ( \rho \sigma t ) ) \left ( F - \tilde F \right ) ( \rho \sigma t , \tilde x ( \rho \sigma t ) , \tilde y ( \rho \sigma t ) ) \tilde y ( \sigma t ) \, d \rho \, d \sigma .
\end{gather}
So
\begin{gather}
  \| \tilde y ( t ) - \hat y ( t ) \| \le L _ y t
\end{gather}
and hence
\begin{gather}
  d ( ( \tilde x , \tilde y ) , ( \hat x , \hat y ) ) \le L \tau
\end{gather}
Also, since $ \Phi $ is a contraction with contraction factor~$ 1 / 2 $,
\begin{gather}
  d ( \Phi ( x , y ) , \Phi ( \tilde x , \tilde y ) )
  \le \frac 1 2 d ( ( x , y ) , ( \tilde x , \tilde y ) ) .
\end{gather}
Since
$ \Phi ( x , y ) = ( x , y ) $ and
$ \Phi ( \tilde x , \tilde y ) = ( \hat x , \hat y ) $
we can rewrite the preceding equation as
\begin{gather}
  d ( ( x , y ) , ( \hat x , \hat y ) )
  \le \frac 1 2 d ( ( x , y ) , ( \tilde x , \tilde y ) ) .
\end{gather}
Thus
\begin{gather}
  d ( ( x , y ) , ( \tilde x , \tilde y ) )
  \le d ( ( x , y ) , ( \hat x , \hat y ) )
  + d ( ( \hat x , \hat y ) , ( \tilde x , \tilde y ) )
  \le \frac 1 2 d ( ( x , y ) , ( \tilde x , \tilde y ) ) 
  + L \tau
\end{gather}
and so
\begin{gather}
  d ( ( x , y ) , ( \tilde x , \tilde y ) ) \le 2 L \tau .
\end{gather}
In view of the definition of the metrix $ d $,
\begin{gather}
  \max \left (
    \left \| \tilde x ( \tau ) - x ( \tau ) \right \| ,
    \left \| \tilde y ( \tau ) - y ( \tau ) \right \|
  \right )
  \le 2 L \tau .
\end{gather}
This holds for all $ \tau \in ( 0 , \tmin ] $, which is our continuity result. This completes the proof of Theorem \ref{ODEthm}.
\end{proof}

\section*{Acknowledgements}
HL, JS and PT acknowledge the hospitality from ICMAT where this work was initiated and HL, EN and PT acknowledge the support from Trinity College Dublin where part of this work was done.

\bibliographystyle{abbrv}

\end{document}